%% file: main.tex
\documentclass[a4paper,UKenglish,cleveref, autoref, thm-restate]{lipics-v2021}

\hideLIPIcs  

\title{Tight Algorithmic Applications of Clique-Width Generalizations} 

\author{Vera Chekan}{Humboldt-Universität zu Berlin, Germany}{vera.chekan@informatik.hu-berlin.de}{https://orcid.org/0000-0002-6165-1566}{Supported by DFG Research Training Group 2434 “Facets of Complexity”.}

\author{Stefan Kratsch}{Humboldt-Universität zu Berlin, Germany}{kratsch@informatik.hu-berlin.de}{https://orcid.org/0000-0002-0193-7239}{}

\authorrunning{V.\ Chekan and S.\ Kratsch} 

\Copyright{Vera Chekan and Stefan Kratsch} 



\category{} 

\relatedversion{} 

\nolinenumbers 

\EventEditors{John Q. Open and Joan R. Access}
\EventNoEds{2}
\EventLongTitle{42nd Conference on Very Important Topics (CVIT 2016)}
\EventShortTitle{CVIT 2016}
\EventAcronym{CVIT}
\EventYear{2016}
\EventDate{December 24--27, 2016}
\EventLocation{Little Whinging, United Kingdom}
\EventLogo{}
\SeriesVolume{42}
\ArticleNo{23}

\usepackage{amsthm}
\usepackage{verbatim}
\usepackage{amsmath}
\usepackage{amssymb}
\usepackage{mathtools}
\usepackage{titlesec}
\usepackage{mathrsfs}

\usepackage{todonotes}
\usepackage{hyperref}

\usepackage{enumerate}
\usepackage[capitalise]{cleveref}
\usepackage{macros}

\begin{document}

\maketitle

\begin{abstract}
    In this work, we study two natural generalizations of clique-width introduced by Martin Fürer.
    Multi-clique-width~(mcw) allows every vertex to hold multiple labels~[ITCS 2017], while for fusion-width~(fw) we have a possibility to merge all vertices of a certain label~[LATIN 2014].
    Fürer has shown that both parameters are upper-bounded by treewidth thus making them more appealing from an algorithmic perspective than clique-width and asked for applications of these parameters for problem solving.
    First, we determine the relation between these two parameters by showing that 
    $\mcw \leq \fw + 1$.
    Then we show that when parameterized by multi-clique-width, many problems (e.g., \textsc{Connected Dominating Set}) admit algorithms with the same running time as for clique-width despite the exponential gap between these two parameters.
    For some problems (e.g., \textsc{Hamiltonian Cycle}) we show an analogous result for fusion-width:
    For this we present an alternative view on fusion-width by introducing so-called glue-expressions which might be interesting on their own.
    All algorithms obtained in this work are tight up to (Strong) Exponential Time Hypothesis.
\end{abstract}

    \input{Introduction}
    \input{Preliminaries}

    \input{Fw-vs-mcw}
    \input{glue-expressions}
    \input{fw-algorithms}

    \input{mcw-algorithms}
    \input{Conclusion}

    \bibliography{Bibliography}
\end{document}

%% file: Introduction.tex
\section{ Introduction}

In parameterized complexity apart from the input size we consider a so-called parameter and study the complexity of problems depending on both the input size and the parameter where the allowed dependency on the input size is polynomial. 
In a more fine-grained setting one is interested in the best possible dependency on the parameter under reasonable conjectures.
A broad line of research is devoted to so-called \emph{structural parameters} measuring how simple the graph structure is: different parameters quantify various notions of possibly useful input structure.
Probably the most prominent structural parameter is treewidth, which reflects how well a graph can be decomposed using small vertex separators. 
For a variety of problems, the tight complexity parameterized by treewidth (or its path-like analogue pathwidth) has been determined under the so-called Strong Exponential Time Hypothesis (e.g., \cite{RooijBR09,JansenLS14,Pilipczuk17,CyganKN18,LokshtanovMS18,CyganNPPRW22,FockeMINSSW23}).
However, the main drawback of treewidth is that it is only bounded in sparse graphs: a graph on $n$ vertices of treewidth $k$ has no more than $nk$ edges.

To capture the structure of dense graphs, several parameters have been introduced and considered.
One of the most studied is clique-width. 
The clique-width of a graph is at most~$k$ if it can be constructed using the following four operations on $k$-labeled graphs: create a vertex with some label from $1, \dots, k$; form a disjoint union of two already constructed graphs; give all vertices with label $i$ label $j$ instead; or create all edges between vertices with labels $i$ and $j$.
It is known that if a graph has treewidth $k$, then it has clique-width at most~$3 \cdot 2^{k-1}$ and it is also known that an exponential dependence in this bound is necessary~\cite{CorneilR05}.
Conversely, cliques have clique-width at most 2 and unbounded treewidth.
So on the one hand, clique-width is strictly more expressive than treewidth in the sense that if we can solve a problem efficiently on classes of graphs of bounded clique-width, then this is also true for classes of graphs of bounded treewidth. 
On the other hand, the exponential gap has the effect that as the price of solving the problem for larger graph classes we potentially obtain worse running times for some graph families.

Fürer introduced and studied two natural generalizations of clique-width, namely fusion-width (fw)~\cite{Furer14} and multi-clique-width (mcw)~\cite{Furer17}.
For fusion-width, additionally to the clique-width operations, he allows an operator that fuses (i.e., merges) all vertices of label~$i$.
Originally, fusion-width (under a different name) was introduced by Courcelle and Makowsky~\cite{CourcelleM02}.
However, they did not suggested studying it as a new width parameter since it is parametrically (i.e., up to some function) equivalent to clique-width.
For multi-clique-width, the operations remain roughly the same as for clique-width but now every vertex is allowed to have multiple labels. 
For these parameters, Fürer showed the following relations to clique-width (cw) and treewidth (tw):
\begin{align}\label{eq:relation-cw}
    \fw \leq \cw \leq \fw \cdot 2^{\fw} && \mcw \leq \cw \leq 2^{\mcw} && \fw \leq \tw + 2 && \mcw \leq \tw + 2
\end{align}
Fürer also observed that the exponential gaps between clique-width and both fusion- and multi-clique-width are necessary.
As our first result, we determine the relation between fusion-width and multi-clique-width: 
\begin{theorem}\label{thm:mcw-fw-relation-intro}
    For every graph $G$, it holds that $\mcw(G) \leq \fw(G) + 1$.
    Moreover, given a fuse-$k$-expression $\phi$ of $G$, a multi-clique-width-$(k + 1)$-expression of $G$ can be created in time polynomial in $|\phi|$ and $k$.
\end{theorem}

The relations in \eqref{eq:relation-cw} imply that a problem is FPT parameterized by fusion-width resp.\ multi-clique-width if and only if this is the case for clique-width.
However, the running times of such algorithms might strongly differ.
Fürer initiated a fine-grained study of problem complexities relative to multi-clique-width, starting with the \textsc{Independent Set} problem.
He showed that this problem can be solved in $\ostar(2^{\mcw})$ where $\ostar$ hides factors polynomial in the input size.
On the other hand, Lokshtanov et al.\ proved that under SETH no algorithm can solve this problem in $\ostar((2-\varepsilon)^{\pw})$ where $\pw$ denotes the parameter called pathwidth~\cite{LokshtanovMS18}.
Clique-width of a graph is at most its pathwidth plus two~\cite{FellowsRRS09} so the same lower bound holds for clique-width and hence, multi-clique-width as well.
Therefore, the tight dependence on both clique-width and multi-clique-width is the same, namely $\ostar(2^k)$. 
We show that this is the case for many further problems.
\begin{theorem}
    Let $G$ be a graph given together with a multi-$k$-expression of $G$. Then:
    \begin{itemize}
        \item \textsc{Dominating Set} can be solved in time $\ostar(4^k)$; 
        \item \textsc{$q$-Coloring} can be solved in time $\ostar((2^q - 2)^k)$;
        \item \textsc{Connected Vertex Cover} can be solved in time $\ostar(6^k)$;
        \item \textsc{Connected Dominating Set} can be solved in time $\ostar(5^k)$.
    \end{itemize}
    And these results are tight under SETH. 

    Further, \textsc{Chromatic Number} can be solved in time $f(k) \cdot n^{2^{\mathcal{O}(k)}}$ and this is tight under~ETH.
\end{theorem}
We prove this by providing algorithms for multi-clique-width with the same running time as the known tight algorithms for clique-width.
The lower bounds for clique-width known from the literature then apply to multi-clique-width as well proving the tightness of our results.
By \cref{thm:mcw-fw-relation-intro}, these results also apply to fusion-width.
For the following three problems we obtain similar tight bounds relative to fusion-width as for clique-width, but it remains open whether the same is true relative to multi-clique-width:
\begin{theorem}
    Let $G$ be a graph given together with a fuse-$k$-expression of $G$. Then:
    \begin{itemize}
        \item \textsc{Max Cut} can be solved in time $f(k) \cdot n^{\mathcal{O}(k)}$;
        \item \textsc{Edge Dominating Set} can be solved in time $f(k) \cdot n^{\mathcal{O}(k)}$;
        \item \textsc{Hamiltonian Cycle} can be solved in time $f(k) \cdot n^{\mathcal{O}(k)}$.
    \end{itemize}
    And these results are tight under ETH.
\end{theorem}
To prove these upper bounds, we provide an alternative view on fuse-expressions, called \emph{glue-expressions}, interesting on its own.
We show that a fuse-$k$-expression can be transformed into a glue-$k$-expression in polynomial time and then present dynamic-programming algorithms on glue-expressions.
Due to the exponential gap between clique-width and both fusion- and multi-clique-width, our results provide exponentially faster algorithms on graphs witnessing these gaps.

\subparagraph*{Related Work}
Two parameters related to both treewidth and clique-width are modular treewidth (mtw)~\cite{BodlaenderJ00,HegerfeldK23-mtw} and twinclass-treewidth~\cite{Mengel16,PaulusmaSS16,Lampis20} (unfortunately, sometimes also referred to as modular treewidth).
It is known that $\mcw \leq \operatorname{mtw} + 3$ (personal communication with Falko Hegerfeld).
Further dense parameters have been widely studied in the literature. 
Rank-width (rw) was introduced by Oum and Seymour and it reflects the $\FF_2$-rank of the adjacency matrices in the so-called branch decompositions.
Originally, it was defined to obtain a fixed-parameter approximation of clique-width~\cite{OumS06} by showing that $\operatorname{rw} \leq \cw \leq 2^{\operatorname{rw} + 1} - 1$.
Later, Bui-Xuan et al.\ started the study of algorithmic properties of rank-width~\cite{Bui-XuanTV10}. 
Recently, Bergougnoux et al.\ proved the tightness of first ETH-tight lower bounds for this parameterization~\cite{BergougnouxKN23}.
Another parameter defined via branch-decompositions and reflecting the number of different neighborhoods across certain cuts is boolean-width (boolw), introduced by Bui-Xuan et al.~\cite{Bui-XuanTV11,Bui-XuanTV13}.
Fürer~\cite{Furer17} showed that $\operatorname{boolw} \leq \mcw \leq 2^{\operatorname{boolw}}$.
Recently, Eiben et~al.\ presented a framework unifying the definitions and algorithms for computation of many graph parameters~\cite{EibenGHJK22}.

\subparagraph*{Organization}
We start with some required definitions and notations in \cref{sec:preliminaries}.
In \cref{app:sec:fw-vs-mcw} we prove the relation between fusion-width and multi-clique-width from \cref{thm:mcw-fw-relation-intro}.
After that, in \cref{app:sec:useful-glue-expressions} we introduce glue-$k$-expressions and show how to obtain such an expression given a fuse-$k$-expression of a graph. 
Then in \cref{app:sec:fw-algorithms} we employ these expressions to obtain algorithms parameterized by fusion-width.
In \cref{app:sec:mcw-algorithms} we present algorithms parameterized by multi-clique-width.
We conclude with some open questions in \cref{sec:conclusion}.

%% file: Preliminaries.tex
\section{ Preliminaries}\label{sec:preliminaries}
For $k \in \NN_0$, we denote by $[k]$ the set $\{1, \dots, k\}$ and we denote by $[k]_0$ the set $[k] \cup \{0\}$.

We use standard graph-theoretic notation.
Our graphs are simple and undirected if not explicitly stated otherwise. 
For a graph $H$ and a partition $(V_1, V_2)$ of $V(H)$, by $E_H(V_1, V_2) = \{ \{v_1, v_2\} \mid v_1 \in V_1, v_2 \in V_2\}$ we denote the set of edges between $V_1$ and $V_2$.
For a set $S$ of edges in a graph $H$, by $V(S)$ we denote the set of vertices incident with the edges in $S$.

A $k$-labeled graph is a pair $(H, \lab_H)$ where $\lab_H \colon V(H) \to [k]$ is a \emph{labeling function} of~$H$. 
Sometimes to simplify the notation in our proofs we will allow the labeling function to map to some set of cardinality $k$ instead of the set $[k]$.
In the following, if the number $k$ of labels does not matter, or it is clear from the context, we omit~$k$ from the notions (e.g., a labeled graph instead of a $k$-labeled graph).
Also, if the labeling function is clear from the context, then we simply call $H$ a labeled graph as well. 
Also we sometimes omit the subscript $H$ of the labeling function $\lab_H$ for simplicity.
For~$i \in [k]$, by $U^H_i = \lab_H^{-1}(i)$ we denote the set of vertices of $H$ with label $i$.
We consider the following four operations on $k$-labeled graphs.
\begin{enumerate}
    \item \emph{Introduce}: For $i \in [k]$, the operator $v\langle i \rangle$ creates a graph with a single vertex $v$ that has label $i$. We call $v$ the \emph{title} of the vertex.
    \item \emph{Union}: The operator $\oplus$ takes two vertex-disjoint $k$-labeled graphs and creates their disjoint union. The labels are preserved.
    \item \emph{Join}: For $i \neq j \in [k]$, the operator $\eta_{i, j}$ takes a $k$-labeled graph $H$ and creates the supergraph $H'$ on the same vertex set with $E(H') = E(H) \cup \{\{u, v\} \mid \lab_H(u) =  i, \lab_H(v) =  j\}$. 
    The labels are preserved.
    \item \emph{Relabel}: For $i \neq j$, the operator $\rho_{i \to j}$ takes a $k$-labeled graph $H$ and creates the same $k$-labeled graph $H'$ apart from the fact that every vertex that with label $i$ in $H$ instead has label $j$ in $H'$.
\end{enumerate}
A well-formed sequence of such operations is called a \emph{$k$-expression} or a \emph{clique-expression}. 
With a $k$-expression $\phi$ one can associate a rooted tree such that every node corresponds to an operator, this tree is called a \emph{parse tree} of $\phi$.
With a slight abuse of notation, we denote it by $\phi$ as well.
By $G^\phi$ we denote the labeled graph arising in $\phi$.
And for a node~$t$ of $\phi$ by $G^\phi_t$ we denote the labeled graph arising in the subtree (sometimes also called a \emph{sub-expression}) rooted at $t$, this subtree is denoted by $\phi_t$.
The graph $G^\phi_t$ is then a subgraph of $G^\phi$.
A graph~$H$ has \emph{clique-width} of at most $k$ if there is a labeling function $\lab_H$ of $H$ and a $k$-expression $\phi$ such that $G^\phi$ is equal to $(H, \lab_H)$.
By $\cw(H)$ we denote the smallest integer $k$ such that $H$ has clique-width at most $k$.
Fürer has studied two generalizations of $k$-expressions~\cite{Furer14,Furer17}.

\emph{Fuse}: For $i \in [k]$, the operator $\theta_i$ takes a $k$-labeled graph $H$ with $\lab^{-1}_H(i) \neq \emptyset$ and \emph{fuses} the vertices with label $i$, i.e., the arising graph $H'$ has vertex set $(V(H) - \lab^{-1}_H(i)) \dot\cup \{v\}$, the edge relation in $V(H) - \lab^{-1}_H(i)$ is preserved, and $N_{H'}(v) = N_{H}(\lab^{-1}_H(i))$.
The labels of vertices in $V(H') - v$ are preserved, and vertex $v$ has label $i$.
A \emph{fuse-$k$-expression} is a well-formed expression that additionally to the above four operations is allowed to use fuses.
We adopt the above notations from $k$-expressions to fuse-$k$-expressions. 
Let us only remark that for a node $t$ of a fuse-$k$-expression $\phi$, the graph $G^\phi_t$ is not necessarily a subgraph of $G^\phi$ since some vertices of $G^\phi_t$ might be fused later in $\phi$.
\begin{remark}\label{remark:fw-larger-introduce}
    Originally, Fürer allows that a single introduce-node creates multiple, say~$q$, vertices with the same label. 
    However, we can eliminate such operations from a fuse-expression~$\phi$ as follows. 
    If the vertices introduced at some node participate in some fuse later in the expression, then it suffices to introduce only one of them. 
    Otherwise, we can replace this introduce-node by $q$ nodes introducing single vertices combined using union-nodes.
    These vertices are then also the vertices of $G^\phi$.
    So in total, replacing all such introduce-nodes would increase the number of nodes of the parse tree by at most $\mathcal{O}(|V(G^\phi)|)$, which is not a problem for our algorithmic applications.  
\end{remark}

Another generalization of clique-width introduced by Fürer is multi-clique-width (mcw)~\cite{Furer17}. 
A multi-$k$-labeled graph is a pair $(H, \lab_H)$ where $\lab_H \colon V(H) \to 2^{[k]}$ is a multi-labeling function.
We consider the following four operations of multi-$k$-labeled graphs.
\begin{enumerate}
    \item \emph{Introduce}: For $q \in [k]$ and $i_1, \dots i_q \in [k]$, the operator $v \langle i_1, \dots, i_q \rangle$ creates a multi-$k$-labeled graph with a single vertex that has label set $\{i_1, \dots, i_q\}$.
    \item \emph{Union}: The operator $\oplus$ takes two vertex-disjoint multi-$k$-labeled graphs and creates their disjoint union. The labels are preserved.
    \item \emph{Join}: For $i \neq j \in [k]$, the operator $\eta_{i, j}$ takes a multi-$k$-labeled graph $H$ and creates its supergraph $H'$ on the same vertex set with $E(H') = E(H) \cup \{\{u, v\} \mid i \in \lab_H(u), j \in \lab_H(v)\}$.
    This operation is only allowed when there is no vertex in $H$ with labels $i$ and~$j$ simultaneously, i.e., for every vertex $v$ of $H$ we have $\{i, j\} \not\subseteq \lab_H(v)$.
    The labels are preserved.
    \item \emph{Relabel}: For $i \in [k]$ and $S \subseteq [k]$, the operator $\rho_{i \to S}$ takes a multi-$k$-labeled graph $H$ and creates the same multi-labeled graph apart from the fact that every vertex with label set~$L \subseteq [k]$ such that $i \in L$ in $H$ instead has label set $(L \setminus \{i\}) \cup S$ in $H'$.
    Note that~$S = \emptyset$ is allowed.
\end{enumerate}
A well-formed sequence of these four operations is called a \emph{multi-$k$-expression}. 
As for fuse-expressions, Fürer allows introduce-nodes to create multiple vertices but we can eliminate this by increasing the number of nodes in the expression by at most $\mathcal{O}(|V(G^\phi)|)$.
We adopt the analogous notations from $k$-expressions to multi-$k$-expressions. 

\subparagraph*{Complexity}
To the best of our knowledge, the only known way to approximate multi-clique-width and fusion-width is via clique-width, i.e., to employ the relation \eqref{eq:relation-cw}.
The only known way to approximate clique-width is, in turn, via rank-width.
This way we obtain a~$2^{2^k}$-approximation of multi-clique-width and fusion-width running in FPT time. 
For this reason, to obtain tight running times in our algorithms we always assume that a fuse- or multi-$k$-expression is provided.
Let us emphasize that this is also the case for all tight results for clique-width in the literature~(see e.g., \cite{BergougnouxKK20,Lampis20}).
In this work, we will show that if a graph admits a multi-$k$-expression resp.\ a fuse-$k$-expression, then it also admits one whose size is polynomial in the size of the graph.
Moreover, such a ``compression'' can be carried out in time polynomial in the size of the original expression.
Therefore, we delegate this compression to a black-box algorithm computing or approximating multi-clique-width or fusion-width and assume that provided expressions have size polynomial in the graph size.

\subparagraph{(Strong) Exponential Time Hypothesis}
The algorithms in this work are tight under one of the following conjectures formulated by Impagliazzo et al.~\cite{ImpagliazzoPZ01}. 
The Exponential Time Hypothesis (ETH) states that there is $0 < \varepsilon < 1$ such that \textsc{3-Sat} with $n$ variables and $m$ clauses cannot be solved in time $\ostar(2^{\varepsilon n})$.
The Strong Exponential Time Hypothesis (SETH) states that for every $0 < \varepsilon < 1$ there is an integer $q$ such that \textsc{$q$-Sat} cannot be solved in time~$\ostar(2^{\varepsilon n})$.
In this work, $\ostar$ hides factors polynomial in the input size.

\subparagraph{Simplifications} 
If the graph is clear from the context, by $n$ we denote the number of its vertices.
If not stated otherwise, the number of labels is denoted by $k$ and a label is a number from $[k]$.

%% file: Fw-vs-mcw.tex
\section{ Relation Between Fusion-Width and Multi-Clique-Width}\label{app:sec:fw-vs-mcw}

In this section, we show that for every graph, its multi-clique-width is at most as large as its fusion-width plus one.
Since we are interested in parameterized complexity of problems, the constant additive term to the value of a parameter does not matter.
To prove the statement, we show how to transform a fuse-$k$-expression of a graph $H$ into a multi-$(k+1)$-expression of $H$.
Fürer has proven the following relation:
\begin{lemma}[\cite{Furer17}] \label{app:lem:furer-fw-cw}
    For every graph $H$, it holds that $\cw(H) \leq \fw(H) \cdot 2^{\fw(H)}$.
\end{lemma}
We will use his idea behind the proof of this lemma to prove our result.

\begin{theorem}\label{app:thm:mcw-fw-relation}
    For every graph $H$, it holds that $\mcw(H) \leq \fw(H) + 1$.
    Moreover, given a fuse-$k$-expression $\phi$ of $H$, a multi-$(k + 1)$-expression of $H$ can be created in time polynomial in $|\phi|$ and $k$.
\end{theorem}

\begin{proof}
    Let $H$ be a graph.
    We start by showing that $\mcw(H) \leq 2 \cdot \fw(H)$ holds. 
    To prove this, we will consider a fuse-$k$-expression of $H$ and from it, we will construct a multi-$2k$-expression of $H$ using labels $\{1, \dots, k, \widehat{1}, \dots, \widehat{k}\}$.
    For simplicity of notation, let $\widehat{[k]} = \{\widehat{1}, \dots, \widehat{k}\}$.
    For this first step, we strongly follow the construction of Fürer in his proof of \cref{app:lem:furer-fw-cw}.
    There he uses $k \cdot 2^k$ labels from the set $[k] \times 2^{[k]}$ so the second component of such a label is a subset of $[k]$.
    We will use that multi-clique-width perspective already allows vertices to have sets of labels and model the second component of a label via subsets of $\widehat{[k]}$.
    Then we will make an observation that allows us to (almost) unify labels $i$ and $\widehat{i}$ for every $i \in [k]$. 
    Using one additional label $\star$, we will then obtain a multi-$(k+1)$-expression of $H$ using labels $[k] \cup \{\star\}$.

    First of all, we perform several simple transformations on $\phi$ without changing the arising graph.
    We suppress all join-nodes that do not create new edges, i.e., we suppress a join-node $t$ if for its child $t'$ it holds $G_t = G_{t'}$.
    Then we suppress all nodes fusing less than two vertices, i.e., a $\theta_i$-node $t$ for some $i \in [k]$ is suppressed if for its child $t'$, the labeled graph $G^\phi_{t'}$ contains less than two vertices with label $i$.
    Now we provide a short intuition for the upcoming transformation.
    Let $x$ be a $\theta_i$-node creating a new vertex, say $u$, by fusing some vertices, say $U$. 
    And let $y$ be an ancestor of $x$ such that $y$ is a fuse-node that fuses vertex $u$ with some further vertices, say $W$.
    Then we can safely suppress the node $x$: the fuse of vertices from $U$ is then simply postponed to $y$, where these vertices are fused together with $W$.
    Now we fix some notation used in the rest of the proof.
    Let $x$ be a node, let $y$ be an ancestor of $x$, and let $t_1, \dots, t_q$ be all inner relabel-nodes on the path from $x$ to $y$ in the order they occur on this path.
    Further, let $s_1, \dots, s_q \in [k]$ and $r_1, \dots, r_q \in [k]$ be such that the node $t_j$ is a $\rho_{s_j \to r_j}$-node for every $j \in [q]$. 
    Then for all $i \in [k]$, we define
    \[
        \rho^*_{x,y}(i) = \sigma_q ( \sigma_{q-1} ( \dots \sigma_1(i) ) ) 
    \]
    where
    \[
        \sigma_j(i') =
        \begin{cases}
            i' & \text{if } i' \neq s_j  \\
            r_j & \text{if } i' = s_j \\ 
        \end{cases}
    \]
    for $j \in [q]$.
    Intuitively, if we have some vertex $v$ of label $i$ in $G^\phi_x$, then $\rho^*_{x, y}(i)$ denotes the label of $v$ in $G^\phi_{y'}$ where $y'$ denotes the child of $y$, i.e., $\rho^*_{x, y}(i)$ is the label of $v$ right before the application of $y$.
    Now for every $i \in [k]$ and every $\theta_i$-node $x$, if there exists an ancestor $y$ of $x$ in $\phi$ such that $y$ is a $\theta_{\rho^*_{x,y}(i)}$-node, we suppress the node $x$.
    In this case, we call $x$ \emph{skippable}.
    Finally, we transform the expression in such a way that a parent of every leaf is a union-node as follows.
    Let $x$ be a leaf with introducing a vertex $v$ of label $i$ for some $i \in [k]$.
    As a result of the previous transformations, we know that the parent $y$ of $x$ is either a relabel- or a union-node. 
    In the latter case, we skip this node.
    Otherwise, let $i_1 \neq i_2 \in [k]$ be such that $y$ is a $\rho_{i_1 \to i_2}$-node.
    If $i_1 \neq i$, then we suppress $y$.
    Otherwise, we suppress $y$ and replace $x$ with a node introducing the same vertex but with label $i_2$.
    This process is repeated for every leaf.
    We denote the arising fuse-$k$-expression of $H$ by $\psi$.

    Now let $x$ be a node of $\psi$ and let $v$ be a vertex of $G^\psi_x$.
    We say that $v$ is a \emph{fuse-vertex} at $x$ if $v$ \emph{participates} in some fuse-operation above $x$, that is, there is an ancestor $y$ of $x$ (in $\psi$) such that $y$ is a $\theta_{\rho^*_{x, y}(i)}$-node.
    Note that first, since we have removed skippable fuse-nodes, if such a node $y$ exists for $x$, then it is unique.
    And second, in this case all vertices of label $i$ in $G^\psi_x$ will participate in the fuse-operation.
    So we also say $i$ is a \emph{fuse-label} at $x$.
    Hence, instead of first, creating these vertices via introduce-nodes and then fusing them, we will introduce only one vertex representing the result of the fusion.
    And the creation of the edges incident with these vertices needs to be postponed until the moment where the new vertex is introduced. 
    For this, we will store the label of the new vertex in the label set of the other end-vertex.
    But for postpone-purposes we will use labels from $\widehat{[k]}$ to distinguish from the original labels. 
    
    We now formalize this idea to obtain a multi-$2k$-expression $\xi$ of $H$.
    In the following, the constructed expression will temporarily contain at the same time vertices with multiple labels and fuse-nodes, we call such an expression \emph{mixed}.
    First, we mark all fuse-nodes in $\psi$ as unprocessed and start with $\xi := \psi$. 
    We proceed for every leaf $\ell$ of $\psi$ as follows.
    Let $v$ and $i \in [k]$ be such that $\ell$ is a $v\langle i \rangle$-node.
    If $v$ is not a fuse-vertex at $\ell$ in $\psi$, we simply change the operation at $\ell$ in $\xi$ to be $1 \langle \{i\} \rangle$.
    Otherwise, let $x$ be the fuse-node in $\psi$ in which $v$ participates.
    Note that since we have suppressed skippable fuse-nodes, such a node $x$ is unique.
    Let $i \in [k]$ be such that $x$ is a $\theta_i$-node. 
    First, we remove the leaf $\ell$ from $\xi$ and suppress its parent in $\xi$.
    Note that since the parent of $\ell$ is $\psi$ is a union-node, the mixed expression remains well-formed.
    Second, if $x$ is marked as unprocessed, we replace the operation at $x$ in $\xi$ to be a union, add a new $1 \langle \{i\} \rangle$-node as a child of $x$, and mark $x$ as processed.
    We refer to the introduce-nodes created in this process as well as to the vertices introduced by these nodes as \emph{new}.
    Observe that first, the arising mixed expression does not contain any fuse-nodes.
    Second, the set of leaves of $\xi$ is now in bijection with the set of vertices of $H$.
    Also, the set of edges induced by vertices, that do not participate in any fuse-operation in $\psi$, has not been affected.
    So it remains to correctly create the edges for which at least one end-point is new.
    This will be handled by adapting the label sets of vertices.

    First, for every $i \neq j \in [k]$, every $\rho_{i \to j}$-node is replaced with a path consisting of a $\rho_{i \to \{j\}}$-node and a $\rho_{\widehat{i} \to \{\widehat{j}\}}$-node.
    Now let $i \neq j \in [k]$ and let $x$ be a $\eta_{i, j}$-node in $\xi$.
    In order to correctly handle the join-operation, we make a case distinction.
    If both $i$ and $j$ are not fuse-labels at $x$ in $\psi$, we skip $x$.
    Next, assume that exactly one of the labels $i$ and $j$, say $i$, is a fuse-label at $x$ in $\psi$.
    Then we replace the operation in $x$ in $\xi$ with $\rho_{j \to \{j, \widehat{i}\}}$ to store the information about the postponed edges in the vertices of label $j$.
    From now on, we may assume that both $i$ and $j$ are fuse-labels at $x$ in $\psi$.
    Observe that $x$ creates only one edge of $H$ since all vertices of label $i$ (resp.\ $j$) are fused into one vertex later in $\psi$.
    Let $x_i$ (resp.\ $x_j$) be the ancestor of $x$ in $\psi$ such that $x_i$ (resp.\ $x_j$) is a $\theta_{p_i}$-node (resp.\ $\theta_{p_j}$-node) where $p_i = \rho^*_{x, x_i}(i)$ (resp.\ $p_j = \rho^*_{x, x_j}(j)$).
    Since we have suppressed skippable fuse-nodes, the nodes $x_i$ and $x_j$ are unique.
    By our construction, $x_i$ (resp.\ $x_j$) is in $\xi$ a union-node that has a child $y_i$ (resp.\ $y_j$) being an introduce-node.     
    Without loss of generality, we may assume that $x_i$ is above $x_j$ in $\xi$.
    Then, we store the information about the postponed edge in $y_j$ as follows.
    Let $S \subseteq [k] \cup \widehat{[k]}$ be the label set such that $y_j$ is currently a $1 \langle S \rangle$-node.
    Note: initially $S$ consists of a single label $p_j$ but after processing several join-nodes, this is not the case in general.
    We now replace the operation in $y_j$ with $1 \langle S \cup \{\widehat{\rho^*_{x,x_j}(i)}\} \rangle$.
    After all join-nodes are processed, we create the postponed edges at every new introduce-node $x$ of $\xi$ as follows.
    Let $y$ be the parent of $x$ in $\xi$ and let $S \subseteq [k]$ be such that $x$ is an $1 \langle S \rangle$-node.
    By construction, there exists a unique label $i \in [k] \cap S$.
    Then right above $y$, we add the sequence $\rho_{\widehat{i} \to \emptyset} \circ \eta_{i, \widehat{i}}$ and we refer to this sequence together with $y$ as the \emph{postponed sequence} of $x$. 
           
    This concludes the construction of a multi-$2k$-expression, say $\alpha$, of $H$.
    It can be verified that we have not changed the construction of Fürer~\cite{Furer17} but only stated it in terms of multi-clique-width.
    Therefore, the construction is correct.
    
    Now as promised, we argue that the number of required labels can be decreased to $k + 1$.
    Before formally stating this, we provide an intuition.
    First, observe that moving from $\xi$ to $\alpha$, we did not change the unique label from $[k]$ kept by each vertex at any step, only the labels from $\widehat{[k]}$ have been affected. 
    We claim that for $i \in [k]$, both labels $i$ and $\widehat i$ may appear in a subgraph $G^\alpha_y$ only in very restricted cases, namely when $y$ belongs to a postponed sequence of a new introduce-node.
    We now sketch why this is the case.
    Let $x$ be a node such that $G^\alpha_x$ contains a vertex with label $\widehat i$.
    This can only occur if $i$ is a fuse-label at $x$ in $\psi$, i.e., there exists a unique fuse-node $z$ such that $z$ is an ancestor of $x$ and the vertices from $G^\psi_x$ of label $i$ participate in the fuse at $z$.
    By the construction of $\xi$, all introduce-nodes creating these vertices have been removed so $G^\alpha_z$ contains a unique vertex holding the label $j := \rho^*_{x, z}(i)$, namely the one introduced at its child, say $t$. 
    Then in the end of the postponed sequence of $t$, the label $\widehat{j}$ is removed from all vertices.
    So the only moment where both labels $j$ and $\widehat{j}$ occur is during the postponed sequence of $t$.
    Also note that postponed sequences do not overlap so if such $j$ exists, then it is unique.
    This is formalized as follows.
    \begin{observation}    
        Let $y$ be a node in $\alpha$ and let $i \in [k]$ be such that the labeled graph $G^{\alpha}_y$ contains a vertex containing label $i$ and a vertex containing label $\widehat{i}$. 
        Then $y$ belongs to the postponed sequence of some new $1 \langle S \rangle$-node $x$ with $S \subseteq [k] \cup \widehat{[k]}$ and $i \in S$.
        Moreover, the only vertex in $G^{\alpha}_y$ containing label $i$ is the vertex introduced at $x$.
        In particular, since the postponed sequences for distinct nodes are disjoint by construction, for every $j \neq i \in [k]$, the graph $G^{\alpha}_y$ does not contain a vertex containing label $j$ or it does not contain a vertex containing label $\widehat{j}$.
    \end{observation}

    So up to postponed sequences, we can unify the labels $i$ and $\widehat{i}$ for every $i \in [k]$.
    And inside postponed sequences, we will use an additional label $\star$ to distinguish between $i$ and $\widehat{i}$.
    So we process new introduce-nodes as follows.
    Let $x$ be a new $1 \langle S \rangle$-node for some $S \subseteq [k] \cup \widehat{[k]}$ and let $i \in [k]$ be the unique value in $S \setminus \widehat{[k]}$.
    We replace the operation in $x$ with a $1 \langle S \setminus \{i\} \cup \{\star\} \rangle$ and we replace the postponed sequence of $x$ with the sequence $\rho_{\star \to i} \circ \rho_{\widehat i \to \emptyset} \circ \eta_{\widehat{i}, \star} \circ \oplus$.
    After processing all new introduce-nodes, we replace every occurrence of label $\widehat i$ with label $i$ for all $i \in [k]$.
    The new multi-expression uses $k+1$ labels and by the above observation, it still creates $H$.
    Also it can be easily seen that the whole transformation can be carried out in polynomial time.
\end{proof}

%% file: glue-expressions.tex
\section{ Reduced Glue-Expressions}\label{app:sec:useful-glue-expressions}

In this section, we show that a fuse-$k$-expression can be transformed into a so-called \emph{reduced glue-$k$-expression} of the same graph in polynomial time. 
Such expressions will provide an alternative view on fusion-width helpful for algorithms.
We formally define them later.
In the following, we assume that the titles used in introduce-nodes of a fuse-expression are pairwise distinct.
Along this section, the number of labels is denoted by $k$ and \emph{polynomial} is a shorthand for ``polynomial in the size of the expression and $k$''.

To avoid edge cases, we will assume that any expression in this section does not contain any \emph{useless} nodes in the following sense.
If a join-node does not create new edges, it is suppressed.
Similarly, if a fuse-node fuses at most one node, it is suppressed.
Also during our construction, the nodes of form $\rho_{i \to i}$ might arise, they are also suppressed.
Further, if $\rho_{i \to j}$ is applied to a labeled graph with no vertices of label $i$, it is suppressed.
Clearly, useless nodes can be found and suppressed in polynomial time.
For this reason, from now on we always implicitly assume that useless nodes are not present.

We say that fuse-expressions $\phi_1$ and $\phi_2$ are \emph{equivalent} if there exists a label-preserving isomorphism between $G^{\phi_1}$ and $G^{\phi_2}$.
In this section, we provide rules allowing to replace sub-expressions with equivalent ones.
For simplicity, the arising expression will often be denoted by the same symbol as the original one.

The following equivalencies can be verified straight-forwardly. 
Although some of them might seem to be unnatural to use at first sight, they will be applied in the proofs of \cref{lem:shift-to-union,lem:fuse-exactly-two-from-different-sides}.
\begin{lemma}\label{lem:transformation-rules}
    Let $k \in \NN$, let $H$ be a $k$-labeled graph, let $q \in \NN$, let $i, j, a, b, a_1, \dots, a_q \in [k]$ be integers, let $H_1, H_2$ be $k$-labeled graphs, and let $v$ be a title.
    Then the following holds if none of the operators on the left-hand side is useless:
	\begin{enumerate}
		\item $\theta_i \circ \eta_{a, b} (H) = \eta_{a, b} \circ \theta_i(H)$; \label{rule:join}
		\item If $i \notin \{a, b\}$, then $\theta_i \circ \rho_{a \to b}(H) = \rho_{a \to b} \circ \theta_i(H)$; \label{rule:relabel}
		\item If $a, b \in \{a_1, \dots, a_q, i\}$, then: 
		\[
			\theta_i \circ \rho_{a_1 \to i}\circ \dots \circ \rho_{a_q \to i} \circ \eta_{a, b}(H) = \theta_i \circ \rho_{a_1 \to i}\circ \dots \circ \rho_{a_q \to i}(H);
		\]
        \label{rule:relabel-join-both-sides-fused}
		\item If $a \in \{a_1, \dots, a_q, i\}$ and $b \notin \{a_1, \dots, a_q, i\}$, then: 
		\[
			\theta_i \circ \rho_{a_1 \to i} \circ \dots \circ \rho_{a_q \to i}\circ \eta_{a, b}(H) = \eta_{i, b} \circ \theta_i \circ \rho_{a_1 \to i} \circ \dots \circ \rho_{a_q \to i}(H);
		\]
        \label{rule:relabel-join-one-side-fused}
		\item If $a, b \notin \{a_1, \dots, a_q, i\}$, then: 
		\[
			\theta_i \circ \rho_{a_1 \to i} \circ \dots \circ \rho_{a_q \to i} \circ \eta_{a, b}(H) = \eta_{a, b} \circ \theta_i \circ \rho_{a_1 \to i} \circ \dots \circ \rho_{a_q \to i}(H);
		\]
        \label{rule:relabel-join-no-side-fused}
		\item If $a, b \notin \{a_1, \dots, a_q, i\}$, then 
		\[
			\theta_i \circ \rho_{a_1 \to i} \circ \dots \circ \rho_{a_q \to i} \circ \rho_{a, b}(H) = \rho_{a, b} \circ \theta_i \circ \rho_{a_1 \to i} \circ \dots \circ \rho_{a_q \to i}(H);
		\]
        \label{rule:relabel-relabel}
		\item If $b \in \{a_1, \dots, a_q\}$, then
		\[
			\theta_i \circ \rho_{a_1 \to i} \circ \dots \circ \rho_{a_q \to i} \circ \rho_{a \to b}(H) = \theta_i \circ \rho_{a_1 \to i} \circ \dots \circ \rho_{a_q \to i} \circ \rho_{a \to i}(H);
		\]
        \label{rule:relabel-to-i-directly}
		\item If $b \notin \{a_1, \dots, a_q, i\}$, then 
		\[
			\theta_i \circ \rho_{a_1 \to i} \circ \dots \circ \rho_{a_q \to i} \circ \rho_{i \to b}(H) = \rho_{a_1 \to i} \circ \theta_{a_1} \circ \rho_{a_2 \to a_1} \circ \dots \circ \rho_{a_q \to a_1} \circ \rho_{i \to b}(H);
		\]
        \label{rule:relabel-from-i}
        \item 
		\[
			\theta_i \circ \rho_{a_1 \to i} \circ \dots \circ \rho_{a_q \to i} (H_1 \oplus H_2) = \theta_i \Bigl(\bigl(\rho_{a_1 \to i} \circ \dots \circ \rho_{a_q \to i}(H_1)\bigr) \oplus \bigl(\rho_{a_1 \to i} \circ \dots \circ \rho_{a_q \to i}(H_2)\bigr) \Bigr);
		\]
        \label{rule:relabel-union}
		\item If $b \in \{a_1, \dots, a_q, i\}$, then:
		\[
			\theta_i \circ \rho_{a_1 \to i} \circ \dots \circ \rho_{a_q \to i} \circ \theta_b(H) = \theta_i \circ \rho_{a_1 \to i} \circ \dots \circ \rho_{a_q \to i} (H);
		\]
        \label{rule:relabel-fuse-the-same}
		\item If $b \notin \{a_1, \dots, a_q, i\}$, then:
		\[
			\theta_i \circ \rho_{a_1 \to i} \circ \dots \circ \rho_{a_q \to i} \circ \theta_b(H) = \theta_b \circ \theta_i \circ \rho_{a_1 \to i} \circ \dots \circ \rho_{a_q \to i} (H);
		\]
		\label{rule:relabel-fuse-another}
        \item $\rho_{a \to j} \circ v \langle a \rangle = v \langle j \rangle$; \label{rule:relabel-introduce-1}
        \item $\eta_{i, j} \circ \rho_{a \to i} (H) = \rho_{a \to i} \circ \eta_{i, j} \circ \eta_{a, j} (H)$; \label{rule:join-relabel-2}
        \item If $a, b \notin \{i, j\}$, then:
        $\eta_{i, j} \circ \rho_{a \to b} (H) = \rho_{a \to b} \circ \eta_{i,j} (H)$. \label{rule:join-relabel-3}
	\end{enumerate}
\end{lemma}

We fix some notation. 
Let $t$ be a fuse-node in some fuse-expression $\phi$.
Since $t$ is not useless, there is at least one successor of $t$ being a union-node. 
The union-nodes are the only nodes with more than one child so there exists a unique topmost successor of $t$ being a union-node, we denote it by $t_\oplus$.
The children of $t_\oplus$ are denoted by $t_1$ and $t_2$.
For a node $t$, we call the maximum number of union-nodes on a path from $t$ to any leaf in the subtree rooted at $t$ the \emph{$\oplus$-height} of $t$.

Informally speaking, a fuse-expression we want to achieve in this section has the following two properties. 
First, for any pair of distinct vertices that are fused at some point, their fuse happens ``as early as possible''. 
Namely, two vertices are fused right after the earliest union these vertices participate in together: in particular, these vertices come from different sides of the union.
This will allow us to replace a sequence of fuse-nodes by a so-called glue-node that carries out non-disjoint union of two graphs under a certain restriction.
Second, we want that each edge of the graph is created exactly once. 
We split the transformation into several steps. 
In the very first step we shift every fuse-node to the closest union-node below it (see \cref{fig:glue-expressions-a}.

\begin{figure}[b]
    \includegraphics{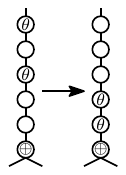}
    \centering
    \caption{Shifting fuse-nodes to union-nodes.} 
    \label{fig:glue-expressions-a}
\end{figure}

\begin{lemma} \label{lem:shift-to-union}
    Let $k \in \NN$ and let $\phi$ be a fuse-$k$-expression. 
    Then in time polynomial in $|\phi| + k$ we can compute a fuse-$k$-expression of the same labeled graph such that for every fuse node $t$, every inner node on the path from $t$ to $t_\oplus$ is a fuse-node.
\end{lemma}

\begin{proof}
    We start with $\phi$ and transform it as long as there is a fuse-node violating the property of the lemma.
    We say that such a node is simply \emph{violating}.
    If there are multiple such nodes, we choose a node $t$ to be processed such that every fuse-node being a proper successor of $t$ satisfies the desired property.
    Since any parse tree is acyclic, such a node exists.
    So let $t$ be a fuse-node to be processed and let $i \in [k]$ be such that $t$ is a $\theta_i$-node. 
    We will shift $t$ to $t_\oplus$ by applying the rules from \cref{lem:transformation-rules} to $t$ and its successors as follows. 
    When we achieve that the child of $t$ is a union- or a fuse-node, we are done with processing $t$.
    
    While processing $t$, with $t_c$ we always refer to the current child of $t$ and by $\alpha$ we denote the operation in $t$. 
    Recall that $t$ is not useless so as long as $t$ is processed, $\alpha$ is a join or a relabel.
    If $\alpha$ is a join-node, then we apply the rule from \cref{rule:join} to swap $t$ and $t_c$.
    Otherwise, we have $\alpha = \rho_{a \to b}$ for some $a \neq b \in [k]$. 
    We proceed depending on the values of $a$ and $b$.
    If $a \neq i$ and $b \neq i$, then the rule from \cref{rule:relabel} is applied to swap $t$ and $\alpha$.
    If $a = i$ and $b \neq i$ (resp. and $b = i$), then $t$ (resp.\ $t_c$) would be useless so this is not the case.
    We are left with the case $a \neq i$ and $b = i$, i.e., $\alpha = \rho_{a \to i}$.
    Note that here we cannot simply swap the nodes $t$ and $t_c$ since the vertices that have label $a$ at the child of $t_c$ also participate in the fuse at $t$.
    So this is where we will have to apply the rules from \cref{lem:transformation-rules} to longer sequences of nodes.
    From now on, we always consider the maximal sequence $(t_0 = t), t_1, \dots, t_q$ (for some $q \in \NN$) such that for every $j \in [q]$, the node $t_j$ is a relabel-node $\rho_{a_j \to i}$ for some $a_j \in [k]$ and $t_j$ is a child of $t_{j-1}$.
    In particular, we have $t_1 = t_c$.
    Since the nodes are not useless, the values $\{a_1, \dots, a_q\}$ are pairwise distinct. 
    Let $t'$ be the child of $t_q$.

    \begin{figure}[t]
        \includegraphics{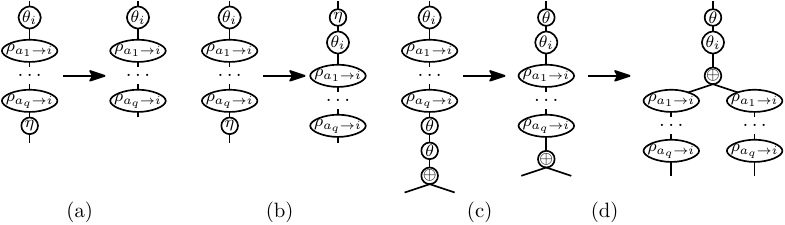}
        \centering
        \caption{Shifts in the proof of \cref{lem:shift-to-union}. (a), (b) A fuse-node $\theta_i$ following relabels to $i$ and a join-node. (c), (d) A fuse-node $\theta_i$ following relabels to $i$ and fuse-nodes.} 
        \label{fig:fuse-shifts-1}
    \end{figure}

    If $t'$ is a join-node, then depending on the joined labels, we apply one of the rules from \cref{rule:relabel-join-both-sides-fused,rule:relabel-join-one-side-fused,rule:relabel-join-no-side-fused} to either suppress $t'$ (see \cref{fig:fuse-shifts-1}~(a)) or shift it above $t$ with possibly changing the labels joined in $t'$ (see \cref{fig:fuse-shifts-1}~(b)).
    If $t'$ is a relabel-node, let $c, d \in [k]$ be such that it is a $\rho_{c \to d}$-node.
    By maximality, we have $d \neq i$.
    Now depending on $c$ and $d$, we can apply one of the rules from \cref{rule:relabel-relabel,rule:relabel-to-i-directly,rule:relabel-from-i}.
    In the case of \cref{rule:relabel-to-i-directly}, the length of the maximal sequence $t_0, \dots, t_q$ increases.
    In the cases of \cref{rule:relabel-relabel,rule:relabel-from-i}, the height of $t$ decreases.
    So in any case we make progress.
    If $t'$ is a union-node, we apply the rule from \cref{rule:relabel-union}. 
    
    Now we may assume that $t'$ is a fuse-node. 
    Observe that while processing $t$ we have not affected the subtree rooted at $t'$ all inner nodes on the path from $t'$ to $t'_\oplus$ are fuse-nodes.
    So there exist $r \in \NN$ with $r > q$, the nodes $(t' = t_{q+1}), \dots, t_r$, and values $b_{q+1}, \dots, b_{r-1} \in [k]$ with the following two properties. 
    First, for every $j \in [q+1, r-1]$, the node $t_j$ is a  $\theta_{b_j}$-node while $t_r$ is a union-node. 
    And second, for every $j \in [q+1, r]$, the node $t_j$ is a child of $t_{j-1}$.
    For $\ell = q+1, \dots, r$ we do the following to achieve that $t_\ell$ is the child of $t_q$.
    This holds at the beginning for $\ell = q+1$.
    Now let $\ell > q + 1$ and suppose this holds for $\ell - 1$.
    Depending on $b_\ell$ we apply the rule from \cref{rule:relabel-fuse-the-same,rule:relabel-fuse-another} to the sequence $(t = t_0), \dots, t_q, t_{\ell-1}$.
    This either suppresses $t_{\ell-1}$ or shifts it to become the parent of $t$.
    In any case, $t_\ell$ becomes the child of $t_q$ as desired.
    In the end, this holds for $\ell = r$, i.e.,
    the vertices $t, t_1, \dots, t_q, t_r$ form a path in the parse tree (see \cref{fig:fuse-shifts-1}~(c)).
    Finally, we apply the rule \cref{rule:relabel-union} to achieve
    that $t$ is a parent of the union-node $t_r$, i.e., $t$ now satisfies the desired condition (see \cref{fig:fuse-shifts-1}~(d)). 
    This concludes the description of the algorithm processing $t$.

    Now we argue that the algorithm terminates and takes only polynomial time.
    We analyze the process for the node $t$ and then conclude about the whole algorithm.
    It can be verified that every application of a rule either decreases the height of $t$ or increases $q$.
    The latter case can only occur at most $k$ times: if $q > k$, then at least one of $t_1, \dots, t_q$ would be redundant.
    So only a polynomial number of rules is applied until $t$ satisfies the property of the lemma.
    The application of any of these rules increases neither the height nor the number of leaves of the parse tree.
    On the other hand, suppressing a useless node below $t$ decreases the height of $t$ as well.
    So to conclude the proof, it suffices to show that for any fuse-node $s$, if $s$ satisfied the property of the lemma before processing $t$, then this still holds after processing $t$, i.e., the number of violating fuse-nodes decreases.
    
    While processing the node $t$, some fuse-node $t^* \in \{t_{q+1}, \dots, t_{r-1}\}$ might violate our desired property when this node is shifted to become the parent of $t$ after the application of the rule from \cref{rule:relabel-fuse-another}.
    But observe that after processing $t$, the path between $t^*$ and $t$ contains only fuse-nodes (which similarly to $s$ have been shifted there as a result of the rule from \cref{rule:relabel-fuse-another}) and the child of $t$ is a union-node.
    So $t^*$ again satisfies the desired condition.
    Therefore, every fuse-node is processed at most once and no new fuse-nodes are created.
    There is a linear number of fuse-nodes and a single application of any rule from \cref{lem:transformation-rules} can be accomplished in polynomial time.
    Above we have argued that per fuse-node, the number of rule applications is polynomial.
    Altogether, the algorithm runs in polynomial time.
\end{proof}

As the next step, we will shift the fuse-nodes further below so that every fuse-node $t$ fuses exactly two vertices, namely one from $G_{t_1}$ with another from $G_{t_2}$.

\begin{lemma}\label{lem:fuse-exactly-two-from-different-sides}
	Let $k \in \NN$ and let $\phi$ be a fuse-$k$-expression. 
	Then in time polynomial in $|\phi| + k$ we can compute a fuse-$k$-expression of the same labeled graph such that for every fuse node $t$, the following holds. 
    First, every inner node on the path from $t$ to $t_\oplus$ is a fuse-node.
    Second, let $i \in [k]$ be such that $t$ is a $\theta_i$-node. 
    Then for every pair $u \neq v \in \lab_t^{-1}(i)$, it holds that $|\{u, v\} \cap V(G_{t_1})| = |\{u, v\} \cap V(G_{t_2})| = 1$.
    In particular, we have $|\lab_t^{-1}(i)| = 2$.
\end{lemma}

\begin{figure}[b]
    \includegraphics{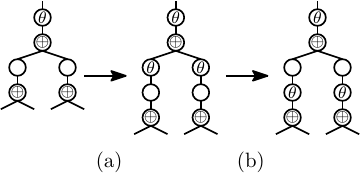}
    \centering
    \caption{(a) Creating copies of fuse-nodes to fuse vertices as early as possible. (b) Shifting the copied fuse-nodes to union-nodes.} 
    \label{fig:glue-expressions-b}
\end{figure}

\begin{proof}
    First of all, we apply \cref{lem:shift-to-union} to transform $\phi$ into a fuse-$k$-expression of the same labeled graph satisfying the properties of that lemma in polynomial time. 
    We still denote the arising fuse-expression by $\phi$ for simplicity.
    We will now describe how to transform $\phi$ to achieve the desired property.
    We will process fuse-nodes one by one and as invariant, we will maintain that after processing any number of fuse-nodes, the expression still satisfies \cref{lem:shift-to-union}.
    
    As long as $\phi$ does not satisfy the desired property, there is a fuse-node $t$ such that at least two fused vertices $u$ and $v$ come from the same side, say $G_{t_1}$ of $t_\oplus$.
    We call $t$ \emph{violating} in this case.
    Since $\phi$ satisfies \cref{lem:shift-to-union}, all inner nodes on the path between $t$ and $t_\oplus$ are fuse-nodes.
    The vertices $u$ and $v$ have therefore the same label in $G_{t_1}$ and they can already be fused before the union. 
    The way we do it might increase the number of fuse-nodes in the expression so we have to proceed carefully to ensure the desired time complexity.
    Now we formalize this idea.
    Among violating fuse-nodes, we always pick a node $t$ with the largest $\oplus$-height to be processed.
    Let $i \in [k]$ be such that $t$ is a $\theta_i$-node.
    First, we subdivide the edge $t_\oplus t_1$ resp.\ $t_\oplus t_2$ with a fresh $\theta_i$-node $t_1'$ resp.\ $t_2'$ (see \cref{fig:glue-expressions-b}~(a)).
    Clearly, this does not change the arising labeled graph and $t$ now fuses at most two vertices: at most one from each side of the union.
    Now the following sets of nodes may become useless: $\{t_1', t\}$, $\{t_2', t\}$, $\{t_1'\}$, $\{t_2'\}$, or $\emptyset$, these nodes are therefore suppressed.
    In particular, if the node $t$ is not suppressed, then it is not violating anymore.
    Let now $\emptyset \neq S \subseteq \{t_1', t_2'\}$ denote the set of non-suppressed nodes in $\{t_1', t_2'\}$.
    These nodes now potentially violate \cref{lem:shift-to-union}.
    So for every node $x'$ in $S$, we proceed as in the proof of \cref{lem:shift-to-union} to ``shift'' $x'$ to $x'_\oplus$ to achieve that all nodes between $x'$ and $x'_\oplus$ are fuse-nodes (see \cref{fig:glue-expressions-b}~(b)).
    Note that this shift only affects the path from $x'$ to $x'_\oplus$.
    
    Observe that every node $x' \in S$ has strictly smaller $\oplus$-height than $t$.
    Thus the order of processing violating fuse-nodes implies that after processing any node, the value $(a, b)$ lexicographically decreases where $a$ denotes the maximum $\oplus$-height over violating fuse-nodes and $b$ denotes the number of violating nodes with $\oplus$-height $a$.
    Indeed, if $t$ was the only violating node with the maximum $\oplus$-height, then $a$ decreases after this process.
    Otherwise, $a$ remains the same, and $b$ decreases.
    Since no new introduce-nodes are created, the maximum $\oplus$-height does not increase and the value of $a$ is always upper-bounded by $|\phi|$.
    Further, recall that after processing a fuse-node, the expression again satisfies \cref{lem:shift-to-union}, i.e., all inner nodes of a path from any fuse-node $s$ to $s_\oplus$ are fuse-nodes. 
    Let us map every fuse-node $s$ to $s_\oplus$.
    The expression never contains useless fuse-nodes so at most $k$ nodes (i.e., one per label) are mapped to any union-node and the value of $b$ never exceeds $k |\phi|$.
    Therefore, the whole process terminates after processing at most $k |\phi|^2$ fuse-nodes.
    Next observe that none of the rules from \cref{lem:transformation-rules} increases the length of some root-to-leaf path.
    Thus, processing a fuse-node $t$ might increase the maximum length of a root-to-leaf path by at most one, namely due to the creation of nodes $t_1'$ and $t_2'$.
    Since on any root-to-leaf path there are at most $|\phi|$ $\oplus$-nodes and there are no useless fuse-nodes, there are at most $k |\phi|$ fuse-nodes on any root-to-leaf path at any moment.
    Initially, the length of any root-to-leaf path is bounded by $|\phi|$ and during the process it increases by at most one for any fuse-node on it.
    Hence, the length of any root-to-leaf path is always bounded by $(k+1)|\phi|$.
    Altogether, processing a single fuse-node can be done in time polynomial in $k$ and $|\phi|$ and the running time of the algorithm is polynomial in $k$ and $|\phi|$.
\end{proof}

Now we may assume that a fuse-expression looks as follows: every union-operation is followed by a sequence of fuse-nodes and each fuse-nodes fuses exactly two vertices from different sides of the corresponding union.
Thus, we can see the sequence consisting of a union-node and following fuse-nodes as a union of two graphs that are \emph{not necessarily vertex-disjoint}.
So these graphs are \emph{glued} at the pairs of fused vertices.
Now we formalize this notion.

A \emph{glue-$k$-expression} is a well-formed expression constructed from introduce-, join-, relabel-, and \emph{glue}-nodes using $k$ labels.
A glue-operation takes as input two $k$-labeled graphs $(H_1, \lab_1)$ and $(H_2, \lab_2)$ satisfying the following two properties:
\begin{itemize}
    \item For every $v \in V(H_1) \cap V(H_2)$, the vertex $v$ has the same label in $H_1$ and $H_2$, i.e., we have $\lab_1(v) = \lab_2(v)$.
    \item For every $v \in V(H_1) \cap V(H_2)$ and every $j \in [2]$, the vertex $v$ is the unique vertex with its label in $H_j$, i.e., we have $|\lab_1^{-1}(\lab_1(v))| = |\lab_2^{-1}(\lab_2(v))| = 1$
\end{itemize}
In this case, we call the $k$-labeled graphs $H_1$ and $H_2$ \emph{glueable}.
The output of this operation is then the union of these graphs denoted by $H_1 \sqcup H_2$, i.e., the labeled graph $(H, \lab)$ with $V(H) = V(H_1) \cup V(H_2)$ and $E(H) = E(H_1) \cup E(H_2)$ where the vertex-labels are preserved, i.e., 
\[
    \lab(v) = 
    \begin{cases}
        \lab_1(v) & \text{if } v \in V(H_1) \\ 
        \lab_2(v) & \text{if } v \in V(H_2)
    \end{cases}
    .
\]
We denote the arising $k$-labeled graph with $H_1 \sqcup H_2$ (omitting $\lab$ for simplicity) and we call the vertices in $V(H_1) \cap V(H_2)$ \emph{glue-vertices}.

\begin{remark}
    Unlike fuse-expressions, if $t$ is a node of a glue-expression $\phi$, then $G^\phi_t$ is a subgraph of $G^\phi$.
\end{remark}

\begin{lemma}
    Let $k \in \NN$ and let $\phi$ be a fuse-$k$-expression. 
    Then in time polynomial in $|\phi| + k$ we can compute a glue-$k$-expression of the same labeled graph.
\end{lemma}

\begin{proof}
    In polynomial time we can obtain a fuse-$k$-expression satisfying \cref{lem:fuse-exactly-two-from-different-sides} that creates the same graph.
    For simplicity, we still denote this expression by $\phi$.
    We assume that the introduce-nodes of $\phi$ use pairwise distinct titles.
    For titles $v$ and $w$, by \emph{identification of $v$ with $w$} we denote the operation that for every $i \in [k]$, replaces every leaf $v \langle i \rangle$ of the current expression with a leaf $w \langle i \rangle$.
    Informally speaking, our goal is to assign the vertices that are fused at some point in the expression the same title.
    Then such vertices will be ``automatically'' glued by a glue-node.
    We start with $\alpha := \phi$ and $\alpha$ will always denote current ``mixed'' expression, i.e., it potentially contains union-, fuse-, and glue-nodes simultaneously.

    We process $\oplus$-nodes in the order of increasing $\oplus$-height as follows.
    Let $t$ be the union-node to be processed.
    Let $f^t_1, \dots, f^t_p$ for some $p \in \NN_0$ denote the maximal sequence of predecessors of $t$ in the parse tree of $\alpha$ such that $f^t_1, \dots, f^t_p$ are fuse-nodes and $t, f^t_1, \dots, f^t_p$ form a path in $\alpha$.
    Simply speaking, $f^t_1, \dots, f^t_p$ are exactly the fuse-operations following $t$.
    For $j \in [p]$, let $i^t_j \in [k]$ be such that $f^t_j$ is a $\theta_{i^t_j}$-node.
    If $p > 0$, we denote $f^t_p$ by $t_{\theta}$ for simplicity.
    If $p = 0$, then with $t_\theta$ we denote $t$ itself.
    If $t$ is clear from the context, we sometimes omit this superscript.

    We will replace the path $t, f^t_1, \dots, f^t_p$ with a single glue-node denoted by $t_\sqcup$ and identify some titles in the sub-expression of $\alpha$ rooted at $t$ so that we maintain the following invariant.
    For every union-node $t$ of $\phi$ such that $t$ has already been processed, the labeled graphs $G^\phi_{t_\theta}$ and $G^\alpha_{t_\sqcup}$ are isomorphic.
    This has the following implication.
    For every node $t$ in $\alpha$ such that $t$ is not a glue-node, the labeled graphs $G^\alpha_t$ and $G^\phi_t$ are isomorphic.
    Up to some formalities, this property just ensures that all sub-expressions still create the same labeled graph.

    Let $s_1$ and $s_2$ be the children of $t$ in $\alpha$.
    The order of processing then implies that $\alpha_{s_1}$ and $\alpha_{s_2}$ are glue-$k$-expressions (i.e., contain no fuse-nodes).
    Let $t_1$ and $t_2$ be the children of $t$ in $\phi$. 
    By invariant, the labeled graphs $G_{t_q}^\phi$ and $G_{s_q}^{\alpha}$ are isomorphic for each $q \in [2]$.
    So since $\phi$ satisfies \cref{lem:fuse-exactly-two-from-different-sides},
    for every $j \in [p]$ and $q \in [2]$, there exists exactly one vertex $v^q_{i_j}$ in $G_{t_q}^\alpha$ with label $i_j$.
    Now for every $j \in [p]$, we identify $v^1_{i_j}$ with $v^2_{i_j}$ in $\alpha$.
    Let $\xi$ denote the arising expression.
    Now we replace the sequence $t, f_1, \dots, f_p$ with a single glue-node denoted by $t_\sqcup$.
    And let $\zeta$ denote the constructed expression.
    
    We claim that $G_{t_\sqcup}^{\zeta}$ is isomorphic to $G^{\alpha}_{t_\theta}$.
    First, note that since union-nodes are processed from bottom to top and in $\phi$ all titles are pairwise distinct, there was no title occurring in both $\alpha_{t_1}$ and $\alpha_{t_2}$. 
    Therefore, after identifying $v^1_{i_j}$ with $v^2_{i_j}$ for every $j \in [p]$, we still have that the labeled graph $G^\xi_{t_1}$ (resp.\ $G^\xi_{t_2}$) is isomorphic to the labeled graph $G^\alpha_{t_1}$ (resp.\ $G^\alpha_{t_2}$). 
    In simple words, no identification has lead to the gluing of vertices inside $G^\alpha_{t_1}$ or $G^\alpha_{t_2}$.
    Moreover, the ordering of processed nodes implies that the titles other than $v^2_{i_1}, \dots, v^2_{i_p}$ are pairwise distinct in $\xi_{t_1}$ and $\xi_{t_2}$.
    Therefore, the glue-node $t_\sqcup$ takes two labeled graphs $G^\xi_{t_1} \cong G^{\alpha}_{t_1}$ and $G^\xi_{t_1} \cong G^{\alpha}_{t_1}$ with shared vertices $\{v^2_{i_1}, \dots, v^2_{i_p}\}$ and produces their union. 
    Observe that this is the same as applying the sequence $f_p \circ \dots \circ f_1 \circ t$ to $G^{\alpha}_{t_1}$ and $G^{\alpha}_{t_2}$.
    Therefore, we have $G^\zeta_t \cong G^\alpha_t$ as desired.
    After $t$ is processed, we set $\alpha := \zeta$ to denote the current expression and move on to the next union-node (if exists). 

    After all nodes are processed, the expression $\alpha$ contains neither union- nor fuse-nodes.
    So $\alpha$ is a glue-$k$-expression such that (by invariant) $G^\alpha \cong G^\phi$ holds, i.e., $\alpha$ creates the same labeled graph.
    The number of union-nodes in $\phi$ is bounded by $|\phi|$ and processing a single node can be done in time polynomial in $k$ and the number of leaves of $\phi$ (i.e., also bounded by $|\phi|$).
    Hence, the transformation takes time polynomial in $k$ and $\phi$.
\end{proof}

\begin{remark}
    Transforming a glue-$k$-expression into a fuse-$k$-expression is trivial: Replace every glue-node by a union-node followed by a sequence $\theta_{i_1}, \dots, \theta_{i_q}$ of fuse-nodes where $i_1, \dots, i_q$ are the labels of vertices shared by the glued graphs.
    This implies that fuse- and glue-expressions are equivalent and there is no reason to define ``glue-width'' as a new parameter.
\end{remark}

As a last step of our transformations, we show that similarly to the existence of irredundant $k$-expressions defined by Courcelle and Olariu~\cite{CourcelleO00} and widely used in dynamic-programming algorithms (e.g., \cite{FominGLS14,BergougnouxKK20,HegerfeldK23}), certain irredundancy can be achieved for glue-expressions.

\begin{lemma}\label{lem:useful-glue-expression}
    Let $k \in \NN$ and let $\phi$ be a fuse-$k$-expression. 
    Then in time polynomial in $|\phi| + k$ we can compute a glue-$k$-expression $\xi$  of the same labeled graph without useless nodes such that: 
    \begin{enumerate}
        \item Let $i, j \in [k]$, let $t$ be a $\eta_{i, j}$-node in $\xi$, and let $t'$ be the child of $t$ in $\xi$. 
        Then $G^\xi_{t'}$ contains no edge $\{v, w\}$ with $\lab^\xi_{t'}(v) = i$ and $\lab^\xi_{t'}(w) = j$.
        \item Let $t$ be a glue-node in $\xi$ and let $t_1$ and $t_2$ be its children.
        Then the graphs $G^{\xi}_{t_1}$ and $G^{\xi}_{t_2}$ are edge-disjoint.
        \item Let $t$ be a glue-node in $\xi$, let $t_1$ and $t_2$ be its children, and let $v$ be a glue-vertex. 
        Then for every $q \in [2]$, the vertex $v$ has an incident edge in $G^\xi_{t_q}$.
    \end{enumerate}
    We call a glue-$k$-expression satisfying these properties \emph{reduced}.
\end{lemma}

\begin{proof}
    First, we apply \cref{lem:fuse-exactly-two-from-different-sides} to obtain a glue-$k$-expression $\phi$ of the same graph in polynomial time. 
    As in the previous proofs of this section, we will transform $\phi$ iteratively until it satisfies the desired properties.
    
    In the first phase, as long as there is a join-node $t$ and an edge $\{v, w\}$ violating the first property, we proceed as follows. 
    There exists a successor $t'$ of $t$ in $\phi$ such that $t'$ is a $\eta_{i', j'}$-node for some $i', j' \in [k]$, the vertices $v$ and $w$ are the vertices of $G_{t'}$, and it holds that $\lab^\phi_{t'}(v) = i'$ and $\lab^\phi_{t'}(w) = j'$.
    There can be multiple such nodes $t'$ so we fix an arbitrary one.
    We suppress the node $t'$.
    Let $\phi'$ denote the arising expression.
    Note that once two vertices have the same label, this property is maintained along the expression.
    So similarly to the construction of irredundant clique-expressions (see \cite{CourcelleO00}), it holds that every edge $e'$ created by $t'$ is also created by $t$.
    Formally, the following holds.
    Since $t'$ is a successor of $t$, for all $v' \in V(G^\phi_{t'})$ the property $\lab^\phi_{t'}(v') = \lab^\phi_{t'}(v)$ implies $\lab^\phi_{t}(v') = \lab^\phi_{t}(v)$ and hence, also $\lab^{\phi'}_{t}(v') = \lab^{\phi'}_{t}(v)$.
    The analogous statement holds for vertices $w' \in V(G^\phi_{t'})$ with $\lab^\phi_{t'}(w') = \lab^\phi_{t'}(w)$.
    Therefore, the labeled graphs $G^{\phi}$ and $G^{\phi'}$ are isomorphic.
    Now we set $\phi := \phi'$, and the process is repeated until $\phi$ satisfies the first property.
    As mentioned above, the node $t'$ is not necessarily unique for $t$ so after $t'$ is suppressed, $t$ and $\{v, w\}$ can still violate the first property of the lemma.
    The number of join-nodes decreased by one though.
    Therefore, the process terminates after at most $|\phi|$ steps and it results in a glue-$k$-expression of the same labeled graph.
    Clearly, each step takes only polynomial time so the running time of this transformation is polynomial.

    In the second phase, we proceed similarly to satisfy the second property.
    As long as there exist a glue-node $t$ and an edge $\{v, w\} \in E(G^\phi_{t_1}) \cap E(G^\phi_{t_2})$ violating the second property, we proceed as follows.
    Note that $v$ and $w$ are then glue-vertices.
    There exists a successor $t'$ of $t_1$ such that $t'$ is a $\eta_{i', j'}$-node for some $i', j' \in [k]$, the vertices $v$ and $w$ are the vertices of $G_{t'}$, and it holds that $\lab^\phi_{t'}(v) = i'$ and $\lab^\phi_{t'}(w) = j'$.
    We claim that $e$ is then the only edge created by $t'$, i.e., $v$ (resp.\ $w$) is the unique vertex with label $i'$ (resp.\ $j'$) in $G^{\phi}_{t'}$.
    Suppose not, then without loss of generality, there exists a vertex $v' \neq v$ in $G^{\phi}_{t'}$ with label $i'$.
    Then the vertices $v$ and $v'$ would also have the same label in $G^{\phi}_{t_1}$.
    But the glueabilty of $G_{t_1}^\phi$ and $G_{t_2}^{\phi}$ implies that $v$ is the unique vertex with label $(\lab^{\phi}_{t_1})^{-1}(v)$ in $G^\phi_{t_1}$ -- a contradiction.
    Let now $\phi'$ denote the expression arising from $\phi$ by suppressing $t'$.
    Then it holds $G^{\phi'}_{t_1} = G^{\phi}_{t_1} - e$.
    Therefore, $G^{\phi'}_t = G^{\phi}_t$ and also $G^{\phi'} = G^{\phi}$.
    Now we set $\phi := \phi'$ and repeat the process until $\phi$ satisfies the second property.
    Since in each repetition the expression loses one join-node, the process terminates after at most $|\phi|$ steps.
    Also, one step takes only polynomial time so the total running time is also polynomial.
    Clearly, the first condition remains satisfied during the process. 

    Now we move on to the third property.
    Let $t$ and $v$ violate it, i.e., without loss of generality $v$ has no incident edge in $G^\phi_{t_1}$.
    The crucial observation for the transformation described below is that since $v$ also belongs to $G^\phi_{t_2}$, it holds that
    \[
        G^\phi_{t_1} \sqcup G^\phi_{t_2} = (G^\phi_{t_1} - \{v\}) \sqcup G^\phi_{t_2}
    \]
    (in particular, the glueability precondition applies to the right side are as well).
    Hence, we want to transform the sub-expression rooted at $t_1$ so that it does not create the vertex $v$.
    For this, we simply need to cut off all introduce-nodes using title $v$ from this sub-expression.
    Formally, we start with $\phi' := \phi$ and as long as $\phi'_{t_1}$ contains a leaf $\ell$ being a $v\langle i \rangle$-node for some $i \in [k]$:
    \begin{enumerate}
        \item As long as the parent $t'$ of $\ell$ is not a glue-node we repeat the following.
        Since there are no useless nodes, $t'$ is a $\rho_{i \to j}$-node so we apply the rule from \cref{rule:relabel-introduce-1} to suppress it.
        \item Now $t'$ is a glue-node so we remove $\ell$ and suppress $t'$.
    \end{enumerate}
    Note that when the last item is reached, the parent $t'$ is a glue-node whose one side of the input is the graph that consists of a single vertex $v$.
    A simple inductive argument shows that $G^{\phi'}_{t_1} = G^\phi_{t_1} - \{v\}$ holds as desired.
    Now we set $\phi := \phi'$ and repeat until the third property is satisfied.
    Clearly, the first two properties are maintained and since in each repetition the size of the expression decreases by at least one, the process take only polynomial time.
\end{proof}

It is known that for clique-expressions, 
the number of leaves of an expression is equal to the number of vertices in the arising graph.
For fuse-, and hence, glue-expressions, the situation is different.
Since a fuse-node, in general, decreases the number of vertices in the arising graph, the number of leaves in a fuse-expression can be unbounded.
However, we now show that the number of leaves of a \textbf{reduced} glue-expression is bounded by $\mathcal{O}(m+n)$ where $n$ resp.\ $m$ is the number of vertices resp.\ edges of the arising graph.
This will lead to an upper bound of $\mathcal{O}(k^2(m+n))$ on the number of nodes of a reduced glue-expression.

\begin{theorem}\label{app:thm:useful-expression}
    Let $\phi$ be a fuse-$k$-expression of a graph $H$ on $n$ vertices and $m$ edges.
    Then in time polynomial in $|\phi|$ and $k$ we can compute a reduced glue-$k$-expression $\zeta$ of $H$ such that the parse tree of $\zeta$ contains $\mathcal{O}(k^2(m+n))$ nodes.
\end{theorem}

\begin{proof}
    By \cref{lem:useful-glue-expression}, in polynomial time we can compute a reduced glue-$k$-expression $\xi$ of the same labeled graph. 
    Let $L(\xi)$ denote the set of leaves of $\xi$.
    We start by bounding the size of this set.
    Let $H = G^\xi$.
    We will now define a mapping $h: L(\xi) \to V(H) \cup E(H)$ and then show some properties of it.
    Let $\ell \in L(\xi)$ and let $v \in V(H)$ and $i \in [k]$ be such that $\ell$ is a $v\langle i \rangle$-node.
    If no other leaf in $L(\xi)$ has title $v$, then we simply set $h(\ell) = v$.
    Otherwise, there exists at least one other leaf with title $v$.
    Hence, there exists at least one glue-node $t$ such that $v$ is a glue-vertex at $t$.
    Note that any such $t$ is an ancestor of $\ell$.
    So let $g(\ell)$ denote the bottommost among such glue-nodes and let $s_1$ and $s_2$ be its children.
    Without loss of generality, we assume that $\ell$ belongs to the sub-expression rooted at $s_1$.
    The reducedness of $\xi$ implies that there is an edge $e \in E(G^\xi_{s_1}) \setminus E(G^\xi_{s_2})$ incident with $v$. 
    We set $h(\ell) = e$.
    In particular we have $e \in G^\xi_{g(\ell)}$.
    
    Observe that any vertex is mapped by $h$ either to itself or to an incident edge.
    Now let $e$ be an edge of $H$ and let $v$ be one of its end-vertices.
    We claim that there exists at most one leaf with title $v$ mapped to $e$.
    For the sake of contradiction, suppose there exist leaves $\ell_1 \neq \ell_2$ and $i_1, i_2 \in [k]$ such that $\ell_1$ (resp.\ $\ell_2$) is a $v\langle i_1 \rangle$-node (resp.\ $v\langle i_2 \rangle$-node) and $h(\ell_1) = h(\ell_2) = e$. 
    Then let $t$ denote the lowest common ancestor of $g(\ell_1)$ and $g(\ell_2)$.
    Let $q \in [2]$ be arbitrary.
    Let $t_q$ be the child of $t$ such that $g(\ell_q)$ is a (not necessarily proper) descendant of $t_q$.
    The property $h(\ell_q) = e$ implies that $e \in E(G^\xi_{g(\ell_q)}) \subseteq E(G^\xi_{t_q})$ holds.
    Therefore, we obtain $e \in E(G^\xi_{t_1}) \cap E(G^\xi_{t_2})$ contradicting the reducedness of $\xi$.
    So there indeed exists at most one leaf in $\xi$ with title $v$ mapped to $e$.
    Since $e$ has two end-points, there are at most two leaves of $\xi$ mapped to $e$.
    Finally, for every vertex $u$ of $H$ in the image of $h$, we know that there exists exactly one leaf with title $u$ so there exists at most one leaf of $\xi$ mapped to $u$.
    These two properties imply that the size of preimage of $h$, i.e., the size of $L(\xi)$ is bounded by $2m + n$.
    It is folklore that a rooted tree with at most $2m+n$ leaves has $\mathcal{O}(m+n)$ inner nodes with at least two children.
    Hence, $\xi$ has $\mathcal{O}(m+n)$ glue-nodes.
    
    Now we will apply a simple folklore trick (originally applied to clique-expressions)
    to bound the number of nodes between any two consecutive glue-nodes.
    For this, we apply rules from \cref{rule:join-relabel-2,rule:join-relabel-3} of \cref{lem:transformation-rules} to ensure that for any two consecutive glue-nodes $t_1$ and $t_2$ (where $t_1$ is an ancestor of $t_2$), the path from $t_2$ to $t_1$ first contains join-nodes and then relabel-nodes.
    Any duplicate on this path would be useless.
    Therefore, this path contains $\mathcal{O}(k^2)$ relabel- and join-nodes (at most one per possible operation).
    Finally, by applying the rule from \cref{rule:relabel-introduce-1} of \cref{lem:transformation-rules} we can ensure that a parent of any introduce-node is a glue-node.
    Let $\zeta$ be the arising glue-expression.
    Clearly, $\zeta$ is still reduced.
    The number of relabel- and join-nodes in $\zeta$ is now bounded by $\mathcal{O}(k^2(m+n))$ and so the total numbers of nodes is also bounded by this value as claimed.
\end{proof}

%% file: fw-algorithms.tex
\section{ Algorithms Parameterized by Fusion-Width}\label{app:sec:fw-algorithms}

In this section, we parameterize by fusion-width. 
We will present three algorithms for problems $W[1]$-hard when parameterized by clique-width, these are \textsc{Hamiltonian Cycle}, \textsc{Max Cut}, and \textsc{Edge Dominating Set}.
The algorithms are XP-algorithms and they have the same running time as the tight (under ETH) algorithms parameterized by clique-width.
Fusion-width is upper-bounded by clique-width and this has two implications.
First, the lower bounds from clique-width apply to fusion-width as well, and hence, our algorithms are also ETH-tight.
And second, our results show that these problems can be solved for a larger class of graphs within the same running time.  
Each of the following algorithms gets a fuse-$k$-expression of the input graph and as the very first step transforms it into a reduced glue-$k$-expression of the same graph in polynomial time (cf.\ \cref{lem:useful-glue-expression}).
By \cref{app:thm:useful-expression}, the size of the expression is then linear in the size of the graph. 

\subsection{ \textsc{Max Cut}}\label{app:subsec:maxcut}
In this problem, given a graph $G = (V, E)$ we are asked about the maximum cardinality of $E_G(V_1, V_2)$ over all partitions $(V_1, V_2)$ of $V$.
In this subsection we solve the \textsc{Max Cut} problem in time $n^{\mathcal{O}(k)}$ given a fuse-$k$-expression of a graph. 
For this, we will rely on the clique-width algorithm by Fomin et al.\ and use the same dynamic-programming tables~\cite{FominGLS14}. 
Then, it suffices to only show how to handle glue-nodes.
Later, in the description of the algorithm, we also sketch why general fuse-expressions (i.e., with unrestricted fuse-nodes) seem to be problematic for the algorithm.

Let $H$ be a $k$-labeled graph.
The table $T_H$ contains all vectors $h = (s_1, \dots, s_k, r)$ with $0 \leq s_i \leq |U^H_i|$ for every $i \in [k]$ and $0 \leq r \leq |E(H)|$ for which there exists a partition $(V_1, V_2)$ of $V(H)$ such that $|V_1 \cap U^H_i| = s_i$ for every $i \in [k]$ and there are at least $r$ edges between $V_1$ and $V_2$ in $H$.
We say that the partition $(V_1, V_2)$ \emph{witnesses} the vector $h$.
For their algorithm, Fomin et al.\ provide how to compute these tables for nodes of a $k$-expression of the input graph.
In particular, they show that this table can be computed for a $k$-labeled graph $H$ with $n$ vertices in time $n^{\mathcal{O}(k)}$ in the following cases:
\begin{itemize}
    \item If $H$ consists of a single vertex.
    \item If $H = \rho_{i \to j}(H')$ for some $i, j \in [k]$ and a $k$-labeled graph $H'$ given $T_{H'}$.
    \item If $H = \eta_{i, j}(H')$ for some $i \neq j \in [k]$ and a $k$-labeled graph $H'$ given $T_{H'}$.
\end{itemize}
The correctness of their algorithm requires that a $k$-expression is irredundant, i.e., no join-node creates an already existing edge.
Our extended algorithm will process a reduced glue-expression and the first property in \cref{lem:useful-glue-expression} will ensure that this property holds so the approach of Fomin et el.\ can indeed be adopted for the above three types of nodes.

First, let us mention that processing a fuse-node $\theta_i$ seems problematic (at least using the same records).
Some vertices of label $i$ might have common neighbors. 
So when fusing these vertices, multiple edges fall together but we do not know how many of them so it is unclear how to update the records correctly.
For this reason, we make use of glue-expressions where, as we will see, the information stored in the records suffices.

To complete the algorithm for the fusion-width parameterization, we provide a way to compute the table $T_H$ if $H = H^1 \sqcup H^2$ for two glueable edge-disjoint $k$-labeled graphs $H^1$ and $H^2$ if the tables $T_{H^1}$ and $T_{H^2}$ are provided. 
Let $\{v_1, \dots, v_q\} = V(H^1) \cap V(H^2)$ for some $q \in \NN_0$ and let $i_1, \dots, i_q$ be the labels of $v_1, \dots, v_q$ in $H^1$, respectively.
The glueability implies that for every $j \in [q]$, it holds that $|U^{H^1}_{i_j}| = |U^{H^2}_{i_j}| = 1$.
Hence, for every entry $(s_1, \dots, s_k, r)$ of $T_{H^1}$ and every $j \in [q]$, it also holds that $s_{i_j} \in \{0, 1\}$ with $s_{i_j} = 1$ if and only if $v_j$ is put into $V_1$ in the partition witnessing this entry.
The same holds for the entries in $T_{H^2}$.
This gives the following way to compute the table $T_H$.
It will be computed in a table $S_H$.
We initialize this table to be empty.
Then we iterate through all pairs of vectors $h^1 = (s_1^1, \dots, s_k^1, r^1)$ from $T_{H^1}$ and $h^2 = (s_1^2, \dots, s_k^2, r^2)$ from $T_{H^2}$.
If there is an index $j \in [q]$ such that $s_{i_j}^1 \neq s_{i_j}^2$, then we skip this pair.
Otherwise, for every $0 \leq r \leq r^1 + r^2$, we add to $S_H$ the vector $h = (s_1, \dots, s_k, r)$ where for all $i \in [k]$
\[
    s_i =
    \begin{cases}
        s_i^1 + s_i^2 & i \notin \{i_1, \dots, i_q\} \\
        s_i^1 \cdot s_i^2 & i \in \{i_1, \dots, i_q\}
    \end{cases}
\]
and we call $h^1$ and $h^2$ \emph{compatible} in this case.

\begin{claim}
    The table $S_H$ contains exactly the same entries as $T_H$.
\end{claim}

\begin{claimproof}
    For the one direction, let $h^1 = (s_1^1, \dots, s_k^1, r^1)$ and $h^2 = (s_1^2, \dots, s_k^2, r^2)$ be compatible entries of $T_{H^1}$ and $T_{H^2}$, respectively. 
    And let $(V^1_1, V^1_2)$ and $(V^2_1,V^2_2)$ be the partitions witnessing $h^1$ in $H^1$ and $h^2$ in $H^2$, respectively.
    Also let $0 \leq r \leq r_1 + r_2$.
    We claim that then $(V_1, V_2)$ with $V_1 = V^1_1 \cup V^2_1$ and $V_2 = V^1_2 \cup V^2_2$ is a partition witnessing a vector $h = (s_1, \dots, s_k, r)$ of $S_H$ constructed as above in $H$ so $h$ belongs to $T_H$.
    First, we show that this is a partition of $V(H)$.
    We have
    \[
        V(H) = V(H^1) \cup V(H^2) = (V^1_1 \cup V^1_2) \cup (V^2_1 \cup V^2_2) = (V^1_1 \cup V^2_1) \cup (V^1_2 \cup V^2_2) = V_1 \cup V_2.
    \]
    Since the sets $V^1_1$ and $V^1_2$ are disjoint and also the sets $V^2_1$ and $V^2_2$ are disjoint, we have:
    \[
       V_1 \cap V_2 = (V^1_1 \cap V^2_2) \cup (V^2_1 \cap V^1_2).
    \]
    Any vertex $v$ in $V^1_1 \cap V^2_2$ belongs to both $H^1$ and $H^2$ so there exists an index $j^* \in [q]$ with $v = v_q$.
    The property $v \in V^1_1$ then implies $s^1_{i_j} = 1$ while the property $v \in V^2_2$ implies $s^2_{i_j} = 0$.
    So we obtain $s^1_{i_j} \neq s^2_{i_j}$ contradicting the fact that $h^1$ and $h^2$ are compatible.
    Therefore, $(V^1_1 \cap V^2_2)$ is empty.
    A symmetric argument shows that $(V^2_1 \cap V^1_2)$ is empty as well.
    Hence $V_1$ and $V_2$ are disjoint and they indeed form a partition of $V(H)$.
    Let $j \in [q]$.
    Since $v_j$ is the unique vertex with label $i_j$ in $H$, the set $V_1$ contains exactly one vertex of this label if $s_{i_j}^1 = s_{i_j}^2 = 1$ and zero such vertices if $s_{i_j}^1 = s_{i_j}^2 = 0$. 
    So we have 
    \[
        |V_1 \cap U^{i_j}_H| = s_{i_j}^1 \cdot s_{i_j}^2 = s_{i_j}.
    \]
    For every $i \in [k] \setminus \{i_1, \dots, i_q\}$ the sets $U^{H^1}_i$ and $V(H^2)$ as well as $U^{H^2}_i$ and $V(H^1)$ are disjoint by definition of $\{i_1, \dots, i_q\}$ and therefore:
    \begin{align*}
        |V_1 \cap U^H| = &|V_1 \cap (U^{H^1}_i \cup U^{H^2}_i)| = \\
        & |(V_1 \cap U^{H^1}_i) \cup (V_1 \cap U^{H^2}_i)| = \\
        & |(V_1 \cap U^{H^1}_i)| + |(V_1 \cap U^{H^2}_i)| = \\
        & |(V_1^1 \cap U^{H^1}_i)| + |(V_1^2 \cap U^{H^2}_i)| = \\
        & s_i^1 + s_i^2 = \\
        & s_i.
    \end{align*}
    Finally, we bound the number of edges in $E_H(V_1, V_2)$.
    It holds that $E_{H^b}(V^b_1, V^b_2) \subseteq E_H(V^b_1, V^b_2) \subseteq E_H(V_1, V_2)$ for every $b \in [2]$ so we obtain
    \[
        E_{H^1}(V^1_1, V^1_2) \cup E_{H^2}(V^2_1, V^2_2) \subseteq E_H(V_1, V_2).
    \]
    Recall that the graphs $H^1$ and $H^2$ are edge-disjoint, then we have
    \[
        r \leq r^1 + r^2 \leq |E_{H^1}(V^1_1, V^1_2)| + |E_{H^2}(V^2_1, V^2_2)| \leq |E_H(V_1, V_2)|.
    \]
    So $(V_1, V_2)$ is indeed a partition witnessing $h$ in $H$.
    
    For the other direction, let $h = (s_1, \dots, s_k, r)$ be an entry of $T_H$.
    Then there exists a partition $(V_1, V_2)$ of $V(H)$ witnessing $h$.
    We will show that there exist entries $h^1$ and $h^2$ of $T_{H^1}$ and $T_{H^2}$, respectively, such that the above algorithm adds $h$ to the table $S_H$ at the iteration of $h^1$ and $h^2$.
    We set $V^{j_2}_{j_1} = V_{j_1} \cap V(H^{j_2})$ for $j_1, j_2 \in [2]$.
    Let $j \in [2]$ be arbitrary but fixed.
    Since $(V_1, V_2)$ is a partition of $V(H)$ and we have $V(H^j) \subseteq V(H)$, the pair $(V_1^j, V_2^j)$ is a partition of $V(H^j)$.
    Let $r^j = |E_{H^j}(V_1^j, V_2^j)|$ and let $h^j = (s_1^j, \dots, s_k^j, r^j)$ be the vector such that $(V_1^j, V_2^j)$ witnesses $h^j$.
    First, consider $i \in \{i_1, \dots, i_q\}$.
    Recall that glueability implies that $s_i^j, s_i \in \{0, 1\}$.
    If $s_i = 1$, then we have $v_i \in V_1$ and therefore also $v_i \in V_1^j$ so we obtain $s_i^j = 1$.
    Similarly, if $s_i = 0$, then we have $v_i \in V_2$ and therefore also $v_i \in V_2^j$ so we obtain $s_i^j = 0$.
    Therefore, it holds that $s_i = s_i^1 \cdot s_i^2$.
    Next, consider $i \in [k] \setminus \{i_1, \dots, i_q\}$.
    The sets $U^{H^1}_i$ and $U^{H^2}_i$ are disjoint.
    Therefore, the sets $V_1 \cap U^{H^1}_i = V_1^1 \cap U^{H^1}_i$ and $V_1 \cap U^{H^2}_i = V_1^2 \cap U^{H^2}_i$ partition $V_1 \cap U^H_i$ so we obtain $s_i = s_i^1 + s_i^2$.
    Let $W^1 = V(H^1) \setminus V(H^2)$, $W^2 = V(H^2) \setminus V(H^1)$, and let $I = V(H^1) \cap V(H^2)$.
    For $j_1, j_2 \in [2]$, let $W^{j_1}_{j_2} = W^{j_1} \cap V_{j_2}$ and let $I_{j_2} = I \cap V_{j_2}$.
    Then $W^1_j$, $W^2_j$, and $I_j$ partition $V_j$ for $j \in [2]$.
    The following holds:
    \begin{align*}\label{eq:maxcut-cut-edges}
        & E_H(V_1, V_2) = \\ 
        & E_H(W^1_1, W_2^1) \cup E_H(W^1_1, I_2) \cup E_H(W^1_1, W_2^2) \cup \\
        & E_H(I_1, W_2^1) \cup E_H(I_1, I_2) \cup E_H(I_1, W_2^2) \cup \\
        & E_H(W^2_1, W_2^1) \cup E_H(W^2_1, I_2) \cup E_H(W^2_1, W_2^2) =\\
        & E_H(W^1_1, W_2^1) \cup E_H(W^1_1, I_2) \cup E_H(W^1_1, W_2^2) \cup \\
        & E_H(I_1, W_2^1) \cup \bigl(E_{H^1}(I_1, I_2) \cup E_{H^2}(I_1, I_2)\bigr) \cup E_H(I_1, W_2^2) \cup \\
        & E_H(W^2_1, W_2^1) \cup E_H(W^2_1, I_2) \cup E_H(W^2_1, W_2^2). \\
    \end{align*}
    Therefore, the sets occurring after the second equality are pairwise disjoint so the size of $E_H(V_1, V_2)$ is the sum of their sizes.
    Next, recall that every edge of $H$ is either an edge of $H^1$ or of $H^2$ and therefore, for every edge of $H$, there exists an index $j^*$ such that both end-points of this edge belong to $H^{j*}$.
    Therefore, the sets $E_H(W^1_1, W_2^2)$ and $E_H(W^2_1, W_2^1)$ are empty.
    This also implies the following equalities:
    \begin{align*}
        &E_H(W^1_1, W_2^1) &=&& E_{H^1}(W^1_1, W_2^1) \\
        &E_H(W^1_1, I_2) &= &&E_{H^1}(W^1_1, I_2) \\
        &E_H(I_1, W_2^1) &= &&E_{H^1}(I_1, W_2^1) \\
        &E_H(I_1, W_2^2) &= &&E_{H^2}(I_1, W_2^2) \\
        &E_H(W^2_1, I_2) &= &&E_{H^2}(W^2_1, I_2) \\
        &E_H(W^2_1, W_2^2) &= &&E_{H^2}(W^2_1, W_2^2) 
    \end{align*}
    Now by using these properties we obtain
    \begin{align*}
        &E_H(V_1, V_2) = \\
        &E_{H^1}(W^1_1, W_2^1) \cup E_{H^1}(W^1_1, I_2) \cup \\
        &E_{H^1}(I_1, W_2^1) \cup E_{H^1}(I_1, I_2) \cup E_{H^2}(I_1, I_2) \cup E_{H^2}(I_1, W_2^2) \cup \\
        &E_{H^2}(W^2_1, I_2) \cup E_{H^2}(W^2_1, W_2^2). \\
    \end{align*}
    By rearranging the terms we get
    \begin{align*}
        &E_H(V_1, V_2) = \\
        &\bigl(E_{H^1}(W^1_1, W_2^1) \cup E_{H^1}(W^1_1, I_2) \cup E_{H^1}(I_1, W_2^1) \cup E_{H^1}(I_1, I_2)\bigr) \cup \\
        &\bigl(E_{H^2}(I_1, I_2) \cup E_{H^2}(I_1, W_2^2) \cup E_{H^2}(W^2_1, I_2) \cup E_{H^2}(W^2_1, W_2^2)\bigr). \\
    \end{align*}
    Finally, note that for $j_1, j_2 \in [2]$, the pair $(W^{j_1}_{j_2}, I_{j_2})$ is a partition of $V^{j_1}_{j_2}$.
    So we obtain 
    \[
        E_H(V_1, V_2) = E_{H^1}(V^1_1, V^1_2) \cup E_{H^2}(V^2_1, V^2_2).
    \]
    Since the graphs $H^1$ and $H^2$ are edge-disjoint, we get
    \[
        r^1 + r^2 = |E_{H^1}(V^1_1, V^1_2)| + |E_{H^2}(V^2_1, V^2_2)| = |E_H(V_1, V_2)| = r.
    \]
    Therefore, at the iteration corresponding to $h^1$ and $h^2$ the algorithm indeed adds $h$ to $S_H$.
    This concludes the proof of the correctness of the algorithm.
\end{claimproof}

Observe that if a graph $H$ has $n$ nodes, the table $T_H$ contains $n^{\mathcal{O}(k)}$ entries.
Therefore, this table can be computed from $T_{H^1}$ and $T_{H^2}$ in time $n^{\mathcal{O}(k)}$ as well.
This results in an algorithm that given a graph $H$ together with a reduced glue-$k$-expression $\xi$ of $H$, traverses the nodes $x$ of $\xi$ in a standard bottom-up manner and computes the tables $T_{G^\xi_x}$ in time $n^{\mathcal{O}(k)}$.
Let $y$ denote the root of $\xi$. 
Then $G^\xi_y$ is exactly the graph $H$ so we output the largest integer $r$ such that $T_{G^\xi_y}$ contains an entry $(s_1, \dots, s_k, r)$ for some $s_1, \dots, s_k \in \NN_0$.
By definition, this value is then the size of the maximum cardinality cut of the graph $H$.

\begin{theorem}
    Given a fuse-$k$-expression of a graph $H$, the \textsc{Max Cut} problem can be solved in time $n^{\mathcal{O}(k)}$.
\end{theorem}

Fomin et al.\ have also shown the following lower bound:
\begin{theorem}\cite{FominGLS14}
    Let $H$ be an $n$-vertex graph given together with a $k$-expression of $H$. 
    Then the \textsc{Max Cut} problem cannot be solved in time $f(k) \cdot n^{o(k)}$ for any computable function $f$ unless the ETH fails.
\end{theorem}

Since any $k$-expression of a graph is, in particular, a fuse-$k$-expression of the same graph, the lower bound transfers to fuse-$k$-expressions as well thus showing that our algorithm is tight under ETH.
\begin{theorem}
    Let $H$ be an $n$-vertex graph given together with a fuse-$k$-expression of $H$. 
    Then the \textsc{Max Cut} problem cannot be solved in time $f(k) \cdot n^{o(k)}$ for any computable function $f$ unless the ETH fails.
\end{theorem}

\subsection{ \textsc{Edge Dominating Set}}\label{app:subsec:eds}

In this problem, given a graph $G = (V, E)$ we are asked about the cardinality of a minimum set $X \subseteq E$ such that every edge in $E$ either belongs to $X$ itself or it has an incident edge in $X$.
In this section, we provide a way to handle the glue-nodes in order to solve the \textsc{Edge Dominating Set} problem.
As in the previous subsection, we rely on the dynamic programming algorithm for the clique-width parameterization by Fomin et al.\ and use their set of records defined as follows.
For a $k$-labeled graph $H$, the table $T_H$ contains all vectors $(s_1, \dots, s_k, r_1, \dots, r_k, \ell)$ of non-negative integers such that there exists a set $S \subseteq E(H)$ and a set $R \subseteq V(H) \setminus V(S)$ with the following properties:
\begin{itemize}
    \item $|S| \leq \ell \leq |E(H)|$;
    \item for every $i \in [k]$, exactly $s_i$ vertices of $U^H_i$ are incident with the edges of $S$;
    \item for every $i \in [k]$, we have $|R \cap U^H_i| = r_i$;
    \item every edge of $H$ undominated by $S$ has an end-vertex in $R$.
\end{itemize}
We say that the pair $(S, R)$ witnesses the vector $(s_1, \dots, s_k, r_1, \dots, r_k, \ell)$ in $H$.
The last property reflects that it is possible to attach a \emph{pendant} edge to every vertex in $R$ so that the set $S$ together with these pendant edges dominates all edges of $H$.
In the following, we will sometimes use this view in our arguments and denote the set of edges pendant at vertices of $R$ by $E^R$.
Note that since no vertex incident with $S$ belongs to $R$, we have $s_i + r_i \leq |U^H_i|$ for any $i \in [k]$.
In particular, for every $i \in \{i_1, \dots, i_q\}$, the property $r_i = 1$ implies $s_i = 0$ and therefore
\begin{equation}\label{eq:s-i-r-i-implication}
    r_i \land \neg s_i = r_i.
\end{equation}

For their algorithm, Fomin et al.\ provide how to compute these tables for nodes of a $k$-expression of the input graph.
In particular, they show that this table can be computed for a $k$-labeled graph $H$ with $n$ vertices in time $n^{\mathcal{O}(k)}$ in the following cases:
\begin{itemize}
    \item If $H$ consists of a single vertex.
    \item If $H = \rho_{i \to j}(H')$ for some $i, j \in [k]$ and a $k$-labeled graph $H'$ given the table $T_{H'}$.
    \item If $H = \eta_{i, j}(H')$ for some $i \neq j \in [k]$ and a $k$-labeled graph $H'$ given the table $T_{H'}$.
\end{itemize}

Similarly to the previous subsection, let us mention that processing a fuse-node $\theta_i$ seems problematic (at least using the same records).
Some vertices of label $i$ might have common neighbors. 
So when fusing these vertices, multiple edges of the set $S$ of a partial solution fall together but we do not know how many of them so it is unclear how to update the records correctly.
For this reason, we make use of glue-expressions where, as we will see, the information stored in the records suffices.

To complete the algorithm for the fusion-width parameterization, we provide a way to compute the table $T_H$ if $H = H^1 \sqcup H^2$ for two glueable edge-disjoint $k$-labeled graphs $H^1$ and $H^2$ if the tables $T_{H^1}$ and $T_{H^2}$ are provided. 
Let $\{v_1, \dots, v_q\} = V(H^1) \cap V(H^2)$ for some $q \in \NN_0$ and let $i_1, \dots, i_q$ be the labels of $v_1, \dots, v_q$ in $H^1$, respectively.
Then for every $j \in [q]$, it holds that $|U^{H^1}_{i_j}| = |U^{H^2}_{i_j}| = 1$.
Hence, for every entry $(s_1, \dots, s_k, r_1, \dots, r_k, \ell)$ of $T_{H^1}$ and every $j \in [q]$, it holds that $s_{i_j} + r_{i_j} \leq 1$.
The same holds for the entries in $T_{H^2}$.
This motivates the following way to compute the table $T_H$.
It will be computed in a table $S_H$.
We initialize this table to be empty.
Then we iterate through all pairs of vectors $h^1 = (s_1^1, \dots, s_k^1, r_1^1, \dots, r_k^1, \ell^1)$ from $T_{H^1}$ and $h^2 = (s_1^2, \dots, s_k^2, r_1^2, \dots, r_k^2, \ell^2)$ from $T_{H^2}$ and for every $\ell^1 + \ell^2 \leq \ell \leq |E(H)|$, we add the vector $(s_1, \dots, s_k, r_1, \dots, r_k, \ell)$ defined as follows.
For every $i \in [k] \setminus \{i_1, \dots, i_q\}$, it holds that $s_i = s_i^1 + s_i^2$ and $r_i = r_i^1 + r_i^2$.
And for every $i \in \{i_1, \dots, i_q\}$, it holds that $s_i = s_i^1 \lor s_i^2$ and $r_i = \neg s_i^1 \land \neg s_i^2 \land (r_i^1 \lor r_i^2)$.

\begin{claim}
    The table $S_H$ contains exactly the same entries as $T_H$.
\end{claim}

\begin{claimproof}
    For the one direction, let $h^1 = (s_1^1, \dots, s_k^1,$ $r_1^1, \dots, r_k^1, \ell^1)$ be an entry of $T_{H^1}$ and $h^2 = (s_1^2, \dots, s_k^2, r_1^2, \dots, r_k^2, \ell^2)$ be an entry of $T_{H^2}$.
    So for $j \in [2]$, there exists a pair $(S^j, R^j)$ witnessing $h^j$ in $H^j$.
    Also let $\ell \in \NN_0$ be such that $\ell^1 + \ell^2 \leq \ell \leq |E(H)|$ and let $h = (s_1, \dots, s_k, r_1, \dots, r_k, \ell)$ be the entry constructed by the algorithm from $h^1$ and $h^2$.
    We now show how to construct a pair $(S, R)$ witnessing $h$ in $H$.
    First, let $S = S^1 \cup S^2$.
    Then we have
    \[
        |S| \leq |S^1| + |S^2| \leq \ell^1 + \ell^2 \leq \ell.
    \]
    Now for every $i \in [k]$, we determine the number $s_i'$ of vertices in $U^H_i = U^{H^1}_{i_1} \cup U^{H^2}_{i_2}$ incident with an edge of $S$. 
    For $i \in [k] \setminus \{i_1, \dots, i_q\}$, the sets $U^{H^1}_{i_1}$ and $U^{H^2}_{i_2}$ are disjoint so we obtain $s_i' = s^1_i + s^2_i = s_i$.
    For every $j \in \{1, \dots, q\}$, the value $s_{i_j}'$ reflects whether the vertex $v_j$ has an incident edge in $S$. 
    Similarly, the values $s_{i_j}^1$ and $s_{i_j}^2$ reflect whether $v_j$ has an incident edge in $S^1$ and $S^2$, respectively.
    Due to $S = S^1 \cup S^2$, we obtain $s_{i_j}' = s_{i_j}^1 \lor s_{i_j}^2 = s_{i_j}$.
    Altogether, we obtain $s_i = s_i'$ for every $i \in [k]$.
    
    Next we set $R = (R^1 \cup R^2) \setminus V(S)$.
    Now for every $i \in [k]$, let $r_i'$ denote the size of $R \cap U^H_i$, i.e., the number of vertices with label $i$ that have a pendant edge attached to it.
    Recall that we have $U^H_i = U^{H^1}_i \cup U^{H^2}_i$.
    First, consider $i \in [k] \setminus \{i_1, \dots, i_q\}$.
    In this case, the sets $U^{H^1}_i$ and $U^{H^2}_i$ are disjoint.
    We claim that in this case we simply have $R \cap U^H_i = (R^1 \cup R^2) \cap U^H_i$.
    Consider a vertex $v \in R^1 \cap U^{H^1}_i$.
    Since $(S^1, R^1)$ witnesses $h^1$ in $H$, the vertex $v$ has no incident edge in $S^1$ and since $v$ does not belong to $H^2$, it also has no incident edge in $S^2$.
    So $v$ has no incident edge in $S$ and therefore belongs to $R$.
    The symmetric argument shows that the vertices of $R^2 \cap U^{H^2}_i$ belong to $R$.
    So we obtain $R \cap U^H_i = (R^1 \cup R^2) \cap U^{H^1}_i$ and since the sets $R^1 \cap U^{H^1}_i$ and $R^2 \cap U^{H^2}_i$ are disjoint, we get 
    \[
        |R \cap U^H_i| = |R^1 \cap U^{H^1}_i| + |R^2 \cap U^{H^2}_i| = r^1_i + r^2_i = r_i.
    \]
    Now let $j \in [q]$.
    Recall that there exists a unique vertex $v_j \in U^H_{i_j}$.
    Also, the vertex $v_j$ is the unique vertex in the set $U^{H^1}_{i_j}$ as well as in $U^{H^2}_{i_j}$.
    By construction, this vertex belongs to $R$ if and only if it belongs to $R^1 \cup R^2$ and has no incident edge in $S$, i.e.,
    \[
        r_{i_j}' = (r^1_{i_j} \lor r^2_{i_j}) \land \neg s_{i_j} =
        (r^1_{i_j} \lor r^2_{i_j}) \land \neg (s^1_{i_j} \lor s^2_{i_j}) = 
        (r^1_{i_j} \lor r^2_{i_j}) \land \neg s^1_{i_j} \land \neg s^2_{i_j}
        = r_{i_j}.
    \]
    So we obtain $r_i' = r_i$ for every $i \in [k]$.
    
    It remains to prove that pendant edges from $R^E$ dominate all edges of $E(H)$ undominated by $S$.
    So let $e$ be an edge of $E(H)$ undominated $S$.
    Without loss of generality, assume that $e$ belongs to $H^1$.
    First, $e$ is an edge of $H^1$ undominated by $S^1$ and therefore, it has an end-point $v$ in $R^1$.
    Second, since $e$ is not dominated by $S$, in particular, the vertex $v$ has no incident edge in $S$ and therefore, by construction, the edge $e$ also belongs to $R$ as desired.
    Altogether, we have shown that $(S, R)$ witnesses $h$ in $H$ and therefore, the vector $h$ belongs to $T_H$.
    
    For the other direction, we consider a vector $h = (s_1, r_1, \dots, s_k, r_k, \ell)$ from $T_H$.
    Let $(S, R)$ be the pair witnessing $h$ in $H$.
    For $j \in [2]$, let $S^j = S \cap E(H^j)$ and let $\ell^j = |S^j|$.
    We then have $\ell^j \leq |E(H^j)|$.
    Since the graphs $H^1$ and $H^2$ are edge-disjoint, we obtain
    \[
        \ell^1 + \ell^2 = |S^1| + |S^2| = |S| \leq \ell.
    \]
    For $i \in [k]$ and $j \in [2]$, let $s_i^j = |S^j \cap U^{H^j}_i|$.
    First, let $i \in [k] \setminus \{i_1, \dots, i_q\}$, $j \in [2]$, and let $v$ be a vertex from $U^H_i \cap V(H^j)$.
    Recall that $U^{H^1}_i$ and $U^{H^2}_i$ are disjoint.
    Then by construction, the vertex $v$ has an incident edge in $S$ if and only if it has one in $S^j$.
    This, together with, again, the disjointness of $U^{H^1}_i$ and $U^{H^2}_i$ implies that $s_i = s_i^1 + s_i^2$ holds.
    Now let $j \in [q]$.
    Then the unique vertex $v_j \in U^H_{i_j}$ has an incident edge in $S$ if and only if it has an incident edge in $S^1$ or in $S^2$, i.e., we have $s_{i_j} = s^1_{i_j} \lor s^2_{i_j}$.
    
    Now we construct the sets $R^1$ and $R^2$ from $S$ and $R$ as follows.
    For $j \in [2]$, we set
    \[
        R^j = \bigl(R \cap V(H^j)\bigr) \cup \{v_p \mid p \in [q], s_{i_p} = 1, s_{i_p}^j = 0\}.
    \]
    And for $i \in [k]$ and $j \in [2]$, we set $r^j_i = R^j \cap U^{H^j}_i$.
    Observe that for $i \in \{i_1, \dots, i_q\}$ and $j \in [2]$, we then have 
    \begin{equation}\label{eq:property-set-R-j-glue}
        r^j_i = r_i \lor (s_i \land \neg s_i^j)
    \end{equation}
    and for $i \in [k] \setminus \{i_1, \dots, i_1\}$, we have \begin{equation}\label{eq:property-set-R-j-not-glue}
        r_i = r^1_i + r^2_i
    \end{equation}
    since $U^{H^1}_i$ and $U^{H^2}_i$ are disjoint.
    
    Let now $j \in [2]$ be arbitrary but fixed.
    We show that the edges of $S^j$ together with pendant edges at vertices in $R^j$ dominate all edges of $H^j$.
    So let $e$ be an edge of $H^j$. 
    Since $e$ is also an edge of $H$, it is dominated by $S \cup E^R$.
    So there exists an end-vertex $u$ of $e$ such that $u$ is incident with an edge in $S$ or $u$ belongs to $R$.
    If $u$ belongs to $R$, then by construction $u$ also belongs to $R^j$ and so $e$ is dominated by a pendant edge from $E^{R^j}$.
    So we now may assume that $u$ does not belong to $R$ and it has an incident edge $e'$ of $H$ in $S$.
    First, assume that we have $u \notin \{v_1, \dots, v_q\}$.
    Since $u$ is not a glue-vertex, any edge of $H$ incident with $u$ must be an edge of $H^j$, i.e., we have $e' \in S \cap E(H^j) = S^j$.
    So $e$ is dominated by $S^j$.
    Now we may assume that $u \in \{v_1, \dots, v_q\}$ holds and let $p \in [q]$ be such that $u = v_p$.
    Suppose $e$ is not dominated by $S^j$ and the edges pendant at $R^j$.
    In particular, it implies that $u$ does not belong to $R^j$.
    Since $u = v_p$ is the unique vertex in $U^{H^j}_{i_p}$, we have $s^j_{i_p} = 0$ and $r^j_{i_p} = 0$.
    Recall that by the above assumption, the vertex $u$ has an incident edge in $S$, i.e., $s_{i_p} = 1$.
    But this contradicts the equality \cref{eq:property-set-R-j-glue}.
    Thus, the edge $e$ is dominated by $S^j \cup E^{R^j}$.
    Since $e$ was chosen arbitrarily, this holds for any edge of $H^j$.
    Altogether, we have shown that $(S^j, R^j)$ witnesses $h^j$ in $H^j$ and therefore, $h^j$ is an entry of $T_{H^j}$.
    
    It remains to show that at the iteration corresponding to $h^1$ and $h^2$, the algorithm adds $h$ to $S_H$.
    So let $(s_1', \dots, s_k', r_1', \dots, r_k', \ell')$ be an entry added to $S_H$ such that $\ell' = \ell$ holds.
    Above we have shown that $\ell^1 + \ell^2 \leq \ell$ holds so such an entry indeed exists.
    Also, we have already shown that $s_i = s_i^1 + s_i^2$ for any $i \in [k] \setminus \{i_1, \dots, i_q\}$ and $s_i = s_i^1 \lor s_i^2$ for any $i \in \{i_1, \dots, i_q\}$.
    So it holds that $s_i = s_i'$ for any $i \in [k]$.
    It remains to show that $r_i = r_i'$ holds for any $i \in [k]$ as well.
    Recall that by the construction of the algorithm, we have $r_i' = r_i^1 + r_i^2$ for any $i \in [k] \setminus \{i_1, \dots, i_q\}$.
    The equality \eqref{eq:property-set-R-j-not-glue} then implies that $r_i = r_i'$ holds.
    For $i \in \{i_1, \dots, i_q\}$ the algorithm sets $r_i' = \neg s_i^1 \land \neg s_i^2 \land (r_i^1 \lor r_i^2)$.
    We then obtain
    \begin{align*}
        r_i' = &\neg s_i^1 \land \neg s_i^2 \land (r_i^1 \lor r_i^2) \stackrel{\eqref{eq:property-set-R-j-glue}}{=} \\
        &\neg s_i^1 \land \neg s_i^2 \land \Bigl(\bigl(r_i \lor [s_i \land \neg s_i^1]\bigr) \lor \bigl(r_i \lor [s_i \land \neg s_i^2]\bigr)\Bigr) \stackrel{s_i = s_i^1 \lor s_i^2}{=} \\
        &\neg s_i^1 \land \neg s_i^2 \land \Bigl(\bigl(r_i \lor [(s_i^1 \lor s_i^2) \land \neg s_i^1]\bigr) \lor \bigl(r_i \lor [(s_i^1 \lor s_i^2) \land \neg s_i^2]\bigr)\Bigr) =\\
        &\neg s_i^1 \land \neg s_i^2 \land \Bigl(\bigl(r_i \lor (s_i^2 \land \neg s_i^1)\bigr) \lor \bigl(r_i \lor (s_i^1 \land \neg s_i^2)\bigr)\Bigr) =\\
        &\neg s_i^1 \land \neg s_i^2 \land \Bigl(r_i \lor (s_i^2 \land \neg s_i^1) \lor (s_i^1 \land \neg s_i^2)\Bigr) =\\
        &\neg s_i^1 \land \neg s_i^2 \land r_i =\\
        &\neg (s_i^1 \lor s_i^2) \land r_i =\\
        &\neg s_i \land r_i \stackrel{\eqref{eq:s-i-r-i-implication}}{=}\\
        &r_i
    \end{align*}
    So we indeed obtain $r_i = r_i'$ for every $i \in [k]$.
    Therefore, at the iteration corresponding to the entries $h^1$ and $h^2$ of $T_{H^1}$ and $T_{H^2}$, respectively, the algorithm indeed adds the entry $h$ to $S_H$.
    Altogether, we obtain that $S_H = T_H$ and the provided algorithm indeed computes the table $T_H$ given the tables $T_{H^1}$ and $T_{H^2}$.
\end{claimproof}

Observe that if a graph $H$ has $n$ nodes, the table $T_H$ contains $n^{\mathcal{O}(k)}$ entries.
Therefore, this table can be computed from $T_{H^1}$ and $T_{H^2}$ in time $n^{\mathcal{O}(k)}$ as well.
This results in an algorithm that given a graph $H$ together with a reduced glue-$k$-expression $\xi$ of $H$, traverses the nodes $x$ of the expression in a standard bottom-up manner and computes the tables $T_{G^\xi_x}$ in time $n^{\mathcal{O}(k)}$.
Let $y$ denote the root of $\xi$. 
Then $G^\xi_y$ is exactly the graph $H$. 
As noted by Fomin et al., the size of the minimum edge dominating set of $H$ is the smallest integer $\ell$ such that the table $T_{G^\xi_y}$ contains an entry $(s_1, \dots, s_k, 0, \dots, 0, \ell)$ for some $s_1, \dots, s_k \in \NN_0$.
So the algorithm outputs this value. 

\begin{theorem}
    Given a fuse-$k$-expression of a graph $H$, the \textsc{Edge Dominating Set} problem can be solved in time $n^{\mathcal{O}(k)}$.
\end{theorem}

Fomin et al.\ have also shown the following lower bound:
\begin{theorem}\cite{FominGLS14}
    Let $H$ be an $n$-vertex graph given together with a $k$-expression of $H$. 
    Then the \textsc{Edge Dominating Set} problem cannot be solved in time $f(k)n^{o(k)}$ for any computable function $f$ unless the ETH fails.
\end{theorem}

Since any $k$-expression of a graph is, in particular, its fuse-$k$-expression, the lower bound transfers to fuse-$k$-expressions as well thus showing that our algorithm is tight under ETH.
\begin{theorem}
    Let $H$ be an $n$-vertex graph given together with a fuse-$k$-expression of $H$. 
    Then the \textsc{Edge Dominating Set} problem cannot be solved in time $f(k)n^{o(k)}$ for any computable function $f$ unless the ETH fails.
\end{theorem}

\input{Hamiltonian-cycle}

%% file: Hamiltonian-cycle.tex
\subsection{ Hamiltonian Cycle}\label{app:subsec:ham-cycle}

In this problem, given a graph $G = (V, E)$ we are asked about the existence of a cycle visiting each vertex exactly once.
Here we provide an algorithm solving this problem.
Similarly, to the previous two problems, we will rely on the existing algorithm for the parameterization by clique-width.
The algorithm is by Bergougnoux et al.\ and runs in time $n^{\mathcal{O}(\cw)}$~\cite{BergougnouxKK20}.
We will show how to handle glue-nodes in the same running time. 
We will follow the general idea for union-nodes from the original paper.
However, with multiple vertices being glued, the situation becomes more complicated and the proof of correctness gets more involved.

We start with some preliminary definitions, most of which were already  introduced by Bergougnoux et al~\cite{BergougnouxKK20}. 
A \emph{path packing} $\PP$ is a graph such that each of its connected components is a path.
We say that a path packing $\PP$ is a path packing in $H$ if $\PP$ is a subgraph of $H$.
A \emph{maximal path packing} of a graph $H$ is a spanning subgraph of $H$ that is a path packing.
Note that no restrictions on the length of the paths are imposed so in particular, paths consisting of a single vertex are allowed.
With a slight abuse of notation, depending on the context, we will refer to $\PP$ as a graph or as a set of paths.
If not stated otherwise, speaking about \emph{paths} in a path packing $\PP$ we always refer to its connected components, i.e., maximal paths in $\PP$.
We sometimes refer to maximal path packings of a graph as \emph{partial solutions} and we denote the set of all partial solutions of a graph $H$ by $\Pi(H)$.
With a (not necessarily maximal) path packing $\PP$ in a $k$-labeled graph $H$ we associate an auxiliary multigraph $\aux_H(\PP)$ on the vertex set $[k]$ such that for every $i \neq j \in [k]$, the multiplicity of the edge $\{i, j\}$ is equal to the number of paths in $\PP$ whose end-points have labels $i$ and $j$; and for every $i \in [k]$, the multiplicity of the loop at the vertex $i$ is equal to the number of paths whose both end-vertices have label $i$ (in particular, this contains the paths consisting of a single vertex of label $i$).
Note that this multigraph depends on the labeling of $H$.
The edges of such a multigraph will often be referred to as \emph{red}, this will allow us to speak about red-blue trails later.

In their work, Bergougnoux et al.\ use the technique of so-called \emph{representative sets}~\cite{BergougnouxKK20}. 
For two maximal path packings $\PP_1$ and $\PP_2$ of a $k$-labeled graph $H$ they write $\PP_1 \simeq_H \PP_2$ if for every $i \in [k]$, it holds that $\deg_{\aux_H(\PP_1)}(i) = \deg_{\aux_H(\PP_2)}(i)$ and the graphs $\aux_H(\PP_1)$ and $\aux_H(\PP_2)$ have the same set of connected components.
This defines an equivalence relation on $\Pi(H)$.
For a set $\AAA \subseteq \Pi(H)$ of partial solutions, the operation $\rreduce_H(\AAA)$ returns a set containing one element of each equivalence class of $\AAA / \simeq_H$. 
The following has been shown by Bergougnoux et al.~\cite{BergougnouxKK20}:
\begin{lemma}[\cite{BergougnouxKK20}]\label{lem:size-of-repr}
    For every $\AAA \subseteq \Pi(H)$, we have $|\rreduce(\AAA)| \leq n^k \cdot 2^{k(\log_2(k)+1)}$ and we can moreover compute $\rreduce_H(\AAA)$ in time $\mathcal{O}(|\AAA| \cdot n k^2 \log_2(nk))$.
\end{lemma}

In the following, we will work a lot with multigraphs on the vertex set $[k]$. 
For two such multigraphs $A$ and $B$, with $A \uplus B$ we denote the multigraph on the vertex set $[k]$ such that the multiplicity of every edge is given by the sum of multiplicities of this edge in $A$ and $B$.
As in the work of Bergougnoux et al.~\cite{BergougnouxKK20}, we say that the edges of a multigraph on the left resp.\ right of~$\uplus$ are colored red resp.\ blue.
They also use the following notion of representativity.
\begin{definition}[\cite{BergougnouxKK20}]
    Let $\AAA_1, \AAA_2 \subseteq \Pi(H)$ be two families of partial solutions of a $k$-labeled graph $H$.
    We write $\AAA_1 \lesssim_H \AAA_2$ and say that $\AAA_1$ $H$-represents $\AAA_2$ if, for every multigraph $\MM$ on the vertex set $[k]$ whose edges are colored blue, whenever there exists a path packing $\PP_2 \in \AAA_2$ such that $\aux_H(\PP_2) \uplus \MM$ admits \rbet{}, there also exists a path packing $\PP_1 \in \AAA_1$ such that $\aux_H(\PP_1) \uplus \MM$ admits a red-blue Eulerian trail, where a \emph{red-blue Eulerian trail} is a closed walk containing each edge exactly once and such that red and blue edges alternate on this walk.
\end{definition}
Crucially, they have shown the following lemma:
\begin{lemma}[\cite{BergougnouxKK20}]
    For every $\AAA \subseteq \Pi(H)$, it holds that $\rreduce_H(\AAA) \lesssim_H \AAA$.
\end{lemma}
Together with \cref{lem:size-of-repr}, this allows to keep the number of partial solutions maintained along the dynamic programming small.
Recall that we aim at handling glue-nodes.
As in standard algorithms based on representative sets, our goal is the following: given two $k$-labeled glueable edge-disjoint graphs $H_1$ and $H_2$ and families $\AAA_1 \lesssim_{H_1} \Pi(H_1)$ and $\AAA_2 \lesssim_{H_2} \Pi(H_2)$ of partial solutions of $H_1$ and $H_2$, respectively, we would like to compute a family $\AAA$ of partial solutions of $H = H_1 \sqcup H_2$ with $\AAA \lesssim_H \Pi(H)$ such that $\AAA$ has bounded size.
After that, the operation $\rreduce_H$ can be applied to $\AAA$ to obtain a smaller representative. 

Bergougnoux et al.\ have shown that for two vertex-disjoint graphs $H_1$ and $H_2$, the set of partial solutions of the graph $H_1 \oplus H_2$ can be computed by simply iterating through all partial solutions $\PP_1$ of $H_1$ and $\PP_2$ of $H_2$ and forming their union $\PP_1 \cup \PP_2$~\cite{BergougnouxKK20}.
For glue-nodes our process will be analogous but there is more interaction between partial solutions.
At a glue-node, multiple paths in partial solutions $\PP_1$ and $\PP_2$ can be glued together. 
First, this can result in longer paths that contain several paths of $\PP_1$ and $\PP_2$ as subpaths (see~\cref{fig:path-glueing}~(a)). 
But also, the glueing can create cycles (see~\cref{fig:path-glueing}~(b)) as well as vertices of degree more than two (see~\cref{fig:path-glueing}~(c)) so that the result of gluing of two partial solutions of $H_1$ and $H_2$, respectively, is not a partial solution of $H_1 \sqcup H_2$ anymore. 

\begin{figure}[b]
    \includegraphics{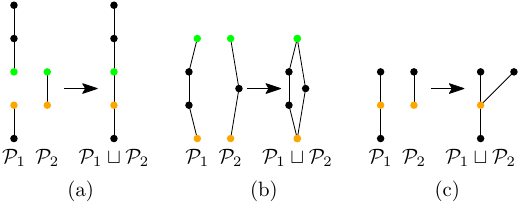}
    \centering
    \caption{Glue-vertices are depicted green and orange, the glueing of two path packings can result in (a) longer paths, (b) cycles, and (c) vertices of degree larger than two.} 
    \label{fig:path-glueing}
\end{figure}

Now we formalize this.
Along this section, let $H_1$ and $H_2$ be two $k$-labeled glueable edge-disjoint graphs.
First, we show that the set of partial solutions of $H_1 \sqcup H_2$ can be obtained in a natural way by gluing all pairs of partial solutions of $H_1$ and $H_2$ and then filtering out graphs that are not path packings.
For a family $\mathcal{H}$ of graphs, by $\filter(\mathcal{H})$ we denote the set of all path packings in~$\mathcal{H}$.
Clearly, the set $\filter(\mathcal{H})$ can be computed in time polynomial in the cardinality of $\mathcal{H}$ and the largest graph in $\mathcal{H}$.

\begin{lemma}\label{lem:all-partial-solutions}
	Let $H_1$ and $H_2$ be two edge-disjoint graphs.
	And let 
	\[
		S = \filter\bigl(\{\PP_1 \sqcup \PP_2 \mid \PP_1 \in \Pi(H_1), \PP_2 \in \Pi(H_2)\bigr).
	\]
	Then it holds that $S = \Pi(H_1 \sqcup H_2)$.
\end{lemma}

\begin{proof}
	For the one direction, let $\PP \in S$ and let $\PP_1 \in \Pi(H_1)$, $\PP_2 \in \Pi(H_2)$ be such that $\PP = \PP_1 \sqcup \PP_2$.
	First, recall that $\PP_1$ resp.\ $\PP_2$ contains all vertices of $H_1$ resp.\ $H_2$ and we have $V(H_1 \sqcup H_2) = V(H_1) \cup V(H_2)$. 
	So $\PP$ contains all vertices of $H_1 \sqcup H_2$.
	Second, since $S$ is the result of the application of the $\filter$ operator, the graph $\PP$ is a path packing. 
	Therefore, the graph $\PP$ is a maximal path packing of $H$, i.e., $\PP \in \Pi(H_1 \sqcup H_2)$, and we obtain~$S \subseteq \Pi(H_1 \sqcup H_2)$.
	
	For the other direction, consider a path packing $\PP \in \Pi(H_1 \sqcup H_2)$.
	For $i \in [2]$, let 
	\begin{align*}
		\PP_i = \{ Q \mid &\text{ $Q$ is an inclusion-maximal subpath of some path in $\PP$} \\
        &\text{ such that all edges of $Q$ belong to $H_i$} \}
    \end{align*}
	Clearly, $\PP_i$ is a subgraph of $H_i$ and it is a path packing due to being a subgraph of $\PP$.
	Each vertex $v$ of $V(H_i) \subseteq V(H_1 \sqcup H_2) = V(H_1) \cup V(H_2)$ lies on exactly one path, say $P$, in $\PP$. 
	Then there is a unique inclusion-maximal subpath $Q$ of $P$ containing $v$ that uses edges of $H_i$. By definition, the subpath $Q$ belongs to $\PP_i$.
	Therefore, the set $\PP_i$ is a maximal path packing of $H_i$, i.e., $\PP_i \in \Pi(H_i)$.
	It remains to show that $\PP_1 \sqcup \PP_2 = \PP$.
	Since $\PP_1$ and $\PP_2$ are maximal path packings of $H_1$ and $H_2$, respectively, the graph $\PP_1 \sqcup \PP_2$ contains all vertices of $H_1 \sqcup H_2$, i.e., all vertices of $\PP$.
	For $i \in [2]$, every edge of $E(H_i) \cap E(\PP)$ is contained in a unique  maximal subpath $Q$ of some path in $\PP$ such that this subpath contains the edges of $H_i$ only, i.e., $Q \in \PP_i$. 
	Therefore, we have $\PP_1 \sqcup \PP_2 = \PP$.
	And since $\PP$ is a path packing, the $\filter$ operation does not discard it.
	So we obtain $\Pi(H_1 \sqcup H_2) \subseteq S$ and this concludes the proof.
\end{proof}

As the next step, we will show that the representativity is maintained in this process, formally:

\begin{lemma} \label{lem:dp-glue-node}
	Let $H_1$ and $H_2$ be two glueable edge-disjoint $k$-labeled graphs. 
	Further, let $\AAA_1 \lesssim_{H_1} \Pi(H_1)$ and $\AAA_2 \lesssim_{H_2} \Pi(H_2)$.
	Then for the set $S$ defined by
	\[
		S = \filter(\{\PP_1 \sqcup \PP_2 \mid \PP_1 \in \AAA_1, \PP_2 \in \AAA_2\})
	\]
	it holds that 
	\[
		S \lesssim_{H_1 \sqcup H_2} \Pi(H_1 \sqcup H_2).
	\]
    Further, given $\AAA_1$ and $\AAA_2$, the set $S$ can be computed in $\ostar(|\AAA_1| |\AAA_2|)$.
\end{lemma}

This lemma will be the key component of our procedure for glue-nodes.
In the remainder of this subsection we mostly concentrate on the proof of this lemma.
It will follow the general idea behind the proof of Bergougnoux et al.\ for union-nodes~\cite{BergougnouxKK20} but the technicalities will become more involved.
We start with some simple claims. 

\begin{observation}\label{obs:aux-degree-conversion}
	Let $H$ be a $k$-labeled graph, let $i \in [k]$ be such that there exists a unique vertex $v$ of label $i$ in $H$, and let $\PP$ be a path packing in $H$ that contains $v$.
	Then for the unique path $P \in \PP$ containing $v$, it holds that:
	\begin{enumerate}
		\item If $P$ has length zero, then $\deg_{\aux_H(\PP)}(\llab(v)) = 2$ and there is a loop at $\llab(v)$ in $\aux_H$.
		\item If $P$ has non-zero length and $v$ is an end-vertex of $P$, then $\deg_{\aux_H(\PP)}(\llab(v)) = 1$.
		\item If $P$ has non-zero length and $v$ is an internal vertex of $P$, then $\deg_{\aux_H(\PP)}(\llab(v)) = 0$.
	\end{enumerate}
	In particular, we have $\deg_{\aux_H(\PP)}(\llab(v)) = 2 - \deg_{\PP}(v)$.
\end{observation}

We can apply this observation to glue-vertices as follows:

\begin{lemma}
	Let $H_1$ and $H_2$ be two $k$-labeled glueable edge-disjoint graphs and let $v \in V(H_1) \cap V(H_2)$ be a glue-vertex of label $i$ for some $i \in [k]$.  
	Further let $\PP_1$ and $\PP_2$ be path packings in $H_1$ and $H_2$, respectively, both containing $v$ such that the graph $\PP_1 \sqcup \PP_2$ is a path packing.
	Then it holds that 
	\[
		\deg_{\aux_{H_1 \sqcup H_2}(\PP_1 \sqcup \PP_2)}(\llab(v)) = \deg_{\aux_{H_1}(\PP_1)}(\llab(v)) + \deg_{\aux_{H_2}(\PP_2)}(\llab(v)) - 2.
	\]
\end{lemma}

\begin{proof}
	The vertex $v$ is the unique vertex of label $i$ in $H_1 \sqcup H_2$.
	Then we have
	\begin{align*}
		&\deg_{\aux_{H_1 \sqcup H_2}(\PP_1 \sqcup \PP_2)}(\llab(v)) \stackrel{\text{\cref{obs:aux-degree-conversion}}}{=} \\
		&2 - \deg_{\PP_1 \sqcup \PP_2}(v) \stackrel{E(H_1) \cap E(H_2) = \emptyset}{=} \\
		&2 - (\deg_{\PP_1}(v) + \deg_{\PP_2}(v)) \stackrel{\text{\cref{obs:aux-degree-conversion}}}{=} \\
		&2 - \Bigl(\bigl(2 - \deg_{\aux_{H_1}(\PP_1)}(\llab(v))\bigr) + \bigl(2 - \deg_{\aux_{H_2}(\PP_2)}(\llab(v))\bigr)\Bigr) = \\
		& \deg_{\aux_{H_1}(\PP_1)}(\llab(v)) + \deg_{\aux_{H_2}(\PP_2)}(\llab(v)) - 2.\qedhere
	\end{align*} 
\end{proof}

\begin{observation}\label{obs:replacement-property}
    Let $H$ be a $k$-labeled graph and let $\PP$ and $\PP'$ be path packings in $H$ with $V(\PP) = V(\PP')$.
    Further, let $\MM$ be a multigraph on the vertex set $[k]$ such that each of the graphs $\aux_H(\PP) \uplus \MM$ and $\aux_H(\PP') \uplus \MM$ admits \rbet{}. 
    Finally, let $v \in V(\PP)$ be a vertex of unique label in $H$, i.e., $|\lab_H^{-1}(\lab_H(v))| = 1$.
    Then the graphs $\aux_H(\PP)$ and $\aux_H(\PP')$ have the same degree sequence and in particular, the following properties hold:
    \begin{itemize}
        \item The vertex $v$ is an internal vertex of a path in $\PP$ iff $v$ is an internal vertex of a path in~$\PP'$.
        \item The vertex $v$ is an end-vertex of a non-zero length path in $\PP$ iff $v$ is an end-vertex of a non-zero length path in $\PP'$.
        \item The vertex $v$ forms a zero-lentgh path in $\PP$ iff $v$ forms a zero-lentgh path in $\PP'$.
    \end{itemize}
\end{observation}

\begin{proof}
    Since each of the graphs $\aux_H(\PP) \uplus \MM$ and $\aux_H(\PP') \uplus \MM$ admits \rbet{}, for every $i \in [k]$ we have
    \[
    	\deg_{\aux_H(\PP)}(i) = \deg_{\MM}(i) = \deg_{\aux_H(\PP')}(i)
    \]
    by a result of Kotzig~\cite{Kotzig68}.
    Therefore, the graphs $\aux_H(\PP)$ and $\aux_H(\PP')$ have the same degree sequence. 
    The remainder of the claim follows by \cref{obs:aux-degree-conversion}.
\end{proof}

With these technical lemmas in hand, we can now prove \cref{lem:dp-glue-node}.
In the proof we will work a lot with multigraphs on vertex set $[k]$.
For $i \in [k]$, by $\lloop_i$ we will denote a loop at the vertex $i$.
Similarly, for $i, j \in [k]$, by $e_{i, j}$ we denote an edge with end-points $i$ and $j$ where $i = j$ is allowed.
This edge is not necessarily unique so with a slight abuse of notation, this way we denote one fixed edge clear from the context between these vertices.
If $\mathcal{H}$ is a multigraph and $e = e_{i, j}$ resp.\ $e = \lloop_i$, by $\mathcal{H} \dot\cup \{e\}$ we denote the multigraph arising from $\mathcal{H}$ by adding an edge with end-points $i$ and $j$ resp.\ adding a loop at $i$.
Here, $\dot\cup$ emphasizes that we add a \emph{new} edge and in particular, increase the number of edges in the multigraph.
Similarly, by $\mathcal{H} - e$ we denote the multigraph arising from $\mathcal{H}$ after removing the edge $e$ and emphasize that $e$ was present in $\mathcal{H}$. 

\begin{proof}
    For simplicity, we denote $D = H_1 \sqcup H_2$. 
    Along this proof no relabeling occurs so every vertex of $H_1$ resp. $H_2$ has the same label in $H$. For this reason we omit the subscripts of labeling functions to improve readability: the label of a vertex $v$ is simply denoted by $\lab(v)$.
    
    Now let $\PP \in \Pi(D)$ be a maximal path packing  of $D$.
    By \cref{lem:all-partial-solutions}, there exist maximal path packings $\PP_1 \in \Pi(H_1)$ and $\PP_2 \in \Pi(H_2)$ such that $\PP = \PP_1 \sqcup \PP_2$.
    Further, let $\MM$ be a blue multigraph on the vertex set $[k]$ such that $\aux_H(\PP) \uplus \MM$ admits \rbet{}, say $T$.
    To prove the lemma, we need to show that there exists a maximal path packing $\PP' \in S = \filter(\{\PP_1' \sqcup \PP_2' \mid \PP_1' \in \AAA_1, \PP_2' \in \AAA_2\})$ of $D$ such that $\aux_H(\PP') \uplus \MM$ admits a red-blue Eulerian trail as well, i.e., there exist maximal path packings $\PP'_1 \in \AAA_1$ and $\PP'_2 \in \AAA_2$ such that $\PP'_1 \sqcup \PP'_2$ is a path packing and $\aux_H(\PP_1' \sqcup \PP'_2) \uplus \MM$ admits a red-blue Eulerian trail.

    Let $t = |\PP_2|$ and let us fix some ordering $\PP_2 = \{P^1, \dots, P^t\}$ of the paths (i.e., connected components) in $\PP_2$.
    For $i \in [t]_0$, we define 
    \[
        \RR^i = \PP_2 \setminus \{P^1, \dots, P^i\}.
    \]
    Now we will construct blue multigraphs $\MM^0, \dots, \MM^t$ with the following properties.
    \begin{claim}\label{claim:blue-multigraphs-sequence}
    There exist blue multigraphs $\MM^0, \dots, \MM^t$ such that for every $i \in [t]_0$, the following two properties hold. 
    First, the multigraph $\aux_D(\PP_1 \sqcup \RR^i) \uplus \MM^i$ admits \rbet{}, we fix one and denote it by $T^i$.
    And second, if $i > 0$ and $\PP_1'$ is a maximal path packing of $H_1$ such that $\aux_D(\PP_1' \sqcup \RR^i) \uplus \MM^i$ admits \rbet{}, then $\aux_D(\PP_1' \sqcup \RR^{i-1}) \uplus \MM^{i-1}$ also admits \rbet{}.
    \end{claim}

    Along the proof of the claim, we will inter alia show that if the graph $\PP_1' \sqcup \RR^i$ (as above in the claim) is a path packing, then $\PP_1' \sqcup \RR^{i-1}$ is also a path packing, i.e., $\aux_D(\PP_1' \sqcup \RR^{i-1})$ is well-defined.
    Recall that $\PP_1 \sqcup \RR^{i-1}$ (a subgraph of $\PP_1 \sqcup \PP_2$) is a path packing and by \cref{obs:replacement-property}, $\aux_D(\PP_1 \sqcup \RR^{i-1})$ has the same degree sequence as $\aux_D(\PP_1' \sqcup \RR^{i-1})$.
    So to prove that $\PP_1' \sqcup \RR^{i-1}$ is a path packing, it will suffice to show its acyclicity.
    
    \begin{claimproof}
    The proof is by induction. 
    Base case $i = 0$:
    Since $\RR^0 = \PP_2$, we can use $\MM^0 = \MM$ and the statement is true by using $T^0 = T$.
    So now let $1 \leq i \leq t$ and suppose the statement holds for $i - 1$.
    We make a case distinction based on the path $P := P^i$.
    For simplicity of notation, after we construct $\MM^i$ in each case, with $\PP_1'$ we refer to any maximal path packing of $H_1$ such that $\aux_D(\PP_1' \sqcup \RR^i) \uplus \MM^i$ admits \rbet{}. 
    We emphasize that in every case, the construction of $\MM^i$ and $T^i$ will be independent of $\PP_1'$.
    The cases 1.2 and 2 will be similar to a union-node as handled by Bergougnoux et al.~\cite{BergougnouxKK20} while the remaining cases are different.

    \textbf{Case 1.1}
    First, suppose that the path $P$ has zero length and the unique vertex, say $v$, of $P$ is a glue-vertex. 
    Since $\PP_1$ and $\PP_1'$ are maximal path packings of $H_1$, both of them contain $v$ on some path.
    So in this case, we simply have 
    \begin{equation*}
        \begin{aligned}
        &\PP_1 \sqcup \RR^i = \PP_1 \sqcup (\RR^{i-1} - P) = \PP_1 \sqcup \RR^{i - 1} \text{ and } \\
        &\PP_1' \sqcup \RR^i = \PP'_1 \sqcup (\RR^{i-1} - P) = \PP_1' \sqcup \RR^{i - 1}
        \end{aligned}
    \end{equation*}
    (here and in the analogous equalities we treat a path packing as a set of vertex-disjoint paths and the operation $- P$ removes the path $P$ from it in terms of set difference).
    Therefore, $\MM^i := \MM^{i-1}$ and $T^i := T^{i-1}$ satisfy the desired condition.

    \textbf{Case 1.2}
    Now again suppose that $P$ is a zero-length path but now the unique vertex, say $v$, of $P$ is not a glue-vertex.
    Then we have 
    \begin{equation*}
        \begin{aligned}
        &\PP_1 \sqcup \RR^i = \PP_1 \sqcup (\RR^{i-1} - P) = (\PP_1 \sqcup \RR^{i-1}) - P \text{ and } \\
        &\PP_1' \sqcup \RR^i = \PP_1' \sqcup (\RR^{i-1} - P) = (\PP_1' \sqcup \RR^{i-1}) - P
        \end{aligned}
    \end{equation*}
    So
    \begin{equation}\label{eq:case-1-2-aux}
        \begin{aligned}
            &\aux_D(\PP_1 \sqcup \RR^i) = \aux_D(\PP_1 \sqcup \RR^{i-1}) - \lloop_{\llab(v)} \text{ and } \\
            &\aux_D(\PP_1' \sqcup \RR^i) = \aux_D(\PP_1' \sqcup \RR^{i-1}) - \lloop_{\llab(v)}.
        \end{aligned}
    \end{equation}

    \begin{figure}
    \includegraphics{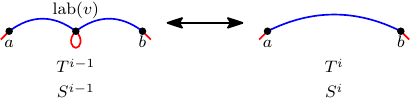}
    \centering
    \caption{Constructing the red-blue Eulerian trails in \textbf{Case 1.2}.} 
    \label{fig:ham-case-1-2}
    \end{figure}

    Since $T^{i-1}$ is \rbet{} of $\aux_D(\PP_1 \sqcup \RR^{i-1}) \uplus \MM^{i-1}$, the loop $\lloop_{\llab(v)}$ occurs on it. 
    So let $e^1$ resp.\ $e^2$ be the blue edge in $\MM^{i-1}$ preceding resp.\ following $\lloop_{\llab(v)}$ in $T^{i-1}$.
    And let $a$ resp.\ $b$ be the label such that $a$ and $\llab(v)$ resp.\ $b$ and $\llab(v)$ are the end-points of $e^1$ resp.\ $e^2$. 
    We then define 
    \begin{equation}\label{eq:case-1-2-blue}
        \MM^i = (\MM^{i-1} - e^1 - e^2) \dot\cup \{e_{a, b}\}.
    \end{equation}
    Now on the one hand, we can easily obtain \rbet{} $T^i$ of $\aux_D(\PP_1 \sqcup \RR^i) \uplus \MM^i$ by taking $T^{i-1}$ and replacing the subtrail $\textcolor{blue}{e^1}, \textcolor{red}{\lloop_{\llab(v)}}, \textcolor{blue}{e^2}$ with a blue edge $\textcolor{blue}{e_{a, b}}$ in it~(see \cref{fig:ham-case-1-2}): by \eqref{eq:case-1-2-aux} and \eqref{eq:case-1-2-blue}, each edge is indeed visited exactly once.
    For the other direction, let $S^i$ be \rbet{} of $\aux_D(\PP_1' \sqcup \RR^i) \uplus \MM^i$.
    Then we can obtain \rbet{} of $\aux_D(\PP_1' \sqcup \RR^{i-1}) \uplus \MM^{i-1}$ by taking $S^i$ and replacing the occurrence of the blue edge $\textcolor{blue}{e_{a, b}}$ by a subtrail $\textcolor{blue}{e^1}, \textcolor{red}{\lloop_{\llab(v)}}, \textcolor{blue}{e^2}$~(see \cref{fig:ham-case-1-2}): By \eqref{eq:case-1-2-aux} and \eqref{eq:case-1-2-blue}, each edge is again visited exactly once.
    This concludes the proof for \textbf{Case 1.2}.

    Now it remains to prove the claim for the case that $P$ has non-zero length.
    We further distinguish on whether $P$ contains glue-vertices and if so, how many of them are end-vertices of $P$.
    Let $v$ and $w$ denote the end-vertices of $P$.
    
    \textbf{Case 2}
    Suppose the path $P$ does not contain glue-vertices.
    This case is very similar to \textbf{Case 1.2} but we provide the proof for completeness.
    It again holds that
    \begin{equation*}
        \begin{aligned}
        &\PP_1 \sqcup \RR^i = \PP_1 \sqcup (\RR^{i-1} - P) = (\PP_1 \sqcup \RR^{i-1}) - P \text{ and } \\
        &\PP_1' \sqcup \RR^i = \PP_1' \sqcup (\RR^{i-1} - P) = (\PP_1' \sqcup \RR^{i-1}) - P.
        \end{aligned}
    \end{equation*}
    So
    \begin{equation}\label{eq:case-2-aux}
        \begin{aligned}
        &\aux_D(\PP_1 \sqcup \RR^i) = \aux_D(\PP_1 \sqcup \RR^{i-1}) - e_{\llab(v), \llab(w)} \text{ and } \\
        &\aux_D(\PP_1' \sqcup \RR^i) = \aux_D(\PP_1' \sqcup \RR^{i-1}) - e_{\llab(v), \llab(w)}.
        \end{aligned}
    \end{equation}
    Without loss of generality, we may assume that in $T^{i-1}$ the edge $e_{\llab(v), \llab(w)}$ is traversed from $\llab(v)$ to $\llab(w)$: otherwise, we can use the reverse of $T^{i-1}$ instead.
    Let again $e^1$ resp.\ $e^2$ be the blue edge in $\MM^{i-1}$ preceding resp.\ following the red edge $e_{\llab(v), \llab(w)}$ in $T^{i-1}$.
    And let $a$ resp.\ $b$ be the labels such that $e^1$ resp.\ $e^2$ has end-vertices $a$ and $\llab(v)$ resp.\ $b$ and $\llab(w)$.
    Then we define
    \begin{equation}\label{eq:case-2-blue}
        \MM^i = (\MM^{i-1} - e^1 - e^2) \dot\cup \{e_{a, b}\}.
    \end{equation}

    \begin{figure}[t]
        \includegraphics{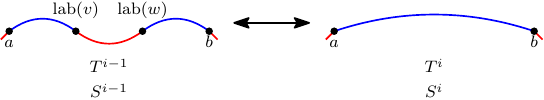}
        \centering
        \caption{Constructing the red-blue Eulerian trails in \textbf{Case 2}.} 
        \label{fig:ham-case-2}
    \end{figure}
    
    Now on the one hand, we can easily obtain \rbet{} $T^i$ of $\aux_D(\PP_1 \sqcup \RR^i) \uplus \MM^i$ by taking $T^{i-1}$ and replacing the subtrail $\textcolor{blue}{e^1}, \textcolor{red}{e_{\llab(v), \llab(w)}}, \textcolor{blue}{e^2}$ with a blue edge $\textcolor{blue}{e_{a, b}}$ in it~(see \cref{fig:ham-case-2}): 
    By \eqref{eq:case-2-aux} and \eqref{eq:case-2-blue}, each edge is indeed visited exactly once.
    For the other direction, let $S^i$ be \rbet{} of $\aux_D(\PP_1' \sqcup \RR^i) \uplus \MM^i$.
    Then we can obtain \rbet{} of $\aux_D(\PP_1' \sqcup \RR^{i-1}) \uplus \MM^{i-1}$ by replacing the occurrence of the blue edge $\textcolor{blue}{e_{a, b}}$ in $S^i$ by a subtrail $\textcolor{blue}{e^1}, \textcolor{red}{e_{\llab(v), \llab(w)}}, \textcolor{blue}{e^2}$~(see \cref{fig:ham-case-2}): By \eqref{eq:case-2-aux} and \eqref{eq:case-2-blue}, each edge is again visited exactly once and the claim holds.
    This concludes the proof for the \textbf{Case 2}.
    
    From now on, me may assume that $P$ has non-zero length and contains at least one glue-vertex.
    Let $u_1, \dots, u_q$ denote all internal glue-vertices of $P$ for some $q \in \NN_0$.
    Since these vertices are internal, we have $\deg_{\PP_2}(u_j) = 2$ for every $j \in [q]$.
    Since $\PP_1 \sqcup \PP_2$ is a maximal path packing of $D$, the path packing $\PP_1$ resp.\ $\PP_2$ is maximal for $H_1$ resp.\ $H_2$, and the graphs $H_1$ and $H_2$ are edge-disjoint, we have
    \[
        2 \geq \deg_{\PP_1 \sqcup \PP_2}(u_j) = \deg_{\PP_1}(u_j) + \deg_{\PP_2}(u_j) = \deg_{\PP_1}(u_j) + 2.
    \]
    Therefore, we have $\deg_{\PP_1}(u_j) = 0$ i.e., $u_j$ forms a zero-length path in $\PP_1$.
    Since the path packing $\RR^i$ does not contain $u_j$ by construction, the vertex $u_j$ also forms a zero-length path in $\PP_1 \sqcup \RR^{i}$.
    Next recall that $u_j$ is a glue-vertex so it is a unique vertex with label $\llab(u_j)$ in $H_1$ and $H_2$.
    Then if there exists a blue multigraph $\MM$ such that each of $\aux_D(\PP_1 \sqcup \RR^{i}) \uplus \MM$ and $\aux_D(\PP_1' \sqcup \RR^{i}) \uplus \MM$ admits \rbet{}, then by \cref{obs:replacement-property}, the vertex $u_i$ also forms a path in $\PP_1' \sqcup \RR^{i}$ (and hence, also in $\PP_1'$).
    We will apply this property to $\MM = \MM^i$ to prove the latter direction of \cref{claim:blue-multigraphs-sequence}.
    
    As mentioned before, the previous two cases were very similar to the correctness of the procedure for a union-node in the clique-width algorithm (see \cite{BergougnouxKK20}).
    In the remaining cases, the approach will be more involved but still natural.
    In the remainder of the paragraph we try to sketch this process and the main difficulty of these cases. 
    The details will become clear in the description of the remaining cases though.
    As before, we will replace a certain subtrail $A$ of \rbet{} $T^{i-1}$ of $\aux_D(\PP_1 \sqcup \RR^{i-1}) \uplus \MM^{i-1}$ by a sequence $B$ of edges to obtain \rbet{} $T^i$ of $\aux_D(\PP_1 \sqcup \RR^i) \uplus \MM^i$.
    In the previous cases, the sequence $B$ consisted of a single edge.
    So when we then considered \rbet{} $S^i$ of $\aux_D(\PP'_1 \sqcup \RR^i) \uplus \MM^i$, it was \emph{easy} to replace this edge by $A$ again to obtain \rbet{} $S^{i-1}$ of $\aux_D(\PP'_1 \sqcup \RR^{i-1}) \uplus \MM^{i-1}$.
    In the following cases, the situation might become less simple since $B$ possibly consists of multiple edges there.
    For this reason, first, $B$ does not necessarily occur as a subtrail in $S^i$, i.e., the edges of $B$ do not necessarily occur consecutively.
    And second, some of the edges of $B$ possibly do not even occur in $\aux_D(\PP'_1 \sqcup \RR^i)$.
    Therefore, the construction of $S^{i-1}$ from $S^i$ is less straight-forward in these cases.
    Therefore, in general it is not possible to simply replace $B$ with $A$ to obtain $S^{i-1}$.
    For this reason, a more careful analysis is required to show that such a \rbet{} $S^{i-1}$ of $\aux_D(\PP_1' \sqcup \RR^{i-1})$ still exists.
    Now we move on to details.
    First, observe that if at least one of the vertices $v$ and $w$ is not a glue-vertex, the acyclicity of $\PP_1' \sqcup \RR^i$ also implies the acyclicity of $\PP_1' \sqcup \RR^{i-1}$ so $\PP_1' \sqcup \RR^{i-1}$ is a path packing.
    We will come back to this issue in \textbf{Case 3.3} where both $v$ and $w$ are glue-vertices. 

    \textbf{Case 3.1}
    First, assume that neither $v$ nor $w$ is a glue-vertex.
    Then $P$ is a path in $\PP_1 \sqcup \RR^{i-1}$.
    Since we know by assumption that $P$ contains a glue-vertex, we have $q > 0$ and it holds that
    \[
        \PP_1 \sqcup \RR^i = \left((\PP_1 \sqcup \RR^{i-1}) - P\right) \dot\cup \left\{\{u_1\}, \dots, \{u_q\}\right\}. 
    \]
    And hence
    \begin{equation}\label{eq:case-3-1-aux}
        \begin{aligned}
            \aux_D(\PP^1 \sqcup \RR^i) = &\left(\aux_D(\PP_1 \sqcup \RR^{i-1}) - e_{\llab(v), \llab(w)}\right) \dot\cup \\
            &\left\{\lloop_{\llab(u_1)}, \dots, \lloop_{\llab(u_q)}\right\}.
        \end{aligned}
    \end{equation}
    As before, we may assume that the edge $e_{\llab(v), \llab(w)}$ is traversed from $\llab(v)$ to $\llab(w)$ in $T^{i-1}$.
    So let again $e^1$ resp.\ $e^2$ be the blue edge in $\MM^{i-1}$ preceding resp.\ following $e_{\llab(v), \llab(w)}$ in $T^{i-1}$.
    And let $a$ resp.\ $b$ be the label such that the end-vertices of $e^1$ resp.\ $e^2$ are $a$ and $\llab(v)$ resp.\ $b$ and $\llab(w)$.
    We then set
    \begin{equation}\label{eq:case-3-1-blue}
        \MM^i = (\MM^{i-1} - e^1 - e^2) \dot\cup \left\{e_{a, \llab(u_1)}, e_{\llab(u_q), b}\right\} \dot\cup \left\{e_{\llab(u_d), \llab(u_{d+1})} \mid d \in [q-1]\right\}.
    \end{equation}
    Then we can obtain \rbet{} $T^i$ of $\aux_D(\PP_1 \sqcup \RR^i) \uplus \MM^i$ by taking $T^{i-1}$ and replacing the subtrail $\textcolor{blue}{e^1}, \textcolor{red}{e_{\llab(v), \llab(w)}}, \textcolor{blue}{e^2}$ with the sequence $L$ given by
    \begin{equation*}
        \begin{aligned}
        L = &\textcolor{blue}{e_{a, \llab(u_1)}}, &\textcolor{red}{\lloop_{\llab(u_1)}}, \\ 
        &\textcolor{blue}{e_{\llab(u_1), \llab(u_2)}}, &\textcolor{red}{\lloop_{\llab(u_2)}}, \\
        &\dots, & \\
        &\textcolor{blue}{e_{\llab(u_{q-1}), \llab(u_q)}} &\textcolor{red}{\lloop_{\llab(u_q)}}, \\
        &\textcolor{blue}{e_{\llab(u_q), b}}&
        \end{aligned}
    \end{equation*}
    (see \cref{fig:ham-case-3-1}).
    Note that by \eqref{eq:case-3-1-aux} and \eqref{eq:case-3-1-blue}, the trail $T^i$ indeed contains all edges.

    \begin{figure}[b]
        \includegraphics{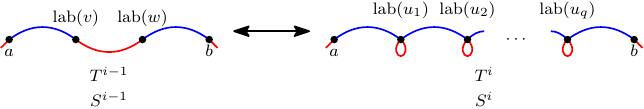}
        \centering
        \caption{Constructing the red-blue Eulerian trails in \textbf{Case 3.1}.} 
        \label{fig:ham-case-3-1}
    \end{figure}
    
    On the other hand, recall two properties. 
    First, every vertex in $u_1, \dots, u_q$ forms a zero-length path in $\PP_1'$ (argued above). 
    And second $u_1, \dots, u_q$ is the set of glue-vertices on $P$.
    Thus, $P$ forms a path in $\PP_1' \sqcup \RR^{i-1}$ and we have
    \[
        \PP_1' \sqcup \RR^i = \left((\PP_1' \sqcup \RR^{i-1}) - P\right) \dot\cup \left\{\{u_1\}, \dots, \{u_q\}\right\}.
    \]
    Therefore
    \begin{equation}\label{eq:case-3-1-aux-prime}
    \begin{aligned}
        \aux_D(\PP_1' \sqcup \RR^i) = &\left(\aux_D(\PP_1' \sqcup \RR^{i-1}\right) -
        e_{\llab(v), \llab(w)}) \dot\cup \\
        &\left\{\lloop_{\llab(u_1)}, \dots, \lloop_{\llab(u_q)}\right\}.
    \end{aligned}
    \end{equation}
    Now consider \rbet{} $S^i$ of $\aux_D(\PP'_1 \sqcup \RR^i) \uplus \MM^i$.
    We make the following observation.
    For every $j \in [q]$, since the only red edge incident with $\lab(u_j)$ in $\aux_D(\PP_1' \sqcup \RR^i)$ is a loop, there are exactly two blue edges, say $f^1$ and $f^2$, incident with $\lab(u_j)$ in $\MM^i$, namely the two edges from the set 
    \[
        \left\{e_{i, \llab(u_1)}, e_{\llab(u_q), j}\right\} \cup \left\{e_{\llab(u_d), \llab(u_{d+1})} \mid d \in [q-1]\right\}
    \]
    that are incident with $\llab(u_j)$.
    This implies that $f^1$, $\lloop_{\llab(u_j)}$, and $f^2$ appear consecutively in~$S^i$.
    Since this holds for every $j \in [q]$, the sequence $L$ or its reverse is a subtrail of~$S^i$.
    As before, we may assume that $L$ is a subtrail of $S^i$.
    Now we can obtain \rbet{} of $\aux_D(\PP'_1 \sqcup \RR^{i-1}) \uplus \MM^i$ by taking $S^i$ and replacing $L$ with $\textcolor{blue}{e^1}, \textcolor{red}{e_{\llab(v), \llab(w)}}, \textcolor{blue}{e^2}$~(see \cref{fig:ham-case-3-1}): 
    By \eqref{eq:case-3-1-blue} and \eqref{eq:case-3-1-aux-prime}, it indeed contains all edges of $\aux_D(\PP_1' \sqcup \RR^{i-1}) \uplus \MM^i$.
    In this case, for both directions of the claim we still were able to replace two subtrails $\textcolor{blue}{e^1}, \textcolor{red}{e_{\llab(v), \llab(w)}}, \textcolor{blue}{e^2}$ and $L$ with each other.
    In the remaining two cases, this will not be true anymore so we will work with different pairs of subtrails in the ``forward'' and ``backward'' directions of the proof (the details follow).

    For the remainder of the proof we may assume that at least one end-vertex of $P$ is a glue-vertex, say $v$.
    Observe the following: since $v$ is an end-vertex of a path $P$ of non-zero length, its degree in $\PP_2$ is exactly one.
    The vertex $v$ is a glue-vertex so by \cref{obs:aux-degree-conversion} the degree of $\llab(v)$ in $\aux_D(\PP_2)$ is one as well.
    Since $\PP_1 \sqcup \PP_2$ is a path packing containing $v$, we have $\deg_{\PP_1 \sqcup \PP_2}(v) \leq 2$.
    Recall that $\PP_1$ is maximal in $H_1$ so it contains $v$. 
    Together with $\deg_{\PP_2}(v) = 1$ this implies $\deg_{\PP_1}(v) \leq 1$.
    Thus, the vertex $v$ is an end-vertex of some path in the path packing $\PP_1$.
    So this also holds for the family $\PP_1 \sqcup \RR^i$.
    By \cref{obs:replacement-property}, this also holds for $\PP_1' \sqcup \RR^i$ and therefore, also for $\PP_1'$.
    This observation also implies that the path in $\PP_1 \sqcup \RR^i$ with end-vertex $v$ has non-zero length if and only if this holds for $\PP_1 \sqcup \RR^i$.
    If $w$ is a glue-vertex as well, the symmetric property holds.
    
    \textbf{Case 3.2}
    In this case, we assume that $w$ is not a glue-vertex.
    Let $\hat{P}$ be the path in $\PP_1 \sqcup \RR^{i-1}$ containing $P$ as a subpath.
    Then $w$ is an end-vertex of $\hat{P}$.
    Let $r \in [k]$ denote the label of the other end-vertex of $\hat{P}$.
    Note that $r = \llab(v)$ is possible.
    Our path $P$ is a suffix or a prefix of $\hat{P}$. 
    Without loss of generality, we assume that $P$ is a suffix of $\hat{P}$: otherwise, we could take the reverse of $\hat{P}$ instead.
    Let $\hat P v$ denote the prefix of $\hat P$ ending in $v$.
    Then the following holds
    \[
        \PP_1 \sqcup \RR^i = \left((\PP_1 \sqcup \RR^{i-1}) - \hat{P}\right) \dot\cup \left\{\hat{P}v, \{u_1\}, \dots, \{u_q\}\right\}
    \]
    and 
    \begin{equation}\label{eq:case-3-2-aux}
        \begin{aligned}
            \aux_D(\PP_1 \sqcup \RR^{i-1}) = &\left(\aux_D(\PP_1 \sqcup \RR^i\right) - e_{r, \llab(w)}) \dot\cup \\
            &\left\{e_{r, \llab(v)}, \lloop_{\llab(u_1)}, \dots, \lloop_{\llab(u_q)}\right\}.
        \end{aligned}
    \end{equation}
    As before, we may assume that in $T^{i-1}$ the edge $e_{r, \llab(w)}$ is traversed from $r$ to $\llab(w)$.
    So let $e$ be the blue edge following $e_{r, \llab(w)}$ in $T^{i-1}$ and let $b \in [k]$ be such that $\llab(w)$ and $b$ are the end-vertices of $e$.
    For simplicity of notation, we set $u_0 = v$.
    We define
    \begin{equation}\label{eq:case-3-2-blue}
        \begin{aligned}
            \MM^i = &(\MM^{i-1} - e) \dot\cup \\
            &\left\{ e_{\llab(u_d), \llab(u_{d+1})} \mid d \in [q - 1]_0\right\} \dot\cup \left\{e_{\llab(u_q), b}\right\}.
        \end{aligned}
    \end{equation}
    Then we can construct \rbet{} $T^i$ by taking $T^{i-1}$ and replacing the subtrail $\textcolor{red}{e_{r, \llab(w)}}, \textcolor{blue}{e}$ with a sequence $L$ where
    \begin{equation*}
    \begin{aligned}
        L = &&\textcolor{red}{e_{r, \llab(v)}},\\
        &\textcolor{blue}{e_{\llab(v), \llab(u_1)}}, &\textcolor{red}{\lloop_{\llab(u_1)}}, \\
        &\textcolor{blue}{e_{\llab(u_1),\llab(u_2)}}, &\textcolor{red}{\lloop_{\llab(u_2)}}, \\
        &\dots, \\
        &\textcolor{blue}{e_{\llab(u_{q-1}),\llab(u_q)}}, &\textcolor{red}{\lloop_{\llab(u_q)}},\\
        &\textcolor{blue}{e_{\llab(u_q), b}}
    \end{aligned}
    \end{equation*}
    (see \cref{fig:ham-case-3-2}~(a)).
    Note that $T^i$ indeed uses all edges of $\aux_D(\PP_1 \sqcup \RR^i) \uplus \MM^i$ 
    (see \eqref{eq:case-3-2-aux} and \eqref{eq:case-3-2-blue}).

    \begin{figure}[t]
    \includegraphics{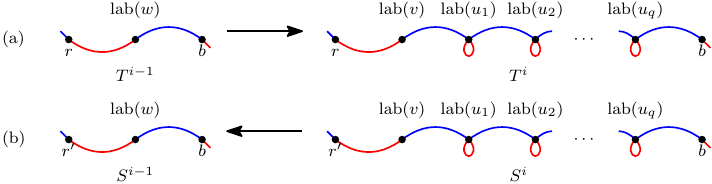}
    \centering
    \caption{(a) Constructing the red-blue Eulerian trail (a) $T^i$ from $T^{i-1}$ and (b) constructing the red-blue Eulerian trail $S^i$ from $S^{i-1}$ in \textbf{Case 3.2}.} 
    \label{fig:ham-case-3-2}
    \end{figure}

    For the other direction, we consider \rbet{} $S^i$ of $\aux_D(\PP_1' \sqcup \RR^i) \uplus \MM^i$.
    Above we have argued that $v$ is an end-vertex of some path, say $P^*$, in $\PP_1' \sqcup \RR^{i}$, and for every $j \in [q]$, the vertex $u_j$ forms a zero-length path in $\PP_1'$.
    Let $r' \in [k]$ be such that $r'$ and $\llab(v)$ are the labels of the end-vertices of $P^*$.
    Note that $r' = \llab(v)$ is possible. 
    Then it holds that
    \[
        \PP_1' \sqcup \RR^{i-1} = \left(\left(\PP_1' \sqcup \RR^i\right) - \left\{P^*, \{u_1\}, \dots, \{u_q\}\right\}\right) \dot\cup \left\{P^* \sqcup P\right\}
    \]
    and therefore, also
    \begin{equation}\label{eq:case-3-2-aux-prime}
    \begin{aligned}
        \aux_D(\PP_1' \sqcup \RR^{i-1}) = &\left(\aux_D(\PP_1' \sqcup \RR^i) - e_{r', \llab(v)} - \lloop_{\llab(u_1)} - \dots - \lloop_{\llab(u_q)}\right) \dot\cup \\ 
        &\left\{e_{r', \llab(w)}\right\}.
    \end{aligned}
    \end{equation}
    Now observe the following. 
    Since $v$ is a glue-vertex and $v \in \PP_1'$ holds, 
    there is exactly one edge, say $e^*$, incident with $\llab(v)$ in $\aux_{D}(\PP_1' \sqcup \RR^i)$.
    And this edge is a loop if and only if $P^* = \{v\}$ holds.
    So the degree of $\llab(v)$ in $\MM^i$ is one if this edge is not a loop and two if it is.
    Now we can make some observations about the occurrence of $e^*$ in $S^i$. 
    
    First, consider the case that $P^*$ consists of $v$ only.
    Above we have argued that then in $\PP_1 \sqcup \RR^i$ there is also a path consisting of $v$ only (and in particular, we have $r = \lab(v)$).
    The existence of \rbet{} implies that in this case there are exactly two blue edges in $\MM^i$ incident with $\llab(v)$ so they must precede and follow $e^*$ in $S^i$.
    Crucial is that these edges are the same (possibly their order is swapped) that follow and precede $e_{t, \llab(v)}$ in $T^i$.
    The same holds for edges incident with a vertex $u_j$ for every $j \in [q]$.
    
    Now consider the case that $P^*$ has non-zero length, i.e., $\llab(v)$ has the red degree of exactly one in $\aux_{D}(\PP_1' \sqcup \RR^i)$.
    Since $S^i$ is \rbet{}, the vertex $\llab(v)$ also has the blue degree of exactly one in $\MM^i$. 
    Now we may again assume that $e^*$ is traversed from $r'$ to $\llab(v)$ in $S^i$.
    Then the edge following $e^*$ in $S^i$ is the unique blue edge incident with $\llab(v)$ in $\MM^i$, i.e., the same blue edge follows $e_{t, \llab(v)}$ in $T^i$.

    Altogether, this implies that a sequence $L'$ is a subtrail of $S^i$ where
    \begin{equation*}
        \begin{aligned}
            L' = &&\textcolor{red}{e_{r', \llab(v)}}, \\ 
            &\textcolor{blue}{e_{\llab(v), \llab(u_1)}}, &\textcolor{red}{\lloop_{\llab(u_1)}}, \\
            &\textcolor{blue}{e_{\llab(u_1), \llab(u_2)}}, &\textcolor{red}{\lloop_{\llab(u_2)}}, \\
            & \dots \\
            &\textcolor{blue}{e_{\llab(u_{q-1}), \llab(u_q)}}, &\textcolor{red}{\lloop_{\llab(u_q)}}, \\
            &\textcolor{blue}{e_{\llab(u_q), b}}.&
        \end{aligned}
    \end{equation*}
    Hence, it can be replaced with a sequence $\textcolor{red}{e_{r', \llab(w)}}, \textcolor{blue}{e}$~(see \cref{fig:ham-case-3-2}~(b)) to obtain \rbet{} of $\aux_D(\PP_1' \sqcup \RR^{i-1}) \uplus \MM^{i-1}$ (see \eqref{eq:case-3-2-blue} and~\eqref{eq:case-3-2-aux-prime} to verify that every edge is indeed used exactly once).

    \textbf{Case 3.3}
    Now we remain with the case where both $v$ and $w$ are glue-vertices.
    For simplicity of notation, let us denote $u_0 = v$ and $u_{q+1} = w$. 
    By \textbf{Case 1.1} we may assume that $v \neq w$ holds.
    Since both $v$ and $w$ are glue-vertices, it also holds that $\llab(v) \neq \llab(w)$.
    Let $\hat{P}$ again be the path in $\PP_1 \sqcup \RR^{i-1}$ that contains $P$ as a subpath.
    We may assume that $v$ occurs on both $P$ and $\hat P$ before $w$: otherwise we may use the reverse of the violating path instead. 
    Let $r \in [k]$ resp.\ $s \in [k]$ be the label of the start- resp.\ end-vertex of $\hat P$.
    Then it holds that
    \[
        \PP_1 \sqcup \RR^i = \left((\PP_1 \sqcup \RR^{i-1}) - \hat P\right) \dot\cup \left\{\hat P v, \{u_1\}, \dots, \{u_q\}, w \hat P\right\}
    \]
    where $\hat P v$ denotes the prefix of $\hat P$ ending in $v$ and $w \hat P$ denotes the suffix of $\hat P$ starting in $w$.
    So we have
    \begin{equation}\label{eq:case-3-3-aux}
        \begin{aligned}
            \aux_D(\PP_1 \sqcup \RR^i) = &\left(\aux_D(\PP_1 \sqcup \RR^{i-1}) - \{e_{r, s}\}\right) \\
            &\dot\cup \left\{e_{r, \llab(v)}, \lloop_{\llab(u_1)}, \dots, \lloop_{\llab(u_q)}, e_{\llab(w), s}\right\}.
        \end{aligned}
    \end{equation}
    We define
    \begin{equation}\label{eq:case-3-3-blue}
        \MM^i = \MM^{i-1} \dot\cup \left\{e_{\llab(u_d) \llab(u_{d+1})} \mid d \in [q]_0 \right\}.
    \end{equation}
    Then we can obtain \rbet{} $T^i$ by taking $T^{i-1}$ and replacing the red edge $\textcolor{red}{e_{r, s}}$ with a subtrail $L$ where 
    \begin{equation*}
        \begin{aligned}
            L = &&\textcolor{red}{e_{r, \llab(v)}},  \\
            &\textcolor{blue}{e_{\llab(v), \llab(q_1)}}, &\textcolor{red}{\lloop_{\llab(u_1)}}, \\
            &\textcolor{blue}{e_{\llab(u_1), \llab(u_2)}}, &\textcolor{red}{\lloop_{\llab(u_2)}}, \\
            &\dots,\\ 
            &\textcolor{blue}{e_{\llab(u_q), \llab(w)}}, &\textcolor{red}{e_{\llab(w), s}}
        \end{aligned}
    \end{equation*}
    (see \cref{fig:ham-case-3-3}~(a)).
    Note that by \eqref{eq:case-3-3-aux} and \eqref{eq:case-3-3-blue}, the trail $T^i$ indeed contains all edges.

    \begin{figure}[b]
    \includegraphics{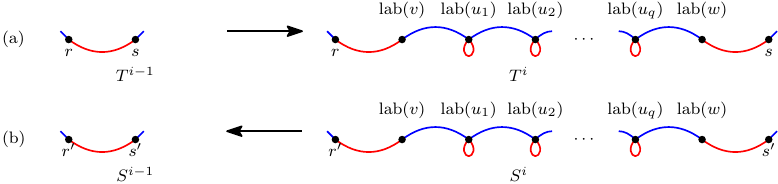}
    \centering
    \caption{(a) Constructing the red-blue Eulerian trail (a) $T^i$ from $T^{i-1}$ and (b) constructing the red-blue Eulerian trail $S^i$ from $S^{i-1}$ in \textbf{Case 3.3}.} 
    \label{fig:ham-case-3-3}
    \end{figure}

    Now we prove the other direction.
    Let $S^i$ be \rbet{} of $\aux_D(\PP_1' \sqcup \RR^i) \uplus \MM^i$.
    First, to show that $\aux_D(\PP'_1 \sqcup \RR^{i-1})$ is well-defined, we prove that $\PP'_1 \sqcup \RR^{i-1}$ is a path packing by showing its acyclicity (above we explained why this needs to be proven in this case only).
    Above we have argued that for every $j \in [q]$ the vertex $u_j$ forms a path in $\PP_1' \sqcup \RR^i$.
    Thus, the only edge incident with the vertex $\llab(u_j)$ in $\aux_D(\PP_1' \sqcup \RR^i)$ is a loop. 
    So the only two blue edges incident with $\llab(u_j)$ in $\MM^i$ appear right before and after this loop in $S^i$.
    Thus, the sequence
    \begin{equation*}
        \begin{aligned}
            Q = &\textcolor{blue}{e_{\llab(v), \llab(u_1)}}, &\textcolor{red}{\lloop_{\llab(u_1)}}, \\
            &\textcolor{blue}{e_{\llab(u_1), \llab(u_2)}}, &\textcolor{red}{\lloop_{\llab(u_2)}}, \\
            &\dots, \\ 
            &\textcolor{blue}{e_{\llab(u_q), \llab(w)}}.&
        \end{aligned}
    \end{equation*}
    (or its reverse) is a subtrail of $S^i$.
    As before we may assume that $Q$ is a subtrail of $S^i$.
    Let $e^1$ resp.\ $e^2$ be the red edges preceding resp.\ following $e_{\llab(v), \llab(u_1)}$ resp.\ $e_{\llab(u_q), \llab(w)}$ in $S^i$.
    And let $P_v'$ resp.\ $P_w'$ be the path in $\aux_D(\PP_1' \sqcup \RR^i)$ with an end-point $v$ resp.\ $w$ (above we have argued that such a path exists).
    Let $r' \in [k]$ resp.\ $s' \in [k]$ be such that $r'$ and $\llab(v)$ resp.\ $\llab(w)$ and $s'$ are the labels of end-vertices of $P_v'$ resp.\ $P_w'$.
    Observe that we have $e^1 = e_{r', \llab(v)}$ and $e^2 = e_{\llab(w), s'}$ since $v$ and $w$ are vertices with unique labels in $D$.
    To prove the claimed acyclicity, suppose it holds that $e^1 = e^2$. 
    Then $S^i$ consists exactly of the following edges
    \begin{equation*}
            \begin{aligned}
            S^i = &&\textcolor{red}{(e^1 = e^2)}, \\
            &\textcolor{blue}{e_{\llab(v), \llab(u_1)}}, &\textcolor{red}{\lloop_{\llab(u_1)}}, \\
            &\textcolor{blue}{e_{\llab(u_1), \llab(u_2)}}, &\textcolor{red}{\lloop_{\llab(u_2)}}, \\
            &\dots,& \\ 
            &\textcolor{blue}{e_{\llab(u_q), \llab(w)}}
            \end{aligned}
    \end{equation*}
    and we have $r' = \llab(w)$ and $s' = \llab(v)$.
    This implies that first, $P_v' = P_w'$ (since $v$ and $w$ are vertices of unique labels) and second, the path packing $\PP_1' \sqcup \RR^i$ consists of a path from $v$ to $w$ and paths $\{u_1\}, \dots, \{u_q\}$ only.
    Therefore, the multigraph $\aux_D(\PP_1 \sqcup \RR^i)$ consists of loops at $\llab(u_1), \dots, \llab(u_q)$ and an edge between $\llab(v)$ and $\llab(w)$.
    The degree sequences of $\aux_D(\PP_1 \sqcup \RR^i)$ and $\aux_D(\PP'_1 \sqcup \RR^i)$ coincide (recall \cref{obs:replacement-property}) so the degree sequences of $\PP_1 \sqcup \RR^i$ and $\PP'_1 \sqcup \RR^i$ coincide as well (recall \cref{obs:aux-degree-conversion}). 
    But then both $\PP_1 \sqcup \RR^i$ and $\{P\}$ contain a path from $v$ to $w$ (where $v$ and $w$ are distinct glue-vertices) so
    \[
        (\PP_1 \sqcup \RR^i) \sqcup \{P\} = \PP_1 \sqcup \RR^{i-1}
    \]
    contains a cycle -- this contradicts the fact that the graph $\PP_1 \sqcup \RR^{i-1})$ (a subgraph of a path packing $\PP_1 \sqcup \PP_2$) is a path packing.
    Hence, it holds that $e^1 \neq e^2$ and therefore, $P_v' \neq P_w'$.
    
    Now we have:
    \[
        \PP_1' \sqcup \RR^{i-1} = \left(\left(\PP_1' \sqcup \RR^i\right) - \left\{P_v', \{u_1\}, \dots, \{u_1\}, P_w'\right\}\right) \dot\cup \left\{P_v' \sqcup P \sqcup P_w'\right\}.
    \]
    Since the paths $P_v'$ and $P_w'$ are distinct, in particular, this implies that $\PP_1' \sqcup \RR^{i-1}$ is acyclic and therefore, it is a path packing.
    We then have
    \begin{equation}\label{eq:case-3-3-aux-prime}
        \begin{aligned}
            &\aux_D(\PP_1' \sqcup \RR^{i-1}) = \\ 
            &\left(\aux_D(\PP_1' \sqcup \RR^i) - e_{r', \llab(v)} - \lloop_{\llab(u_1)} - \dots - \lloop_{\llab(u_q)}, e_{\llab(w), s'}\right) \dot\cup
            \left\{e_{r', s'}\right\}.
        \end{aligned}
    \end{equation}
    Above we have argued that the sequence $L'$ is a subtrail of $S^i$ where
    \begin{equation*}
        \begin{aligned}
            L' = &&\textcolor{red}{e_{r', \llab(v)}}, \\
            &\textcolor{blue}{e_{\llab(v), \llab(u_1)}}, &\textcolor{red}{\lloop_{\llab(u_1)}}, \\
            &\textcolor{blue}{e_{\llab(u_1), \llab(u_2)}}, &\textcolor{red}{\lloop_{\llab(u_2)}}, \\
            &\dots, \\ 
            &\textcolor{blue}{e_{\llab(u_q), \llab(w)}}, &\textcolor{red}{e_{\llab(w), s'}}.
        \end{aligned}
    \end{equation*}
    Now this subtrail of $S^i$ can be replaced by the red edge $e_{r', s'}$~(see \cref{fig:ham-case-3-3}~(b)) to obtain \rbet{} $T_{i-1}'$ of $\aux_D(\PP_1' \sqcup \RR^{i-1}) \uplus \MM^{i-1}$ (see \eqref{eq:case-3-3-blue} and \eqref{eq:case-3-3-aux-prime} to verify that every edge is indeed used exactly once).
    This concludes the proof of the claim.
    \end{claimproof}

    Applied to $i = t$, the claim implies the existence of a blue multigraph $\MM^t$ with the following two properties.
    First, the multigraph $\aux_D(\PP_1 \sqcup \RR^t) \uplus \MM^t$ admits a red-blue Eulerian trail $T^t$.
    Since $\RR^t = \emptyset$, we obtain that 
    \[
        \aux_D(\PP_1 \sqcup \RR^t) \uplus \MM^t =\aux_D(\PP_1) \uplus \MM^t = \aux_{H_1}(\PP_1) \uplus \MM^t 
    \]
    admits a red-blue Eulerian trail $T^t$.
    Then the property $\AAA_1 \lesssim_{H_1} \Pi(H_1)$ implies that there exists a maximal path packing $\PP_1' \in \AAA_1$ such that 
    \[
        \aux_{H_1}(\PP_1) \uplus \MM^t = \aux_D(\PP_1) \uplus \MM^t 
    \]
    admits \rbet{}, say $S^t$.
    Then the second part of the claim applied to $i = t, \dots, 1$ implies that the multigraph 
    \[
        \aux_D(\PP_1' \sqcup \RR^0) \uplus \MM^0 = \aux_D(\PP_1' \sqcup \PP^2) \uplus \MM
    \]
    admits a \rbet{}.
    
    Now, a symmetric argument implies that there also exists a maximal path packing $\PP_2' \in \AAA_2$ such that $\aux_D(\PP_1' \sqcup \PP_2')\uplus \MM$ admits a \rbet{} (and in particular, $\PP_1' \sqcup \PP_2'$ is a path packing so the auxiliary graph is well-defined).
    By definition of $S$, we obtain that $\PP_1' \sqcup \PP_2'$ belongs to $S$.
    
    Altogether, we have shown that for every blue multigraph $\MM$ and every maximal path packing $\PP \in \Pi(D)$, if $\aux_D(\PP) \uplus \MM$ admits \rbet{}, then there exists a maximal path packing $\PP' \in S$ of $D$ such that $\aux_D(\PP) \uplus \MM$ also admits \rbet{}, i.e., we have $S \lesssim_D \Pi(D)$ as desired.
    Computing $S$ in time $\ostar(|\AAA_1| |\AAA_2|)$ is trivial: we iterate over all pairs $\PP_1 \in \AAA_1$ and $\PP_2 \in \AAA_1$, compute $\PP_1 \sqcup \PP_2$ in time polynomial in the size of $H$ and then check whether this subgraph is acyclic.
\end{proof}

Now we are ready to provide the algorithm solving the \HC{} problem parameterized by fusion-width.

\begin{theorem}
    Given a fuse-$k$-expression of a graph $H$, the \HC{} problem can be solved in time $n^{\mathcal{O}(k)}$.
\end{theorem}

\begin{proof}
    First, by \cref{app:thm:useful-expression}, in time polynomial in the size of the given fuse-$k$-expression and $k$, we can compute a reduced glue-$k$-expression $\phi$ of $H$ whose size is polynomial in the size of $H$ and $k$.

    Let $e = uv$ be an arbitrary but fixed edge of $H$.
    In the following, our algorithm will decide, whether $H$ admits a Hamiltonian cycle containing $e$.
    Then, by checking this for every edge of $H$, we can solve \HC{}.
    First, we slightly transform $\phi$ into a reduced glue-$(k+2)$-expression $\xi$ of $H$ such that the root of $\xi$ is a join-node that creates the edge $e$ only.
    For this we proceed as follows.
    For the simplicity of notation, let $i_u = k+1$ and $i_v = k+2$.
    First, all leaves of $\phi$ with title $u$ (resp.\ $v$) are replaced with $u \langle i_u \rangle$ (resp.\ $v \langle i_v \rangle$).
    After that, we iterate through all join-nodes $t$.
    Let $i, j \in [k]$ be such that $t$ is a $\eta_{i, j}$-node.
    If the vertex $u$ belongs to $G^\phi_t$ and the label $j_u$ of $u$ in $G^\phi_t$ is equal to $i$ (resp.\ $j$), we add a new $\eta_{i_u, j}$-node (resp.\ $\eta_{i_u, i}$-node) right above $t$.
    Similarly, if the vertex $v$ belongs to $G^\phi_t$ and the label $j_v$ of $v$ in $G^\phi_t$ is equal to $i$ (resp.\ $j$), we add a new $\eta_{i_v, j}$-node (resp.\ $\eta_{i_v, i}$-node) right above $t$.
    After processing all nodes, we add a new $\eta_{i_u, i_v}$-node above the root making it a new root. 
    The process takes only polynomial time.
    By construction, this expression still creates the graph $G$.
    Moreover, it still satisfies the first two properties of a reduced glue-$k$-expression: no join-node creates an already existing edge and the glued graphs are always edge-disjoint (see \cref{lem:useful-glue-expression} for a formal definition).
    After that, we proceed as in the proof of \cref{lem:useful-glue-expression} to ensure that the last property of a reduced glue-$(k+2)$-expression is still satisfied. 
    If we look into that proof, we note that this transformation does not change the root so the join-node creating the edge $e$ is still the root.
    We denote the arisen glue-$(k+2)$-expression by $\xi$.
    By $x$ we denote the child of the root of $\xi$.
    Note that since $u$ and $v$ now have unique labels, the node $x$ creates the edge $e$ only so we have $G^\xi_x = H - e$.
    
    Now given the result of Bergougnoux et al.\ for introduce-, join-, and relabel-nodes~\cite{BergougnouxKK20} as well as our \cref{lem:dp-glue-node}, we can traverse $\xi$ bottom-up to compute a set $\AAA_x$ of partial solutions of $G^\xi_x$ such that $\AAA_x \lesssim_{G^\xi_x} \Pi(G^\xi_x)$.
    Namely, we start with the leaves and then given a set (or sets) of partial solutions representing the set (resp.\ sets) of all partial solutions of the child (resp.\ children) of the current node, say $y$, we first compute a set of partial solutions of $y$ representing all partial solutions of $G^\xi_y$ and then apply the $\rreduce_{G^\xi_y}$-operation to ensure that the number of partial solutions kept at each step is bounded by $n^{\mathcal{O}(k)}$.
    We emphasize that the first two properties of reduced a glue-$(k+2)$-expression (see \cref{{lem:useful-glue-expression}}) ensure that every edge is created exactly once.
    So the expression is ``irredundant'' in terms of clique-width and therefore, the procedures for introduce-, relabel-, and join-nodes by Bergougnoux et al.\ are still correct~\cite{BergougnouxKK20}.
    
    Now we show to decide whether $H$ admits a Hamiltonian cycle using the edge $e$ given the set $\AAA_x$.
    We claim that this is the case if and only if $\AAA_x$ contains a maximal path packing $\PP$ consisting of a single path $P$ with end-vertices $u$ and $v$.
    One direction is almost trivial: Recall that $\PP \in \AAA_x \subseteq \Pi(G^\xi_H)$ and $G^\xi_H$ is a subgraph of $H$ so $P$ is a path in $H$.
    Since $V(H) = V(G^\xi_x)$ and $\PP$ is maximal, the path $P$ contains every vertex of $H$ and together with the edge $e$ it forms a Hamiltonian cycle of $H$.

    For the other direction, let $H$ contain a Hamiltonian cycle using the edge $e$, let $P'$ denote the path connecting $u$ and $v$ on this cycle and not using the edge $e$. 
    Let $\PP'$ denote the path packing of $H$ consisting of $P'$ only.
    Note that since we consider a Hamiltonian cycle, the path $P'$ contains every vertex of $H$ so $\PP'$ is indeed a maximal path packing of $H$.
    Moreover, since it does not contain the edge $e$, $\PP'$ is a path packing of $G^\xi_x$ as well, i.e., $\PP' \in \Pi(G^\xi_x)$.
    Since the vertices $u$ and $v$ have unique labels and they are never relabeled, the red graph $\aux_{G^\xi_x}(\PP')$ on the vertex set $[k+2]$ has a single edge and this edge has end-points $i_u$ and $i_v$.
    Consider a blue multigraph $\MM$ on the vertex set $[k+2]$ whose single edge is between $i_u$ and $i_v$.
    Trivially, the multigraph $\aux_{G^\xi_x}(\PP') \uplus \MM$ admits \rbet{}.
    Then the property $\AAA_x \lesssim_{G^\xi_x} \Pi(G^\xi_x)$ implies that there is a path packing $\PP \in \AAA_x$ such that $\aux_{G^\xi_x}(\PP) \uplus \MM$ admits \rbet{} as well.
    Since $\MM$ consists of a single blue edge between $i_u$ and $i_v$, the red multigraph $\aux_{G^\xi_x}(\PP)$ consists of a single red edge between $i_u$ and $i_v$.
    Again, recall that $u$ resp.\ $v$ are the unique vertices with label $i_u$ resp.\ $i_v$ so the path packing $\PP$ consists of a single path $P$ containing all vertices from $V(G^\xi_x) = V(H)$ (i.e., maximal) with end-points $u$ and $v$.

    Recall that the number of partial solutions kept at each node is bounded by $n^{\mathcal{O}(k)}$ due to the application of the $\operatorname{reduce}$ operator after processing every node.
    By \cref{lem:dp-glue-node}, a glue-node is processed in time polynomial in the number of partial solutions kept for its children, i.e., in $n^{\mathcal{O}(k)}$.
    Also by the results of Bergougnoux et al.~\cite{BergougnouxKK20}, the remaining nodes can also be handled in time~$n^{\mathcal{O}(k)}$.
    Finally, recall that a reduced glue-$(k+2)$-expression contains a polynomial number of nodes.
    So the algorithm runs in time $n^{\mathcal{O}(k)}$.
\end{proof}

Fomin et al.\ have also shown the following lower bound:
\begin{theorem}\cite{FominGLSZ19}
    Let $H$ be an $n$-vertex graph given together with a $k$-expression of $H$. 
    Then the \textsc{Hamiltonian Cycle} problem cannot be solved in time $f(k) \cdot n^{o(k)}$ for any computable function $f$ unless the ETH fails.
\end{theorem}

Since any $k$-expression of a graph is, in particular, its fuse-$k$-expression, the lower bound transfers to fuse-$k$-expressions as well thus showing that our algorithm is tight under ETH.
\begin{theorem}
    Let $H$ be an $n$-vertex graph given together with a fuse-$k$-expression of $H$. 
    Then the \textsc{Hamiltonian Cycle} problem cannot be solved in time $f(k) \cdot n^{o(k)}$ for any computable function $f$ unless the ETH fails.
\end{theorem}

%% file: mcw-algorithms.tex
\section{ Algorithms Parameterized by Multi-Clique-Width}\label{app:sec:mcw-algorithms}

In this section we show that for many problems we can obtain algorithms parameterized by multi-clique-width with the same running time as known SETH-tight (and for \textsc{Chromatic Number} even an ETH-tight) algorithms parameterized by clique-width. 
Due to the relation 
\[
    \mcw \stackrel{\cref{app:thm:mcw-fw-relation}}{\leq} \fw + 1 \leq \cw + 1,
\]
the obtained running times are then (S)ETH-tight for all three parameters.
For these problems, we will use known algorithms for the clique-width parameterization and describe what adaptations are needed to handle multiple labels per vertex.
First, we make some simple observations to restrict ourselves to simpler expressions.
\begin{observation}
    Let $H$ be a $k$-labeled graph, let $r \in \NN$ and let $i, s_1, \dots, s_r \in [k]$.
    Then if $i \in \{s_1, \dots, s_r\}$, then it holds that
    \[
        \rho_{i \to \{s_1, \dots, s_r\}} (H) = \rho_{i \to \{i, s_r\}} \circ \dots \circ \rho_{i \to \{i, s_1\}} (H)
    \]
    and if $i \notin \{s_1, \dots, s_r\}$, then it holds that
    \[
        \rho_{i \to \{s_1, \dots, s_r\}} (H) = \rho_{i \to \emptyset} \circ \rho_{i \to \{i, s_r\}} \circ \dots \circ \rho_{i \to \{i, s_1\}} (H).
    \]
    We also have
    \[
        1 \langle s_1, \dots, s_r \rangle = \rho_{s_1 \to \{s_1, s_r\}} \circ \dots \circ \rho_{s_1 \to \{s_1, s_2\}} \circ 1 \langle s_1 \rangle
    \]
    if $r > 0$
    and 
    \[
        1 \langle \emptyset \rangle = \rho_{1 \to \emptyset} 1 \langle 1 \rangle.
    \]
\end{observation}
Therefore, by first applying the above rules to all relabel-nodes whose right side has size one or more than two and then suppressing relabel-nodes of form $\rho_{i \to \{i\}}$, we may restrict ourselves to relabel-nodes that either remove some label $i$ or add a label $j$ to every vertex with label $i$.
Similarly, using the last two equalities we may assume that every introduce-node uses exactly one label.
Note that the length of the multi-$k$-expression increases by at most a factor of $k$ after these transformations.
Also we can assume that for every join-node, the joined label sets are non-empty. 

Finally, me may reduce the number of nodes in a multi-$k$-expression to polynomial as follows.
First, similarly to clique-expressions, we may assume that every union-node is first, followed by a sequence of join-nodes, then a sequence of relabel-nodes, and then a union-node (if exists). 
Then we may assume that between any two consecutive union-nodes, there are at most $k^2$ join-nodes: at most one per pair $i, j \in [k]$.
Now we sketch how to achieve that we also have at most $k$ relabel-nodes between two consecutive union-nodes, namely at most one per possible left side $i \in [k]$.
Suppose there are two distinct nodes $x_1$ and $x_2$ being $\rho_{i \to S_1}$- and $\rho_{i \to S_2}$-nodes, respectively, for $S_1, S_2 \subseteq [k]$.
We choose $x_1$ and $x_2$ so that there is no further $\rho_{i \to S'}$-nodes between them.
If the label set $i$ is empty right before the application of $x_2$, we simply suppress $x_2$.
Otherwise, observe that in all vertices that had label $i$ before $x_1$, this label was replaced by $S_1$ (with possibly $i \in S_1$).
So every vertex that has label $i$ right before $x_2$ got this label at some relabel-node on the path from $x_1$ to $x_2$ (including $x_1$).
Therefore, for every $\rho_{j \to S}$-node $x$ on this path with $i \in S$, we replace the operation in $x$ with $\rho_{j \to (S \setminus \{i\}) \cup S_2}$.
And after that we suppress $x_2$.
Note that this is correct since no $\rho_{i \to S'}$-node occurs between $x_1$ and $x_2$.
By repeating this process, we obtain that for each $i$ there is at most one $\rho_{i \to S}$-node between the two consecutive nodes.
As for a clique-expression, the leaves of a multi-expression are in bijection with the vertices of the arising graph (i.e., there at most $n$ leaves). 
And since union-nodes are the only nodes with more than one child, there are at most $\mathcal{O}(n)$ union-nodes.
Finally, the above argument implies that there are at most $\mathcal{O}(k^2 n)$ relabel- and join-nodes.

\begin{lemma}\label{lem:mcw-special-relabels}
    Let $\phi$ be a multi-$k$-expression of a graph $H$ on $n$ vertices.
    Then given $\phi$ in time polynomial in $|\phi|$ and $k$, we can compute a multi-$k$-expression $\xi$ of $H$, such that the are at most $\mathcal{O}(k^2 n)$ nodes and for every node $t$ of $\xi$ the following holds. 
    If $t$ is a $\rho_{i \to S}$-node for some $i \in [k]$ and $S \subseteq [k]$, then we have $|S| = \emptyset$ or $S = \{i, j\}$ for some $j \neq i \in [k]$.
    If $t$ is a $\eta_{i, j}$-node for some $i \neq j \in [k]$, then we have $U^t_i \neq \emptyset$ and $U^t_j \neq \emptyset$.
    And if $t$ is a $1\langle S \rangle$-node for some $S \subseteq [k]$, then we have $|S| = 1$.
\end{lemma}
For the remainder of this section we assume that an expression has this form.
Now we show how existing algorithms for clique-width can be adapted to achieve the same running time for multi-clique-width.

\subsection{ Dominating Set} \label{app:subsec:ds}
In the \textsc{Dominating Set} problem, given a graph $G = (V, E)$ we are asked about the cardinality of the smallest set $S \subseteq V$ with $N_G[S] = V$.
Bodlaender et al.\ have developed a $\ostar(4^{\cw})$ algorithm~\cite{BodlaenderLRV10}.
The idea behind it is to store a pair of Boolean values for every label: the first value reflects whether the label set contains a vertex from the partial solution while the second reflects whether all vertices of the label set are dominated.
Crucially (as for most problems handled below), the algorithm does not make use of the fact the every vertex holds exactly one label.
So we can use almost the same algorithm to process a multi-$k$-expression. 
The procedures for introduce-, join-, and union-nodes can be reused from Bodlaender et al.
And it remains to handle relabel-nodes: in this case, the state of every single vertex remains the same and we only need to represent the states with respect to the new labeling function.
For a $\rho_{i \to \emptyset}$-node, all vertices of label $i$ are now dominated and no vertex belongs to the dominating set while the other labels remain unaffected.
And for a $\rho_{i \to \{i, j\}}$-node, first, all vertices of label $j$ are dominated iff this was true for labels $i$ and $j$ before the relabeling; and second, a vertex of label $j$ belongs to a partial solution iff this was the case for the label $i$ or the label $j$ before; the state of other labels (in particular, of the label $i$ remains).
We omit a formal description since analogous ideas occur several times in the next problems.
This yields an $\ostar(4^{\mcw})$ algorithm.
Katsikarelis et al.\ have proven the matching lower bound for clique-width which then also applies to multi-clique-width~\cite{KatsikarelisLP19}.

\begin{theorem}\label{app:thm:dominating-set}
    Let $G$ be a graph given together with a multi-$k$-expression of $G$. Then     \textsc{Dominating Set} can be solved in time $\ostar(4^k)$. 
    Unless SETH fails, this problem cannot be solved in time $\ostar((4 - \varepsilon)^k)$ for any $\varepsilon > 0$.
\end{theorem}

\subsection{ \textsc{Chromatic Number}}\label{app:subsec:chromatic-number}
In the \textsc{Chromatic Number} problem, given a graph $G = (V, E)$ we are asked about the smallest integer $q$ such that there exists a proper $q$-coloring of $G$, that is, a mapping $\phi \colon V \to [q]$ such that $\phi(u) \neq \phi(v)$ for all $uv \in E$.
Kobler and Rotics have developed an algorithm that given a graph and a $k$-expression of this graph, solves the \textsc{Chromatic Number} problem in time $f(k) \cdot n^{2^{\mathcal{O}(k)}}$~\cite{KoblerR03}.
Later, Fomin et al.\ have proven the ETH-tightness of this result by showing that under ETH, there is no algorithm solving this problem in $f(k) \cdot n^{2^{o(k)}}$ even if a $k$-expression of the graph is provided~\cite{FominGLSZ19}.
The algorithm by Kobler and Rotics is based on dynamic programming. 
The records are of the form $N: 2^{[k]} \setminus \{\emptyset\} \to [n]_0$ and for a node $t$ of a $k$-expression, a mapping $N$ is a \emph{record} if there exists a proper coloring $c: V(G_t) \to \NN$ of $G_t$ such that for every subset $\emptyset \neq S \subseteq [k]$ of labels, we have 
\[
    \left|\left\{j \in \NN \mid \forall i \in S \colon U^t_i \cap c^{-1}(j) \neq \emptyset, \forall i \in [k] \setminus S \colon U^t_i \cap c^{-1}(j) = \emptyset\right\}\right| = N[S].
\]
Simply speaking, $N[S]$ reflects how many colors are there that occur at each label from $S$ but at no further label.
Let $T_t$ denote the set of records at the node $t$.
Let $r$ denote the root of the $k$-expression.
In the end, their algorithm outputs the smallest number $q$ such that there exists a record $N$ in $T_r$ with $\sum_{S \subseteq [k]} N[S] = q$.
This sum is exactly the number of colors used by a corresponding coloring since for every color, there exists a (unique) set $S$ of labels on which this color is used.
For the multi-clique-width setting, there is a small hurdle: it might happen that at some point a vertex loses all its labels so we might thus ``forget'' that such a vertex also uses some color.
To overcome this issue, given a multi-$k$-expression $\phi$ of a graph $G$, we create a multi-$k+1$-expression of the same graph by replacing every $1\langle i \rangle$-node (for $i \in [k]$ with a $1\langle i, k+1 \rangle$-node.
After that, we apply the simple transformations describe earlier to achieve that the expression satisfies \cref{lem:mcw-special-relabels}.
This ensures that at every sub-expression, every vertex has at least one label (namely $k+1$) so for any record $N$, the value $\sum_{\emptyset \neq S \subseteq [k+1]} N[S]$ still reflects the total number of colors used by the corresponding coloring.
Apart from that, their algorithm never uses that a vertex holds exactly one label and therefore can be easily adapted to multi-clique-width as follows.
Introduce-, join-, and union-nodes can be adopted from the original algorithm by Kobler and Rotics.
And it remains to handle relabel-nodes.

First, let $t$ be a $\rho_{i \to \emptyset}$-node for some $i \in [k+1]$ and let $t'$ be its child.
For every record $N'$ of $t'$, we create a record $N$ such that for every $\emptyset \neq S \subseteq [k+1]$, we have:
\[
    N[S] = 
    \begin{cases}
        0 & i \in S \\
        N'[S] + N'[S \cup \{i\}] & i \notin S
    \end{cases}
    .
\]
Then we add $N$ to the set of records of $t$.
It is straight-forward to verify that this process results exactly the set of records of $t$.

Second, let $t$ be a $\rho_{i \to \{i, j\}}$-node for some $i \neq j \in [k]$ and let $t'$ be its child.
We may assume that $i$ is non-empty at $t'$ since $t$ can be suppressed otherwise.
For every record $N'$ of $t'$, we create a record $N$ as follows.
If the label $j$ is empty in $t'$, then for all $\emptyset \neq S \subseteq [k]$ we set 
\[
    N[S] = N'[S']
\]
where $S'$ is obtained from $S$ by swapping the roles of $i$ and $j$, i.e.,
\[
    S' = 
    \begin{cases}
        S & i, j \notin S \\
        S & i, j \in S \\
        S \setminus \{i\} \cup \{j\} & i \in S, j \notin S \\
        S \setminus \{j\} \cup \{i\} & i \notin S, j \in S \\
    \end{cases}
    .
\]
Otherwise, we may assume that $j$ is non-empty in $t'$.
Then for every $\emptyset \neq S \subseteq [k]$, we set:
\[
    N[S] = 
    \begin{cases}
        0 & i \in S, j \notin S \\
        N'[S] + N'[S \setminus \{j\}] & i \in S, j \in S \\
        N'[S] & i \notin S
    \end{cases}
    .
\]
Again, it is easy to verify these equalities.
In the end, as in the original algorithm, we output the smallest number $q$ for which there is a record $N$ in $T_r$ with $\sum_{\emptyset \neq S \subseteq [k+1]} N[S] = q$.
The process runs in time $n^{2^{\mathcal{O}(k)}}$ for every node and since the number of nodes can be assumed to be polynomial, the whole algorithm has the same running time.
Fomin et al.\ showed that under ETH the problem cannot be solved in time $f(k) \cdot n^{2^{o(k)}}$ for any computable function $f$.
Since multi-clique-width lower-bounds clique-width, this also applies in our case.

\begin{theorem}\label{app:thm:chromatic-number}
    Let $G$ be a graph given together with a multi-$k$-expression of $G$. Then     \textsc{Chromatic Number} problem can be solved in time $f(k) \cdot n^{2^{\mathcal{O}(k)}}$.
    Unless ETH fails, this problem cannot be solved in time $f(k) \cdot n^{2^{o(k)}}$ for any computable function $f$.
\end{theorem}

\subsection{ q-Coloring} \label{app:subsec:q-coloring}
In the \textsc{$q$-Coloring} problem, given a graph $G = (V, E)$ we are asked about existence of a proper $q$-coloring of $G$, that is, a mapping $\phi \colon V \to [q]$ such that $\phi(u) \neq \phi(v)$ for all $uv \in E$.
Now we sketch the main idea behind the SETH-tight $\ostar((2^q - 2)^{\cw})$ algorithm for the \textsc{$q$-Coloring} problem parameterized by clique-width by Lampis~\cite{Lampis20}.
The naive algorithm traversing a clique-expression would store for every label, the set of colors occuring on the vertices of that label, and thus have the running time of $\ostar((2^q)^{\cw})$.
But Lampis observed that two states can be excluded.
First, the empty set of colors occurs at some label if and only if this label is empty so such information is trivial and there is no need to store it.
Second, if some label $i$ contains all $q$ colors and later the vertices of this label obtain a common neighbor $v$, then such a coloring would become not proper since at least of the colors also occurs on label $v$.
Therefore, the set of all colors can only appear on a label that would not get any new common neighbors (and in particular, this label does not participate in any join later).
This led Lampis to a notion of a so-called live label.
A label is live at some node $t$ of the expression if it contains a vertex which will later get an incident edge (and hence, all vertices in the label set will get a common neighbor).
In particular, a live label is non-empty.
For multi-clique-width we will follow a similar idea, however we need to slightly adapt the notion of a live label. 
For motivation, we provide the following example.
Let $q = 3$, we consider a multi-$3$-labeled graph on four vertices with label sets \{1, 2\}, $\{3\}$, $\{1\}$, and $\{1\}$.
Suppose this graph is edgeless, and a partial solution colors these vertices with colors $\bblue$, $\ggreen$, $\ggreen$, and $\rred$, respectively (see \cref{fig:q-coloring}).
And now a $\eta_{2, 3}$-operation occurs. 
Although all three colors appear on label $1$ and a vertex of label $1$ now gets a neighbor, this does not make a partial solution invalid.
So although the label $1$ is live as defined by Lampis, it can still hold all colors.
The reason is that the edge creation happens \emph{due to some other label held by a vertex with label $1$}.
This motivates our following definition.

\begin{figure}[b]
    \includegraphics{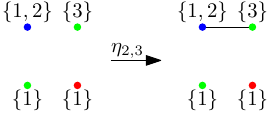}
    \centering
    \caption{Although the label $1$ contains all colors, creating an edge incident with a vertex holding the label $1$ does not lead to conflicts.} 
    \label{fig:q-coloring}
\end{figure}

We say that a label $i$ is \emph{active} at a node $t$ of a multi-$k$-expression $\phi$ if $U^t_i$ is non-empty and there exist labels $a \neq b \in [k]$ and a $\eta_{a, b}$-node $t'$ such that:
\begin{itemize}
    \item The node $t'$ is an ancestor of $t$;
    \item Let $t_1, \dots, t_q$ be all inner relabel-nodes on the path from $t$ to $t'$ (in the order they occur on this path) and let $s_1, \dots, s_q \in [k]$ and $S_1, \dots, S_q \subseteq [k]$ be such that for every $j \in [q]$, the node $t_j$ is a $\rho_{s_j \to S_j}$-node. Then it holds that 
    \[
        a \in \sigma_q ( \sigma_{q-1} ( \dots (\sigma_1(\{i\}) ) ) )
    \]
    where for $T \subseteq [k]$
    \[
        \sigma_j(T) =
        \begin{cases}
            T & \text{if } s_j \notin T \\
            (T \setminus \{s_j\}) \cup S_j & \text{if } s_j \in T \\ 
        \end{cases}
        ;
    \]
    \item The set $U^{t'}_b$ is non-empty.
\end{itemize}
Informally speaking, the set $\sigma_q ( \sigma_{q-1} ( \dots \sigma_1(\{i\}) ) )$ contains the labels to which the label $i$ has been relabeled on the way to $t'$.
Then, for any label $a$ from this set, all vertices with label $i$ in $t$ contain the label $a$ in $t'$ so the join-node $t'$ creates a common neighbor for these vertices.
By $\ell(t)$ we denote the set of all active labels at node $t$.
With this definition, we are now ready to provide an algorithm that given a graph and its multi-$k$-expression solves the \textsc{$q$-Coloring} problem in time $\ostar((2^q - 2)^{k})$ by keeping track of the sets of colors used by every active label.

By \cref{lem:mcw-special-relabels}, we may assume that we are given a multi-$k$-expression of $H$ satisfying the properties of the lemma.
We will follow the dynamic programming by Lampis, provide how to handle relabel-nodes, and observe that the approach for introduce, join-, and union-nodes can be adopted from his work without changes.
By $\CC$ we denote the set $2^{[q]} \setminus \{\emptyset, [q]\}$ of relevant sets of colors, then we have $|\CC| = 2^q - 2$.
The dynamic programming table $A_t$ is indexed by $\CC^{\ell(t)}$, i.e., the assignments of colors to active labels.
And for $f \in \CC^{\ell(t)}$, the value $A_t[f]$ is the number of proper $q$-colorings of $G_t$ such that for every active label $i$, the coloring uses exactly the colors $f(i)$.

Now let $t$ be a $\rho_{i \to \emptyset}$-node with $i \in [k]$ and let $t'$ be its child.
Observe that the label $i$ is then inactive in both $t$ and $t'$.
The remaining labels are not affected so we simply have $\ell(t) = \ell(t')$ and $A_t = A_{t'}$.

Next let $t$ be a $\rho_{i \to \{i, j\}}$-node for $i, j \in [k]$ and let $t'$ be its child.
We may assume that $i \neq j$ holds since otherwise $t$ can be suppressed.
Now we have to make a case distinction based on the activity of labels $i$ and $j$.

\textbf{Case 1}: Assume that $i$ is inactive in $t'$. 
Then by definition of activity, both labels $i$ and $j$ are inactive in $t$.
Other labels are not affected so we again have $\ell(t) = \ell(t')$ and $A_t = A_{t'}$.

\textbf{Case 2}: From now on, we assume that $i$ is active in $t'$. 
Then by definition, at least one of labels $i$ and $j$ is active in $t$.

\textbf{Case 2.a}: Assume that both $i$ and $j$ are active in $t$.
Then label $j$ is either active (and then $\ell(t) = \ell(t')$) or empty in $t'$ (then $\ell(t) = \ell(t') \dot \cup \{j\}$ by definition).
We will compute the entries of the table $A_t$ in the entries of a table $B_t$ indexed by $\CC^{\ell(t)}$ as follows.
First, all entries of $B_t$ are initialized with zeros.
After that we iterate over all footprints $f \in \CC^{\ell(t')}$.
Let $f'$ be such that for all $p \in \ell(t)$ we have:
\[
    f'(p) =
    \begin{cases}
        f(p) & p \neq j \\
        f(i) \cup f(j) & p = j \text{ and } j \text{  is active in } t' \\
        f(i) & p = j \text{ and } j \text{  is empty in } t' \\
    \end{cases}
    .
\]
If $j$ is active in $t'$ and $f'(j) = f(i) \cup f(j)$, then we skip $f'$: the property that $j$ is an active label implies that all vertices of this label will get a common neighbor and all $q$-colorings with such footprints will be invalidated.
Otherwise, we increase the value $B_t[f']$ by the value $A_{t'}[f]$.
This process requires a linear in the number of entries of $A_{t'}$ number of arithmetic operations.
As a result, the entries of $B_t$ coincide with $A_t$: indeed, in this process we have simply recomputed the footprints of every $q$-coloring and eliminated the colorings that would be invalidated later in the expression anyway.

\textbf{Case 2.b}: Assume that $i$ is active in $t$ and $j$ is inactive in $t$.
Then by definition $j$ is inactive in $t'$ as well.
Then we simply have $\ell(t) = \ell(t')$ and $A_t = A_{t'}$ since the underlying graph has not been changed.

\textbf{Case 2.c}: Assume that $i$ is inactive in $t$ and $j$ is active in $t$.
Then $j$ is either active (and then $\ell(t) = \ell(t') - \{i\}$) or empty (then $\ell(t) = (\ell(t') - \{i\}) \cup \{j\}$) in $t'$.
This case is similar to Case 2.a and is handled as follows.
We will compute the entries of the table $A_t$ in the entries of a table $B_t$ indexed by $\CC^{\ell(t)}$ as follows.
First, we iniatilize all values of $B_t$ with zeros.
After that we iterate over all footprints $f \in \CC^{\ell(t')}$.
Let $f' \in \CC^{\ell(t)}$ be such that for all $p \in \ell(t)$ we have:
\[
    f'(p) =
    \begin{cases}
        f(p) & p \neq j \\
        f(i) \cup f(j) & p = j \text{ and } j \text{ is active in } t' \\
        f(i) & p = j \text{ and } j \text{ is empty in } t'
    \end{cases}
    .
\]
If $j$ is active in $t'$ and $f'(j) = f(i) \cup f(j)$, then we skip $f'$: the property that $j$ is an active label implies that all vertices of this label will get a common neighbor and all $q$-colorings with such footprints will be invalidated.
Otherwise, we increase the value $B_t[f']$ by the value $A_{t'}[f]$.
This process requires a linear in the number of entries of $A_{t'}$ number of arithmetic operations.
As a result, the table $B_t$ contains the same entries as $A_t$ and the argument is analogous to Case 2.a.
This concludes the procedure for relabel-nodes.

Now let $t$ be a $\eta_{i, j}$-node for some $i \neq j \in [k]$ and let $t'$ be its child.
If $i$ would be inactive in $t'$, then $i$ would be empty in $t'$ and the join-node $t$ could be suppressed.
So we may assume that $i$, and similarly $j$, is active in $t'$.
Now we can proceed the same way as Lampis so we only sketch it very briefly.
Some of the labels $i$ and $j$ may become inactive in $t$ so let $I = \ell(t') \setminus \ell(t) \subseteq \{i, j\}$.
Then we set
\[
    A_t[f] =
    \begin{cases}
        0 & f(i) \cap f(j) \neq \emptyset \\
        \sum_{c \in \CC^I} A_{t'}[f \times c] & f(i) \cap f(j) = \emptyset \\
    \end{cases}
    .
\]
The first case reflects that every $q$-coloring of $G_t$ in which labels $i$ and $j$ share a color is not proper.
In the second case, we keep the information about the coloring and store it with the correct footprint.
Crucial here is that although we use a different notion of active labels (compared to live labels by Lampis), the equality remains the same.
As we see next, the same holds for union-nodes so that their fast processing can be adopted from Lampis.

Let $t$ be a union-node and let $t_1$ and $t_2$ be its children.
Observe that we have $\ell(t) = \ell(t_1) \cup \ell(t_2)$.
Moreover, if some label $i \in [k]$ is active in $t$ but not in $t_1$, then $i$ is empty in $t_1$ so the set of colors used on it in any coloring is empty.
This is the property that ensures that the approach of Lampis is correctly applicable in our case.
The approach relies on fast subset convolution and can be described as follows.
He first computes the entries $B_{t_1}[f]$ and $B_{t_2}[f]$ of auxiliary tables such that for label $i$, the value $f(i)$ provides the set of colors \emph{allowed} to be used on the label $i$, i.e., an upper bound on the set of colors instead of the exact value.
Then the analogous table $B_t$ for the node $t$ can be computed as pointwise multiplication of $B_{t_1}$ and $B_{t_2}$ and finally, the reverse procedure is used to compute $A_t$ from $B_t$.
We refer to the paper by Lampis for all details~\cite{Lampis20}.
The whole process requires the number of arithmetic operations linear in the number of entries of $A_t$.

Now we are able to handle all types of nodes of a multi-$k$-expression and compute the table $A_r$ where $r$ denotes the root of the expression.
Observe that no label is active in $r$ so this table contains a unique entry equal to the number of proper $q$-colorings of $G_r = H$ and this approach even solves the counting version of the problem.
Each of the considered tables has at most $(2^q - 2)^k$ entries and as argued above, at each node we carry out a number of arithmetic operations linear in the size of the table.
Each entry is bounded by $q^n$ (the largest possible number of $q$-colorings of $H$) and the number of nodes in the expression is polynomial in the size of $H$ and $k$ so the total running time is bounded by $\ostar((2^q - 2)^k)$ as claimed.
Lampis also showed that his algorithm is tight under SETH and since multi-clique-width lower-bounds clique-width, this also applies in our case.

\begin{theorem}\label{app:thm:q-coloring}
    Let $G$ be a graph given together with a multi-$k$-expression of $G$. Then     \textsc{$q$-Coloring} can be solved in time $\ostar((2^q - 2)^k)$. 
    Unless SETH fails, this problem cannot be solved in time $\ostar((2^q - 2 - \varepsilon)^k)$ for any $\varepsilon > 0$.
\end{theorem}

\subsection{ \textsc{Connected Vertex Cover}} \label{app:subsec:cvc}

In the \textsc{Connected Vertex Cover} problem, given a graph $G = (V, E)$ we are asked about the cardinality of the smallest set $S \subseteq V$ such that $G[S]$ is connected and for every edge $uv \in E$ we have $u \in S$ or $v \in S$.
Although the \textsc{Connected Vertex Cover} problem seems to be very different from \textsc{$q$-Coloring} at first sight, for most parameterizations known (e.g., \cite{HegerfeldK20,CyganNPPRW22,HegerfeldK23}), the SETH-tight algorithms rely on the Cut\&Count technique introduced by Cygan et al.~\cite{CyganNPPRW22}.
Very briefly they show that to decide whether a graph admits a connected vertex cover of certain size, it suffices to count the number of pairs $(L, R)$ where $L \dot\cup R$ is a vertex cover of the desired size and there are no edges between $L$ and $R$.
This very rough description hides some technical details like fixing a vertex in $L$, assuming that a vertex cover of minimum weight is unique via Isolation lemma etc.\ (see \cite{CyganNPPRW22}).
But this shows that on a high level, solving this problem reduces to counting the number of colorings of the graph with colors $L$, $R$, and $N$ (stands for not being in a vertex cover) where $LR$ and $NN$ edges are forbidden: the former pair forbids edges between $L$ and $R$ and the latter ensures that every edge of the input graph has at least one vertex in the vertex cover $L \cup R$.
Hegerfeld and Kratsch employ this observation to obtain an $\ostar(6^{\cw})$ algorithm for \textsc{Connected Vertex Cover} similar to the above algorithm for $q$-Coloring by Lampis.
There are 8 subsets of $\{L, R, N\}$, however, the empty set of colors can only occur on an empty label; while if a label contains all colors, then joining this label with some other label would necessarily lead to a $LR$ or a $NN$ edge~\cite{HegerfeldK23}.
Thus, stated in our terms, a label containing all colors from $\{L, R, N\}$ is necessarily inactive.
Therefore, it suffices to keep track of only six possible combinations of colors that may occur on an active label.
These observations together with a sophisticated convolution at union-nodes yield their $\ostar(6^{\cw})$ algorithm.

For us, two properties are crucial.
First, neither the definition of relevant colorings nor the procedure at any node-type uses the fact that any vertex holds only one label. 
The second property is more technical and requires a closer look at the algorithm by Kratsch and Hegerfeld.
For correctness of the algorithm, they assume that a $k$-expression of a graph is irredundant, namely no edge of the graph is created by multiple join-nodes.
It is folklore that any $k$-expression can be transformed into an irredundant $k$-expression of the same graph in polynomial time. 
Unfortunately, we do not know whether an analogous statement holds for multi-$k$-expressions. 
Also, instead of active labels as we define them in the previous section, they use the live labels as defined by Lampis~\cite{Lampis20}.
However, a closer look at their procedure for the union-node and its correctness reveals that it still works if an expression is not necessarily irredundant but we work with active labels instead of live labels.
Namely, they only use irredundant $k$-expressions to ensure that whenever there is a join-node $\eta_{i,j}$, the vertices of label $i$ will later get a \emph{new} common neighbor.
We observe that the expression does not need to be irredundant for this: even if some edges incident with the label $i$ and created by the join-operation are already present in the graph, the vertices of label $i$ will still share a neighbor after this join and therefore, they cannot use all colors.
With this observation, we can adapt the algorithm by Hegerfeld and Kratsch.
All node types apart from relabel-nodes can be handled the same way while relabel-nodes are analogous to the previous subsection.
Hegerfeld and Kratsch also showed that their algorithm is tight under SETH and since multi-clique-width lower-bounds clique-width, this also applies in our case.

\begin{theorem}\label{app:thm:cvc}
    Let $G$ be a graph given together with a multi-$k$-expression of $G$. Then     \textsc{Connected Vertex Cover} problem can be solved in time $\ostar(6^k)$.
    The algorithm is randomized, it does not return false positives and returns false negatives with probability at most $1/2$.
    Unless SETH fails, this problem cannot be solved in time $\ostar((6 - \varepsilon)^k)$ for any $\varepsilon > 0$.
\end{theorem}

\input{CDS-mcw}

%% file: CDS-mcw.tex
\subsection{ \textsc{Connected Dominating Set}} \label{app:subsec:cds}
Another problem for which Hegerfeld and Kratsch provided an algorithm is \textsc{Connected Dominating Set}~\cite{HegerfeldK23}.
In the \textsc{Connected Vertex Cover} problem, given a graph $G = (V, E)$ we are asked about the cardinality of the smallest set $S \subseteq V$ such that $G[S]$ is connected and $N_G[S] = V$.
As for some other parameterizations (e.g., \cite{HegerfeldK20,CyganNPPRW22}), the algorithm is based on a combination of Cut\&Count and the inclusion-exclusion approach.
However, the idea behind the reduction of the number of states is different to \textsc{Connected Vertex Cover} or \textsc{$q$-Coloring}: for this problem, there is usually (e.g., \cite{HegerfeldK20,HegerfeldK23}) a state, called \emph{Allowed}, allowing edges to any other states.
So unlike the previous algorithms, if some label $i$ contains all states and it is joined with a label $j$ containing only \emph{Allowed} vertices, no conflict occurs.
Instead, they unify multiple combinations of states into the same state and then show that the precise state combination does not matter if one is interested in counting the solutions modulo 2.
We will rely on the main idea behind this algorithm, however there are multiple minor changes needed to be carried out so we provide our whole algorithm for multi-clique here. 
We will also try to emphasize the parts that are different from the original algorithm by Kratsch and Hegerfeld and argue why they are needed.

Now we define a consistent cut of a graph to state the Cut\&Count result for \textsc{Connected Dominating Set} and use it as a black-box later.
Let $G$ be a graph, let $v^*$ be an arbitrary but fixed vertex of $G$, and let $\omega: V(G) \to [2|V(G)|]$ be a weight function.
We say that $(L, R)$ is a consistent cut in $G$ if the following properties hold: $v^* \in L$, the sets $L$ and $R$ are disjoint, and there are no edges between $L$ and $R$ in $G$.
For $c \in \NN_0$ and $w \in \NN_0$ with $|L \cup R| = c$ and $\omega(L \cup R) = w$ we say that $(L, R)$ has \emph{weight} $w$ and \emph{cardinality} $c$. 
By $\ccut^{c, w}_G$ we denote the family of all consistent cuts of cardinality $c$ and weight $w$ in $G$.
We also denote 
\[
  \CC^{c, w} = \{(L, R) \in \ccut^{c, w}_G \mid L \cup R \text{ is a dominating set of } G\}.
\]

The Cut\&Count result for \textsc{Connected Dominating Set} by Cygan et al.~\cite{CyganNPPRW22} can be stated as follows:
\begin{theorem}[\cite{CyganNPPRW22}] \label{thm:cds-cnc}
    Let $G = (V, E)$ be a graph,~$v^* \in V$ a fixed vertex, $\omega\colon V \to \bigl[2|V|\bigr]$ a weight function, and $c \in [n]_0, w \in \bigl[2|V|^2\bigr]_0$ integers. 
    If there exists an algorithm that given the above input computes the size of $\CC^{c, w}$ modulo 2 in time $\ostar(T)$ for some function computable non-decreasing function $T$, then there exists a randomized algorithm that given a graph $G$ solves the \textsc{Connected Dominating Set} problem in time $\ostar(T)$. The algorithm cannot give false positives and may give false negatives with probability at most $1/2$.	
\end{theorem}

So from now on we concentrate on the computation of $|\CC^{c, w}| \mod 2$ given a multi-$k$-expression $\phi$ of $G$.
First, we may assume that $\phi$ satisfies \cref{lem:mcw-special-relabels}.
To simplify the algorithm, at the top of $\phi$ we insert $k$ relabel-nodes: a node $\rho_{i \to \emptyset}$ for each $i \in [k]$.
Clearly, this does not change the underlying graph so with a slight abuse of notation, we still denote the arising expression by $\phi$.

As Hegerfeld and Kratsch, we define the following families of partial solutions.
We say that $(A, B, C)$ is a subpartition of $G$ if $A \cup B \cup C \subseteq V(G)$ and $A, B, C$ are pairwise disjoint.
For a node $t$ of $\phi$ and $\cc, \ww \in \NN$, by $\BB^{\cc, \ww}_t$ we denote the family of all subpartitions $(L, R, F)$ of $G_t$ such that
$(L, R) \in \ccut^{\cc, \ww}_{G_t}$ and there is no edge between $L \cup R$ and $F$.
We call such subpartitions \emph{partial solutions} of size $\overline c$ and weight $\overline w$.
Let us emphasize the key difference to the partial solutions by Hegerfeld and Kratsch: they distinguish between live and dead labels and additionally require that $L \cup R$ dominates all vertices of dead labels.
We will take care of domination in the very end.
Now we introduce the \emph{signatures} of partial solutions as defined by Hegerfeld and Kratsch.
Unlike their work, our signatures are over all labels instead of live labels.
First, for a subpartition $(L, R, F)$ of $V_t$ and a label $i \in [k]$, we define $S^i_t(L, R, F) \subseteq \{\stateL, \stateR, \stateF\}$ so that
\begin{itemize}
    \item $\stateL \in S^i_t(L, R, F)$ iff $L \cap U^i_t \neq \emptyset$, 
    \item $\stateR \in S^i_t(L, R, F)$ iff $R \cap U^i_t \neq \emptyset$,
    \item $\stateF \in S^i_t(L, R, F)$ iff $F \cap U^i_t \neq \emptyset$.
\end{itemize}
As already mentioned before, Hegerfeld and Kratsch unify the subsets of $\{\stateL, \stateR, \stateF\}$ that contain at least two elements into a state $\statetwo$ so the set of used states is defined as $\states = \{\{\stateL\}, \{\stateR\}, \{\stateF\}, \statetwo, \emptyset\}$.
A \emph{signature} is a mapping $f \colon [k] \to \states$.
We say that a subpartition $(L, R, F)$ of $G_t$ is compatible with $f$ if the following holds for every $i \in [k]$:
\begin{itemize}
    \item If $|S^i_t(L, R, F)| < 2$, then $f(i) = S^i_t(L, R, F)$.
    \item If $|S^i_t(L, R, F)| \geq 2$, then $f(i) = \statetwo$.
\end{itemize}
Observe that there exists exactly one signature with which $(L, R, F)$ is compatible.
For a signature $f$, we define 
\[
    \BB^{\cc, \ww}_t(f) = \{(L, R, F) \in \BB^{\cc, \ww}_t \mid (L, R, F) \text{ is compatible with } f\}
\]
and $B^{\cc, \ww}_t(f) = |\BB^{\cc, \ww}_t(f)| \mod 2$.

So our goal for now is to traverse a multi-$k$-expression bottom-up to compute the values $B^{\cc, \ww}_t(f)$. 
In the end, we will summarize how to obtain the size of $\CC^{c, w}$ modulo 2 from it in order to apply \cref{thm:cds-cnc}.
In the following, we assume that the values $\cc$ and $\ww$ are \emph{reasonable}, namely $0 \leq \cc \leq |V(G)|$ and $0 \leq \ww \leq 2|V(G)|^2$.
For values outside these ranges, we implicitly treat any $B^{\cc, \ww}_t(f)$ as zero.
In the following, by $\equiv$ we denote the equality in $\FF_2$.

First, let $t$ be a $1 \langle i \rangle$-node for some $i \in [k]$ introducing a vertex $v$.
Then it holds that
\begin{align*}
    B^{\cc, \ww}_t(f) = &[v \neq v^* \lor f(i) \in \{\stateL\}] \cdot \\
    &\Bigl[\bigl(\cc = \ww = 0 \land f(i) \in \{\emptyset, \{\stateF\}\}\bigr) 
    \lor \\
    &\bigl(\cc = 1 \land \ww = \omega(v) \land f(i) \in \{ \{\stateL\}, \{\stateR\}\}\bigr)
    \Bigr] \cdot \\
    & \bigl[\forall j \in [k] \setminus \{i\} \colon f(j) = \emptyset\bigr] 
\end{align*}

Next, let $t$ be a $\eta_{i, j}$-node for some $i \neq j \in [k]$ and let $t'$ be its child.
Here, we can adopt the approach of Kratsch and Hegerfeld without changes, namely:
\[
    B^{\cc, \ww}_t(f) = \operatorname{feas}(f(i), f(j)) \cdot B^{\cc, \ww}_{t'}
\]
where $\operatorname{feas} \colon \states \times \states \to \{0, 1\}$ is given by
\begin{center}
\begin{tabular}{ c|ccccc } 
$\operatorname{feas}$ & $\emptyset$ & $\{\stateF\}$ & $\{\stateL\}$ & $\{\stateR\}$ & $\statetwo$ \\
\hline
 $\emptyset$ & 1 & 1 & 1 & 1 & 1 \\
 $\{\stateF\}$ & 1 & 1 & 0 & 0 & 0 \\
 $\{\stateL\}$ & 1 & 0 & 1 & 0 & 0 \\
 $\{\stateR\}$ & 1 & 0 & 0 & 1 & 0 \\
 $\statetwo$ & 1 & 0 & 0 & 0 & 0 \\
\end{tabular}
.
\end{center}
In simple words, $\operatorname{feas}$ invalidates partial solutions with $LR$-, $LF$-, or $RF$-edges between $i$ and $j$.

Now let $t$ be a $\rho_{i \to \emptyset}$ for some $i \in [k]$ and let $t'$ be its child.
This operation does not change the set of partial solutions of the arising graph, only their signatures by making the set $U^i_t$ empty. 
So it holds that
\[
    B^{\cc, \ww}_t(f) \equiv [f(i) = \emptyset] \cdot \sum\limits_{s \in \states} B^{\cc, \ww}_{t'}(f[i \to s]).
\]

Similarly, let $t$ be a $\rho_{i \to \{i, j\}}$-node for $i \neq j \in [k]$ and let $t'$ be its child.
Again, the set of partial solutions remains the same but signatures change: a signature at label $j$ is now the union of old signatures at labels $i$ and $j$.
Therefore, we have
\[
    B^{\cc, \ww}_t(f) \equiv \sum\limits_{\substack{s \in \states \colon \\ 
    \operatorname{merge}(s, f(i)) = f(j)}} B^{\cc, \ww}(f[j \to s])
\]
where $\operatorname{merge} \colon \states \times \states$ is defined by Hegerfeld and Kratsch as
\begin{center}
\begin{tabular}{ c|ccccc } 
$\operatorname{merge}$ & $\emptyset$ & $\{\stateF\}$ & $\{\stateL\}$ & $\{\stateR\}$ & $\statetwo$ \\
\hline
$\emptyset$ & $\emptyset$ & $\{\stateF\}$ & $\{\stateL\}$ & $\{\stateR\}$ & $\statetwo$ \\
$\{\stateF\}$ & $\{\stateF\}$ & $\{\stateF\}$ & $\statetwo$ & $\statetwo$ & $\statetwo$ \\
$\{\stateL\}$ & $\{\stateL\}$ & $\statetwo$ & $\{\stateL\}$ & $\statetwo$ & $\statetwo$ \\
$\{\stateR\}$ & $\{\stateR\}$ & $\statetwo$ & $\statetwo$ & $\{\stateR\}$ & $\statetwo$ \\
$\statetwo$ & $\statetwo$ & $\statetwo$ & $\statetwo$ & $\statetwo$ & $\statetwo$ \\
\end{tabular}
.
\end{center}
It is easy to see that for the previous node types, all values $B^{\cc, \ww}_t(f)$ for reasonable $\cc$ and $\ww$ can be computed in $\ostar(5^k)$.

Finally, let $t$ be a union-node and let $t_1$ and $t_2$ be its children.
Then, informally speaking, every partial solution of $G_t$ corresponds to a pair of partial solutions of $G_{t_1}$ and $G_{t_2}$ by forming their pointwise union where a signature at some label is the union of signatures at this label in both partial solutions.
Formally, we have
\[
    B^{\cc, \ww}_t(f) \equiv \sum\limits_{\substack{\cc_1 + \cc_2 = \cc \\ \ww_1 + \ww_2 = \ww}} \sum_{\substack{f_1, f_2 \colon [k] \to \states \\ \operatorname{merge}(f_1, f_2) = f}} B^{\cc_1, \ww_1}_{t_1}(f_1) \cdot B^{\cc_2, \ww_2}_{t_2}(f_2)
\]
where $\operatorname{merge}$ of two functions is componentwise, i.e., $\operatorname{merge}(f_1, f_2)(i) = \operatorname{merge}(f_1(i), f_2(i))$ for all $i \in [k]$.
This equality is analogous to Hegerfeld and Kratsch with the only difference that in our case, the signatures keep track of all labels and not only live ones.
Similarly to their work, we may observe that there is only a polynomial number of reasonable tuples $(\cc_1, \cc_2, \ww_1, \ww_2)$. 
Then we may iterate over all of them in polynomial time.
Now we may assume that such a tuple is fixed and we aim to compute
\[
    \sum_{\substack{f_1, f_2 \colon [k] \to \states \\ \operatorname{merge}(f_1, f_2) = f}} B^{\cc_1, \ww_1}_{t_1}(f_1) \cdot B^{\cc_2, \ww_2}_{t_2}(f_2) \mod 2.
\]
By Lemma 4.6 in \cite{HegerfeldK23}, this can be accomplished in time $\ostar(5^k)$.

These equalities provide a way to compute the values $B^{c, w}_r(f)$ where $r$ denotes the root of $\phi$ for all signatures $f$. 
Any node is processed in time $\ostar(5^k)$ and since we may assume that the number of nodes in $\phi$ is polynomial in $k$ and $|V(G)|$, the values can be computed in time $\ostar(5^k)$.
At the beginning, we already mentioned that the transformation from the numbers $B^{c, w}_r(f)$ to the value $|\CC^{c, w}| \mod 2$ has to be carried out differently for multi-clique-width.
For clique-width, Hegerfeld and Kratsch do it \emph{labelwise}: namely, they ensure that the vertices of a label are dominated once the label is not live anymore. 
This ensures that every vertex is processed exactly once like this.
There are two issues related to this in our case.
First, for this they rely on the existence of irredundant clique-expressions and we do not know whether this is true for multi-expressions.
Second, transforming a vertex every time one of its labels is not active anymore would potentially lead to transforming this vertex multiple times resulting in an uncontrolled behaviour.
To overcome these issues, we will carry out such a transformation at the very end.
Let us note that although the procedure is different from the original work of Hegerfeld and Kratsch, the idea behind it remains the same.

Recall that at the top of $\phi$, we have a $\rho_{i \to \emptyset}$-node for every $i \in [k]$.
Thus, we have $U^i_r = \emptyset$ for all $i \in [k]$ and every partial solution of $G_r = G$ has the signature $f_\emptyset \colon [k] \to \states$ where $f_\emptyset(i) = \emptyset$ for every $i \in [k]$.
So we have
\[
    |\BB^{c, w}_r| \equiv B^{c, w}_r(f_\emptyset).
\]
Now we claim that $|\CC^{c, w}| \equiv |\BB^{c, w}_r|$ holds.
For the simplicity of notation, let
\[
    \DD^{c, w} = \{ (L, R, \emptyset) \mid (L, R) \in \CC^{c, w}\}.
\]
Clearly, it holds that $|\DD^{c, w}| = |\CC^{c, w}|$ so it suffices to show that $|\DD^{c, w}| \equiv |\BB^{c, w}_r|$.
First, observe that $\DD^{c, w} \subseteq \BB^{c, w}_r$ holds: indeed, for every element $(L, R, \emptyset)$ of $\DD^{c, w}$, the pair $(L, R)$ is a consistent cut of $G_r = G$ and the cardinality resp.\ weight of $L \cup R$ is $c$ resp. $w$.
Next we show that the cardinality of $\BB^{c, w}_r \setminus \DD^{c, w}$ is even.
For this, consider an arbitrary but fixed pair $(L, R)$ such that there exists $F$ with 
\[
    (L, R, F) \in \BB^{c, w}_r \setminus \DD^{c, w}.
\]
First, we claim that $L \cup R$ is not a dominating set of $G = G_r$.
If $F = \emptyset$, then by definition of these sets, the only reason for $(L, R, F)$ to belong to $\BB^{c, w}_r$ but not $\DD^{c, w}$ is that $L \cup R$ is not a dominating set of $G$;
On the other hand, if there is a vertex $v \in F$, then there is no edge between $L \cup R$ and $v \in F$ so $v$ is undominated by $L \cup R$.
Let $U = V(G) \setminus N_G[L \cup R]$.
By our claim, the set $U$ is non-empty.
Then the sets $F$ satisfying $(L, R, F) \in \BB^{c, w}_r \setminus \DD^{c, w}$ are exactly the subsets of $U$ since there is no edge between $L \cup R$ and $F$.
Therefore there are exactly $2^{|U|}$ such sets $F$.
Recall that $U$ is non-empty so $2^{|U|}$ is even.
Altogether, for every fixed pair $(L, R)$ there exist either no or an even number of sets $F$ with $(L, R, F) \in \BB^{c, w}_r \setminus \DD^{c, w}$.
So the size of $\BB^{c, w}_r \setminus \DD^{c, w}$ is indeed even.
Altogether we obtain that
\[
    |\CC^{c, w}| = |\DD^{c, w}| \equiv |\BB^{c, w}_r|.
\]
Thus, the above algorithm computes the size of $\CC^{c, w}$ in time $\ostar(5^k)$ and by \cref{thm:cds-cnc}, the \textsc{Connected Dominating Set} problem can also be solved in time $\ostar(5^k)$.
Hegerfeld and Kratsch also showed that their algorithm is tight under SETH and since multi-clique-width lower-bounds clique-width, this also applies in our case.

\begin{theorem}\label{app:thm:cds}
    Let $G$ be a graph given together with a multi-$k$-expression of $G$. Then     \textsc{Connected Dominating Set} problem can be solved in time $\ostar(5^k)$.
    The algorithm is randomized, it does not return false positives and returns false negatives with probability at most $1/2$.
    Unless SETH fails, this problem cannot be solved in time $\ostar((5 - \varepsilon)^k)$ for any~$\varepsilon > 0$.
\end{theorem}

%% file: Conclusion.tex
\section{ Conclusion}\label{sec:conclusion}
In this work, we studied two generalizations of clique-width, namely fusion-width and multi-clique-width, both introduced by Fürer~\cite{Furer14,Furer17}.
First, we showed that the fusion-width of a graph is an upper bound for its multi-clique-width. 
For the other direction, the best upper bound we are aware of is $\fw \leq 2^{\mcw}$ and we leave open whether this is tight.
By extending existing algorithms for clique-width, we have
obtained tight algorithms parameterized by multi-clique-width for
\textsc{Dominating Set}, \textsc{Chromatic Number},
\textsc{$q$-Coloring}, \textsc{Connected Vertex Cover}, and \textsc{Connected Dominating Set}.
The running times are the same as for (S)ETH-optimal algorithms parameterized by clique-width.

For \textsc{Hamiltonian Cycle}, \textsc{MaxCut}, and \textsc{Edge Dominating Set}, we were not able to achieve analogous results and these complexities remain open.
Instead, we have introduced glue-expressions equivalent to fuse-expressions and then we employed them for these three problems to obtain tight algorithms parameterized by fusion-width with the same running times as ETH-optimal algorithms for clique-width.

Finally, in all algorithms we assume that a multi-$k$-expression / fuse-$k$-expression is provided.
However, the complexity of computing these parameters is unknown.
To the best of our knowledge, the best approximation would proceed via clique-width, have FPT running time, and a double-exponential approximation ratio.